\documentclass[journal]{IEEEtran}

\usepackage[nospace,noadjust]{cite}
\usepackage{amsmath,amssymb,amsfonts,times}
\usepackage{graphicx}
\usepackage{textcomp}
\usepackage{xcolor}
\usepackage{subfigure}
\usepackage{epsfig}
\usepackage{multirow}
\usepackage{algorithm}
\usepackage{algorithmicx}
\usepackage{algpseudocode}
\usepackage{array}
\usepackage{epstopdf}
\usepackage{color}
\usepackage{diagbox}
\usepackage{bm}
\usepackage{bbm}
\usepackage{mathrsfs}
\usepackage{tcolorbox}
\usepackage{pdfsync}

\usepackage{caption}
\usepackage{lipsum}

\makeatletter
\def\@IEEEsectpunct{\ \,}
\def\paragraph{\@startsection{paragraph}{4}{\z@}{1.5ex plus 1.5ex minus 0.5ex}%
{0ex}{\normalfont\normalsize\sffamily\bfseries}}
\makeatother

\usepackage{url}
\usepackage[hidelinks]{hyperref}
\usepackage{amsmath,amsthm,amssymb,amsbsy}
\usepackage{paralist}
\usepackage{xcolor}
\usepackage{color}
\usepackage{comment}
\usepackage{multirow}

\usepackage{fancyhdr}
\usepackage{cite}
\usepackage{cleveref}

\newtheorem{lemma}{Lemma}
\newtheorem{defi}{Definition}

\newtheorem{theorem}{Theorem}

\theoremstyle{remark}

\theoremstyle{problem}

\newcommand{\R}{\mathbb{R}}
\def \real    { \mathbb{R} }

\newcommand{\C}{\mathbb{C}}

\newcommand{\e}{\begin{equation}}
\newcommand{\ee}{\end{equation}}
\newcommand{\en}{\begin{equation*}}
\newcommand{\een}{\end{equation*}}
\newcommand{\eqn}{\begin{eqnarray}}
\newcommand{\eeqn}{\end{eqnarray}}
\newcommand{\bmat}{\begin{bmatrix}}
\newcommand{\emat}{\end{bmatrix}}

\DeclareMathAlphabet\mathbfcal{OMS}{cmsy}{b}{n}
\renewcommand{\P}[1]{\operatorname{\mathbb{P}}\left(#1\right)}
\newcommand{\E}{\operatorname{\mathbb{E}}}

\newcommand{\vct}[1]{\boldsymbol{#1}}
\newcommand{\mtx}[1]{\boldsymbol{#1}}

\newcommand{\<}{\langle}
\renewcommand{\>}{\rangle}

\newcommand{\trace}{\operatorname{trace}}

\newcommand{\rank}{\operatorname{rank}}

\newcommand{\dist}{\operatorname{dist}}

\newcommand{\set}[1]{\mathbb{#1}}

\DeclareMathOperator*{\argmin}{\text{arg~min}}
\DeclareMathOperator*{\argmax}{\text{arg~max}}

\newcommand{\wh}{\widehat}

\newcommand{\wt}{\widetilde}
\newcommand{\ol}{\overline}
\newcommand{\nqbit}{n}

\newcommand{\norm}[2]{\left\| #1 \right\|_{#2}}

\newcommand{\abs}[1]{\left| #1 \right|}
\newcommand{\bracket}[1]{\left( #1 \right)}
\newcommand{\parans}[1]{\left(#1\right)}
\newcommand{\innerprod}[2]{\left\langle #1,  #2 \right\rangle}

\newcommand{\calA}{\mathcal{A}}
\newcommand{\calB}{\mathcal{B}}
\newcommand{\calC}{\mathcal{C}}

\newcommand{\calI}{\mathcal{I}}

\newcommand{\calN}{\mathcal{N}}

\newcommand{\calP}{\mathcal{P}}

\newcommand{\calX}{\mathcal{X}}

\newcommand{\va}{\vct{a}}
\newcommand{\vb}{\vct{b}}

\newcommand{\ve}{\vct{e}}

\newcommand{\vh}{\vct{h}}

\newcommand{\vp}{\vct{p}}

\newcommand{\vr}{\vct{r}}

\newcommand{\vu}{\vct{u}}
\newcommand{\vv}{\vct{v}}

\newcommand{\vz}{\vct{z}}

\newcommand{\vphi}{\vct{\phi}}
\newcommand{\vpsi}{\vct{\psi}}

\newcommand{\veta}{\vct{\eta}}
\newcommand{\vxi}{\vct{\xi}}

\newcommand{\vrho}{\vct{\rho}}
\newcommand{\vzero}{\vct{0}}

\newcommand{\mA}{\mtx{A}}
\newcommand{\mB}{\mtx{B}}
\newcommand{\mC}{\mtx{C}}

\newcommand{\mU}{\mtx{U}}
\newcommand{\mV}{\mtx{V}}

\newcommand{\mX}{\mtx{X}}

\newcommand{\mSigma}{\mtx{\Sigma}}

\newcommand{\mId}{{\bf I}}

\newcommand{\setO}{\set{O}}

\newcommand{\setS}{\set{S}}

\newcommand{\setU}{\set{U}}

\newcommand{\setX}{\set{X}}

\graphicspath{{./figs/}}

\newlength{\imgwidth}
\newboolean{twoColVersion}
\setboolean{twoColVersion}{false}
\newcommand{\twoCol}[2]{\ifthenelse{\boolean{twoColVersion}} {#1} {#2} }

\usepackage{mathtools}
\DeclarePairedDelimiter\ket{\lvert}{\rangle}

\begin{document}

\title{Quantum State Tomography for\\ Matrix Product Density Operators}

\author{Zhen~Qin, Casey~Jameson, Zhexuan~Gong, Michael~B.~Wakin,~\IEEEmembership{Fellow,~IEEE,} and~Zhihui~Zhu,~\IEEEmembership{Member,~IEEE}
\thanks{Zhen Qin, and Zhihui Zhu are with the Department of Computer Science and Engineering, Ohio State University, Columbus, Ohio 43201, USA. (e-mail:\{qin.660,zhu.3440\}@osu.edu).}
\thanks{Casey Jameson and Zhexuan Gong are with the Department of Physics, Colorado School of Mines, Golden, Colorado 80401, USA. (e-mail:\{cwjameson,gong\}@mines.edu).}
\thanks{Michael B. Wakin is with the Department of Electrical Engineering, Colorado School of Mines, Golden, Colorado 80401, USA. (e-mail:mwakin@mines.edu).}}

\maketitle

\begin{abstract}
The reconstruction of quantum states from experimental measurements, often achieved using quantum state tomography (QST), is crucial for the verification and benchmarking of quantum devices. However, performing QST for a generic unstructured quantum state requires an enormous number of state copies that grows \emph{exponentially} with the number of individual quanta in the system, even for the most optimal measurement settings. Fortunately, many physical quantum states, such as states generated by noisy, intermediate-scale quantum computers, are usually structured. In one dimension, such states are expected to be well approximated by matrix product operators (MPOs) with a matrix/bond dimension independent of the number of qubits, therefore enabling efficient state representation. Nevertheless, it is still unclear whether efficient QST can be performed for these states in general. In other words, there exist no rigorous bounds on the number of state copies required for reconstructing MPO states that scales polynomially with the number of qubits.

In this paper, we attempt to bridge this gap and establish theoretical guarantees for the stable recovery of MPOs using tools from compressive sensing and the theory of empirical processes. We begin by studying two types of random measurement settings: Gaussian measurements and Haar random projective measurements. We show that the information contained in an MPO with a constant bond dimension can be preserved using a number of random measurements that depends only \emph{linearly} on the number of qubits, assuming no statistical error of the measurements. We then study MPO-based QST with  Haar random projective measurements that can in principle be implemented on quantum computers. We prove that only a \emph{polynomial} number of state copies in the number of qubits is required to guarantee bounded recovery error of an MPO state. Remarkably, such recovery can be achieved by measuring the state in each random basis only once, despite the large statistical error associated with the outcome of each measurement. Our work may be generalized to accommodate random local or t-design measurements that are more practical to implement on current quantum computers. It may also facilitate the discovery of efficient QST methods for other structured quantum states.
\end{abstract}



\section{Introduction}
\IEEEPARstart{D}{riven} by advances in hardware and experimental techniques, the size of quantum computers has rapidly increased in recent years, with some of the most advanced processors having over 100 qubits \cite{preskill2018quantum,arute2019quantum,chow2021ibm}. As quantum computing and quantum simulation continue to advance, fully characterizing the large quantum many-body states produced by experimental quantum devices has become a significant challenge, as the number of parameters needed to characterize these states scales exponentially in the number of qubits in general. Nevertheless, for verification and benchmarking purposes, it is important to reconstruct such quantum states with an affordable amount of resources and with high accuracy.

The reconstruction of quantum states is typically achieved by a technique known as quantum state tomography (QST) \cite{vogel1989determination}. The goal of QST is to find a density matrix that describes the quantum state under interest with high accuracy.\footnote{See \Cref{sec:quantum-mechanics} for a review on the basic concepts needed to understand QST.} In a quantum system consisting of $n$ qudits (which are $d$-level quantum systems; qubits have $d=2$), the state can be expressed by a density matrix $\vrho$ of size $d^\nqbit\times d^\nqbit$. To find $\vrho$ of an experimental quantum state, in general we need to perform quantum measurements on many identical copies of the state. Any physical measurement on a quantum system is described by a Positive Operator-Valued Measure (POVM), which is a collection of positive semi-definite (PSD) matrices or operators $\{\mA_1,\ldots,\mA_K\}$ that sum to the identity operator. Each operator $\mA_k$ ($k=1,\dots, K$) in the POVM corresponds to a possible measurement outcome, and the probability of obtaining that outcome is given by $p_k = \trace(\mA_k \vrho)$. Thus, this \emph{probabilistic nature} of quantum measurements often requires
the state to be measured many (say $M$)  times with the same POVM to obtain an approximately accurate statistical estimate $\wh p_k$ of each $p_k$. Without considering the statistical error, $\{p_k\}$ can be viewed as $K$ linear measurements of the state~$\vrho$. Thus, adopting terminology from machine learning, we may refer to $\{p_k\}$ and their empirical estimates~$\{\wh p_k\}$ as population and empirical measurements of the state, respectively. From this viewpoint, QST can be viewed as a matrix sensing problem~\cite{recht2010guaranteed,candes2011tight}, but with a specific type of measurement operators, and with measurements that are inherently probabilistic. Furthermore, measurements on a quantum state are usually destructive and therefore we need many identical copies of the state for performing many measurements. Typically, an interesting quantum many-body state can be generated using a quantum computer or quantum simulator in a time scale ranging from microseconds to milliseconds for common hardware platforms. If the number of state copies required by QST scales exponentially in the number of qudits, then we cannot perform QST in practice for even a few tens of qubits.

Many different methods have been proposed for QST, including  maximum likelihood \cite{hradil1997quantum,vrehavcek2001iterative}, Bayesian \cite{blume2010optimal,granade2016practical,lukens2020practical}, region \cite{blume2012robust,faist2016practical}, and least squares \cite{kyrillidis2018provable,brandao2020fast} estimators, as well as machine learning techniques \cite{torlai2018neural,carleo2019machine,lohani2020machine}. For generic quantum states, the number of state copies needed for QST always grows exponentially with the number of qudits. A significant amount of work has been dedicated, however, to optimal QST methods for states represented by low-rank density matrices, which are physically common \cite{KuengACHA17,guctua2020fast,francca2021fast,voroninski2013quantum, haah2017sample}. Various measurement settings have been adopted in this context, including 4-design \cite{KuengACHA17}, Pauli \cite{liu2011universal,guctua2020fast}, Clifford \cite{francca2021fast}, Haar-random unitary \cite{voroninski2013quantum}, etc. It has been shown that as long as the measurements are performed on one state at a time, a minimum number of total state copies proportional to $d^n r^2/\epsilon^2$ is required to estimate a rank-$r$ density matrix with accuracy given by $\epsilon$ in the trace norm between the reconstructed density matrix and the true density matrix \cite{haah2017sample,francca2021fast}. This means that even for a rank-one density matrix (corresponding to a pure quantum state that can only be created by a noiseless quantum device), the number of state copies required for QST still scales as $2^n$ for $n$ qubits.

To achieve QST for current quantum computers at the scale of $\sim$100 qubits, the number of required state copies should scale only polynomially with the number of qubits $n$. This is possible only if the target state itself is structured in a way such that it has a compact representation with $\text{poly}(n)$ independent parameters. Fortunately, many physical quantum states indeed have such structure. Examples include ground states of most quantum systems with short-range interactions and states generated by such quantum systems in a finite amount of time \cite{eisert2010}. These states usually do not contain a large amount of quantum entanglement such that a compact representation via a matrix product state (MPS) or tensor network is often possible \cite{eisert2010}. A similar intuition applies to states generated by noisy quantum computers, where the noise could also limit the amount of quantum entanglement and thus enable an efficient state representation. In particular, it is widely believed that most states generated by a one-dimensional noisy quantum computer are well approximated by matrix product operators (MPOs) with a finite matrix dimension \cite{noh2020efficient}. Therefore, it becomes practically important to find efficient QST methods for states with an efficient MPO representation.

An MPO consists of $nd^2$ matrices each with dimension at most $\ol r\times \ol r$. The matrix dimension $\ol r$ is more often called the bond dimension, or the rank of the MPO (see \Cref{sec:MPO} for a detailed description of MPO). MPOs that represent physical quantum states are also called matrix product density operators (MPDO). An MPO is also mathematically equivalent to a tensor train (TT) used for compact representation of large tensors \cite{Oseledets11}. Assuming the bond dimension $\ol r$ is a constant, the MPO contains a number of parameters that scales only linearly with the number of qudits, and is thus a very efficient representation. Nevertheless,
such an efficient representation does not guarantee that the number of state copies required for QST is also small. In fact, for a general MPO state with bond dimension $\ol r$, there exists no known QST method that guarantees a required number of state copies that scales polynomially with the number of qubits \cite{baumgratz2013scalable,lidiak2022quantum}. This is in contrast to an MPS state (a pure state with a compact representation using $nd$ matrices), where such a guarantee exists for almost all physical MPS states \cite{cramer2010efficient,lanyon2017efficient,wang2020scalable,verstraete2004matrix,pirvu2010matrix,werner2016positive,jarkovsky2020efficient}.
Therefore, we ask the following main question:
\smallskip
\begin{tcolorbox}[colback=white,left=1mm,top=1mm,bottom=1mm,right=1mm,boxrule=.3pt]
\centering {\bf Question}: Given a structured $n$-qudit quantum state represented by a constant bond dimension MPO, is it possible to reconstruct the state with guaranteed accuracy using only $\text{poly}(n)$ state copies?
\end{tcolorbox}

\subsection{Main results}

In this paper, we show that the answer to the above main question is yes, assuming that we can perform measurements of the given quantum state in Haar random bases. We note that this affirmative answer does not imply efficient QST for general MPO states since an exponentially large number (in $n$) of local quantum gates may be required to achieve such Haar random basis measurements with high accuracy. Nevertheless, our results paves the way to fully efficient QST methods as one may be able to reduce such number of required local quantum gates to polynomial in $n$ via unitary t-designs \cite{brandao2016local}.

Our particular focus on Haar random bases is motivated by the tremendous success of randomized measurements in compressive sensing for signals exhibiting low-dimensional structure such as sparse, low rank, or manifold structure~\cite{donoho2006compressed,candes2006robust,
candes2008introduction,recht2010guaranteed,eftekhari2015new,tropp2015convex}. The incorporation of randomness often enables nearly optimal upper bounds to be established for the sufficient number of measurements to recover structured signals. Moreover, randomized measurements have been recognized as a powerful tool that can efficiently transform quantum systems into classical representations, capturing numerous features of the original quantum state \cite{francca2021fast,huang2020predicting,yu2021experimental}; see \cite{elben2023randomized} for a review on this topic.

{\it The first main contribution of this paper---presented in \Cref{sec:Stable embedding MPO}---is that we investigate the number of population measurements (without statistical errors) to guarantee a stable embedding of MPOs.} In particular, we first establish the restricted isometry property (RIP, see {Definition} \ref{def:RIP}) for complex Gaussian measurements where each matrix element of $\mA_k$ is an independent and identically distributed (\emph{i.i.d.})\ standard complex Gaussian random variable for all $k = 1,\ldots, K$. Although these measurement operators are not PSD and may not be implementable in practical quantum experiments, this analysis sheds light on the optimal number of population measurements to ensure unique recovery of the MPO. We then study rank-one Gaussian measurement ensembles $\{\mA_k\}$ taking the form $\mA_k = \va_k \va_k^\dagger$ where $\va_k$ is randomly generated from a multivariate Gaussian distribution. As such rank-one measurements do not obey the RIP condition \cite{zhong2015efficient}, we instead establish a weaker version of an embedding guarantee. In order to do this, we use Mendelson's small ball method \cite{mendelson2015learning, koltchinskii2015bounding,tropp2015convex}, which has previously been used to establish stable embeddings for low-rank matrices under rank-one measurements~\cite{KuengACHA17}. For both generic Gaussian measurement ensembles and rank-one Gaussian measurement ensembles, we show that $\wt\Omega(nd^2\ol r^2)$ total linear measurements\footnote{The notation $\wt\Omega(\cdot)$ is defined in Section~\ref{sec:notation}.} are sufficient to achieve stable embeddings of MPOs with high probability. This result is nearly optimal as the MPO contains $nd^2\ol r^2$ independent parameters.

We then extend the results to Haar random projective measurements, where the measurement operators of each measurement is a collection of PSD matrices $\{\vphi_k\vphi_k^\dagger\}, k = 1,\ldots,d^n$ with $\mU = \begin{bmatrix}\vphi_1 & \cdots & \vphi_{d^n} \end{bmatrix}$ being a Haar-distributed random unitary matrix. As will be formally illustrated in \Cref{sec:quantum-mechanics}, such a measurement scheme is equivalent to first rotating the state with the unitary matrix $\mU$ and then performing measurements in the standard computational basis, which can be implemented (albeit not efficiently) on universal quantum computers~\cite{francca2021fast}. We establish similar stable embedding results for $Q = \wt \Omega(nd^2 \ol r^2)$ such Haar random bases, assuming zero statistical error.

Second, we study the recovery of an MPO from empirical quantum measurements (physical measurements containing statistical errors) and establish recovery bounds with respect to the number of state copies, using the above-mentioned Haar random projective measurements. {\it The second main contribution of this paper---presented in \Cref{sec:stable recovery}---is that we establish theoretical bounds on the accuracy of a particular estimator---the solution to a constrained least-squares optimization problem---for recovering an MPO.}
We summarize the results informally as follows.
\begin{theorem}[informal version of \Cref{Statistical Error_Haar_Measurement}]
\label{Statistical Error_Haar_Measurement_informal}
Given an $n$-qudit MPO state with bond dimension $\ol r$, randomly generate $Q$ Haar random projective measurement bases and measure the state in each basis $M$ times. For any $\epsilon>0$, assume $Q= \wt\Omega(nd^2\ol{r}^2)$ and the number of total state copies $QM=\wt\Omega(n^3 d^2\ol{r}^2/\epsilon^2)$. Then, with high probability, a properly constrained least-squares minimization with the empirical measurements stably recovers the ground-truth state with $\epsilon$-closeness in the Frobenius norm.
\label{thm:informal-main}\end{theorem}
Our result ensures a stable recovery of the ground-truth state with a total number of state copies $QM$ growing only polynomially in the number of qudits $n$. Compared to the requirement of $\Omega(d^n)$ state copies for estimating a general low-rank density matrix \cite{haah2017sample}, utilizing the MPO structure can significantly reduce the number of state copies (from $d^n$ to $n^3$). In addition, there is no other requirement on the number of state copies $M$ for each measurement basis. In other words, our result also provides theoretical support for the practical use of
single-shot measurements (setting $M = 1$, i.e., measuring the state in each basis once) that have been practically adopted in \cite{wang2020scalable,huang2020predicting}. On the other hand, our recovery guarantee builds upon the stable embedding results and is established in the Frobenius norm instead of the trace norm (i.e., nuclear norm). 
However, if the density matrix represented by the MPO has a low matrix rank, we can also establish recovery guarantee in the trace norm by using a strong bound between trace distance and Hilbert-Schmidt distance for low-rank states \cite{coles2019strong}. While this simple approach provides a vacuous bound when the state has high rank, we conjecture the result in \Cref{thm:informal-main} can be extended for the trace norm, but we leave this as future work. We provide a detailed discussion right after \Cref{Statistical Error_Haar_Measurement}.

We note that obtaining the constrained least squares estimate requires solving a nonconvex problem. To tackle this problem, we employ iterative hard thresholding (i.e., projected gradient descent) \cite{Rauhut17} and showcase its efficacy through numerical experiments. We do not provide a formal guarantee for the algorithm and leave its analysis for future work. 

\subsection{Related work involving tensor train decompositions}

Having mentioned that the MPO model is equivalent to a tensor train (TT) decomposition, we discuss some related work on sampling and recovery of tensors. The work \cite{Rauhut17} established the first RIP bound for structured tensors (including the TT format) with real generic subgaussian measurements. Our proof of the RIP for complex Gaussian measurements uses the same technique as \cite{Rauhut17}; see the discussion following \Cref{thm:SubgaussianRIP-appendix} for more information.
The work \cite{cai2022provable} studied the tensor completion problem with random samples of a TT format tensor, but the result requires an exponentially large number of samples. Another line of work \cite{OseledetsLAA10,savostyanov2014quasioptimality,osinsky2019tensor} extended matrix \emph{cross approximation} techniques \cite{goreinov2001maximal,hamm2020perspectives,cai2021robust} for computing a TT format from selected subtensors. The work \cite{qin2022error} has provided accuracy guarantees in terms of the entire tensor for TT cross approximation, and the work \cite{lidiak2022quantum} applied TT cross approximation for reconstructing MPOs by only measuring local operators. Numerical simulation results demonstrate the effectiveness of this technique, but no explicit theoretical bound on the number of state copies is provided \cite{lidiak2022quantum}. While the algorithm is not the focus of this work, we note that there are many proposed algorithms for estimating TT format tensors from linear measurements \cite{bengua2017efficient,imaizumi2017tensor, wang2016tensor, Rauhut17, rauhut2015tensor,budzinskiy2021tensor,wang2019tensor,cai2022provable,qin2024guaranteed}. These include algorithms based on convex relaxation \cite{bengua2017efficient,imaizumi2017tensor}, alternating minimization \cite{wang2016tensor}, projected gradient descent (also known as iterative hard thresholding (IHT))  \cite{Rauhut17}, and Riemannian methods \cite{budzinskiy2021tensor,wang2019tensor,cai2022provable,qin2024guaranteed}.

\subsection{Notation}
\label{sec:notation}

We use calligraphic letters (e.g., $\mathcal{X}$) to denote tensors,  bold capital letters (e.g., $\bm{X}$) to denote matrices,  bold lowercase letters (e.g., $\bm{x}$) to denote column vectors, and italic letters (e.g., $x$) to denote scalar quantities. Elements of matrices and tensors are denoted in parentheses. For example, $\calX(i_1, i_2, i_3)$ denotes the element in position
$(i_1, i_2, i_3)$ of the order-3 tensor $\calX$. The calligraphic letter $\calA$ is reserved for the linear measurement map.  For a positive integer $K$, $[K]$ denotes the set $\{1,\dots, K \}$. The superscripts $(\cdot)^\top$ and $(\cdot)^\dagger$ denote the transpose and Hermitian transpose, respectively\footnote{As is conventional in the quantum physics literature (but not in information theory and signal processing), we use $(\cdot)^\dagger$ to denote the Hermitian transpose.}. For two matrices $\mA,\mB$ of the same size, $\innerprod{\mA}{\mB} = \trace(\mA^\dagger\mB)$ denotes the inner product between them.
$\|\mA\|$ (or $\|\mA\|_{2\to 2}$) and $\|\mA\|_F$ respectively represent the spectral norm and Frobenius norm of $\mA$.
For a vector $\va$ of size $N\times 1$, its $l_n$-norm is defined as $||\va||_n=(\sum_{m=1}^{N}|a_m|^n)^\frac{1}{n}$.
For two positive quantities $a,b\in \real$, the inequality $b\lesssim a$ or $b = O(a)$ means $b\leq c a$ for some universal constant $c$; likewise, $b\gtrsim a$ or $b = \Omega(a)$ represents $b\ge ca$ for some universal constant $c$. We define $\widetilde{\Omega}$ as the function obtained by removing the logarithmic factors from~$\Omega$.

\section{Basic concepts for Quantum State Tomography}
\label{sec:quantum-mechanics}

In this section, we review basic concepts needed to understand quantum state tomography, since these concepts may be unfamiliar to researchers in information theory and signal processing. These concepts are commonly used in the field of quantum information \cite{nielsen2002quantum}, and they can be understood with the knowledge of linear algebra and probability theory.

\subsection{States and density operators}
In quantum physics, the state of an isolated quantum system is fully described by a state vector $\ket{\psi}$ (using the Dirac notation), which represents a unit-length vector in a complex vector space known as the Hilbert space. For example, the state of the simplest quantum system, known as a \emph{qubit}, is represented by a vector in a two-dimensional Hilbert space. One can choose two orthonormal basis vectors for this Hilbert space denoted by $|0\rangle$ and $|1\rangle$, which typically represent two distinct physical states of a qubit (for example, the lowest and second-lowest energy states of an atom). An arbitrary state of the qubit can then be written as  $\ket{\psi} = a\ket{0} + b\ket{1}$, where $a$ and $b$ are complex numbers satisfying $|a|^2 + |b|^2 = 1$, which ensures that $\ket{\psi}$ is unit length. The state vector $\ket{\psi}$ can thus be equivalently represented by a $2\times1$ vector
\[
\vpsi := \begin{bmatrix}a\\ b\end{bmatrix} \in \C^2.
\]

A \emph{qudit} is a generalization of the idea of a qubit to a $d$-level system or $d$-dimensional Hilbert space, where each state vector can be equivalently represented by a unit-length vector in $\C^d$. While most quantum computers process information using qubits just as classical computers use bits, we use qudits in this paper for a more general framework, as they are commonly used for quantum simulation as spins and may be used for quantum computing as well.

A quantum computer or simulator usually consists of many qudits. For such quantum many-body systems, which are the focus of this paper, the full state space is the tensor product of the state spaces of each qudit. Specifically, for a composite system of $\nqbit$ qudits, each state vector $\vpsi$ belongs to $\C^{d^\nqbit}$ and has unit length.

Until now we have considered quantum systems whose state can be fully described by a state vector $\vpsi$. Such a quantum system is said to be in a \emph{pure state}.  More broadly, though, a quantum system can be in one of a number of states $\vpsi_i$ with respective probabilities $\alpha_i$. In this case, we say the quantum system is in a \emph{mixed state}, which may be described as $\{\alpha_i, \vpsi_i  \}$ where $0\le \alpha_i\le 1$ are the probabilities with $\sum_i \alpha_i = 1$. A mixed state naturally arises due to the interactions (which create quantum entanglement) between the quantum system and its environment, such that the state of the system becomes indeterminate on its own.

A quantum system in a mixed state is described by a \emph{density operator} or \emph{density matrix}.\footnote{Formally speaking, a density matrix is a representation of a density operator in a given choice of basis in the underlying Hilbert space. In this paper, we always choose the standard computational basis for the qudits denoted by $\{|0\rangle, |1\rangle, \cdots, |d-1\rangle\}$. Therefore, we use the two terms density matrix and density operator interchangeably.}  The density operator of a pure state $\vpsi\in \C^{d^\nqbit}$ is given by
\[
\vrho = \vpsi \vpsi^\dagger \in \C^{d^\nqbit \times d^\nqbit}.
\]
For a mixed state, the density operator can be written as
\[
\vrho = \sum_i \alpha_i \vpsi_i \vpsi_i^\dagger \in \C^{d^\nqbit \times d^\nqbit}.
\]
Thus, a density operator with rank equal to one corresponds to a pure state; otherwise it corresponds to a mixed state. In all cases, we have that $(i)$~the density operator $\vrho \succeq \vzero$ is a PSD matrix, and $(ii)$~$\trace(\vrho) = 1$.

\subsection{Quantum measurements}
\label{sec:POVM measurements}
Quantum state tomography aims to reconstruct or estimate the density operator $\vrho$ of a quantum system using measurements on multiple copies of the same quantum state. The most general measurements one can physically perform on a quantum system are described by Positive Operator Valued Measures (POVMs) \cite{nielsen2002quantum}, as explained below.

\begin{defi} [POVM \cite{nielsen2002quantum}] A Positive Operator Valued Measure (POVM) is a set of PSD matrices $\{\mA_1,\ldots,\mA_K \}$ such that
\begin{eqnarray}
\label{The defi of POVM 1}
\sum_{k=1}^K \mA_k = \mId.
\end{eqnarray}
Each POVM element $\mA_k$ is associated with a possible outcome of a quantum measurement, and the probability $p_k$ of detecting the $k$-th outcome when measuring the density operator $\vrho$ is given by
\begin{eqnarray}
\label{The defi of POVM 2}
p_k = \innerprod{\mA_k}{\vrho},
\end{eqnarray}
where $\sum_{k=1}^Kp_k=1$ due to \eqref{The defi of POVM 1} and the fact that $\trace(\vrho) = 1$. We often repeat the measurement process $M$ times and take the average of the statistically independent outcomes to generate the empirical probabilities
\begin{equation}
\wh p_{k} = \frac{f_k}{M}, \  k \in[K]:=\{1,\ldots,K\},
\label{eq:empirical-prob}\end{equation}
where $f_k$ denotes the number of times the $k$-th outcome is observed. In information theory and signal processing communities, we call $\{p_k\}$ and $\{\wh p_k\}$ the population and empirical (linear) measurements, respectively.
\label{def:POVM}\end{defi}

Collectively, the random variables $f_1,\ldots,f_K$ are characterized by the multinomial distribution $\operatorname{Multinomial}(M, \vp)$ \cite{severini2005elements} with parameters $M$ and $\vp = \begin{bmatrix} p_1 & \cdots & p_K\end{bmatrix}^\top$, where $p_k$ is defined in \eqref{The defi of POVM 2}. It follows that the empirical probability $\wh p_k$ in \eqref{eq:empirical-prob} is an unbiased estimator of the probability $p_k$.
One can bound the estimation error $\abs{\wh p_k - p_k}$ by $O(1/\sqrt{M})$ with high probability via concentration inequalities. For example, the Dvoretzky-Kiefer-Wolfowith (DKW) theorem \cite{dvoretzky1956asymptotic,kosorok2008introduction} ensures that the empirical probability $\wh p_{k}$ is close to $p_k$ for all $k$ simultaneously when $M$ is sufficiently large. In particular, for any $\epsilon>0$,
\e
\P{ \max_{k}  \abs{p_k - \wh p_{k}} \ge \epsilon} \le 2e^{-\frac{1}{2} M \epsilon^2}.
\label{eq:dfw2}
\ee

\paragraph*{Haar random projective measurement}

A particular type of POVM that is commonly used in quantum experiments is the so-called projective measurement. A projective measurement is a rank-one POVM of the form $\{ \mA_k=\vphi_k \vphi_k^\dagger \}$ with $\sum_{k=1}^K \vphi_k \vphi_k^\dagger = \mId$ and $K = d^n$. We also require $\{\vphi_k$\} to be unit length and orthogonal to each other (i.e. orthonormal). Therefore, each $\mA_k$ is a projection operator onto the corresponding basis vector $\vphi_k$. The probability in \eqref{The defi of POVM 2} can be rewritten as
\e
p_k = \innerprod{\mA_k}{\vrho} = \innerprod{\vphi_k \vphi_k^\dagger}{\vrho} = \vphi_k^\dagger \vrho \vphi_k.
\label{Probability of rank-one POVM}\ee

In practice, measurements on a quantum computer or quantum simulator are usually projective measurements using some physically convenient basis such as the computational basis $\{|0\rangle,|1\rangle\}$ for a qubit. If we want to perform a projective measurement defined by an arbitrary basis $\{\vphi_k\}$, we can introduce a unitary matrix $\mU = \begin{bmatrix} \vphi_1 & \cdots & \vphi_K \end{bmatrix}\in\C^{d^n\times d^n}$ and apply $\mU$ on the state $\vrho$ before performing the projective measurement with a physically convenient basis (denoted by $\{\ve_k\}$ below) where $\mU \ve_k = \vphi_k$ for any $k$. Mathematically, this process is written as
\e
p_k = \ve_k^\dagger \bracket{\mU^{\dagger} \vrho \mU} \ve_k,
\label{Probability of rank-one orth POVM}
\ee

Here we are particularly interested in performing a projective measurement in a random basis, where the unitary matrix $\mU$ is randomly drawn according to the Haar measure. A universal quantum computer can approximately generate such random unitary to any given precision, although the number of single and two-qubit quantum gates required in general scale exponentially with the number of qubits  \cite{knill1995approximation}. We call such projective measurement a {\em Haar random  projective measurement}. Such measurement has been commonly used in optimal quantum state tomography protocols \cite{haah2017sample}, as it typically provides the most unbiased information of an unknown quantum state.

\paragraph*{Multiple POVMs} A projective measurement in a specific basis is insufficient to recover a general quantum state $\vrho$ even if we repeat the measurement infinitely many times, since the probabilities $\{p_k\}$ in Eq.\,\eqref{Probability of rank-one POVM} only provide us the diagonal elements of $\vrho$ in the basis formed by $\{\vphi_k\}$. In most QST protocols, one performs projective measurements in different bases, or more generally, multiple POVMs, in order to gain full information of the quantum state. In the following, we denote the number of different POVMs (or projective measurement bases) by $Q$, and the measurement operators for the $i$-th POVM by $\{\mA_{i,1},\ldots,\mA_{i,K} \}$. For simplicity, we have assumed that each POVM contains the same number of measurement operators (which is always true for projective measurements), although in general this number may vary between POVMs.

We use each POVM to measure a state $M$ times to obtain the empirical measurements as described in {Definition} \ref{def:POVM}. Thus in total we need $QM$ copies of the quantum state we want to perform QST. For a generic, unstructured quantum state, the value of $QM$ scales exponentially with the number of qudits, making QST impractical for large quantum systems. This is true even if the state can be represented by a low-rank density matrix \cite{haah2017sample,francca2021fast}. Efficient QST may be possible if we have a structured quantum state that can be represented efficiently, i.e. by a number of independent parameters polynomial in the number of qudits. In the next subsection, we introduce a particular type of states that can be efficiently represented by the so-called matrix product operators \cite{jarkovsky2020efficient}.

\begin{figure*}[t]
\centering
\includegraphics[width=17cm, keepaspectratio]%
{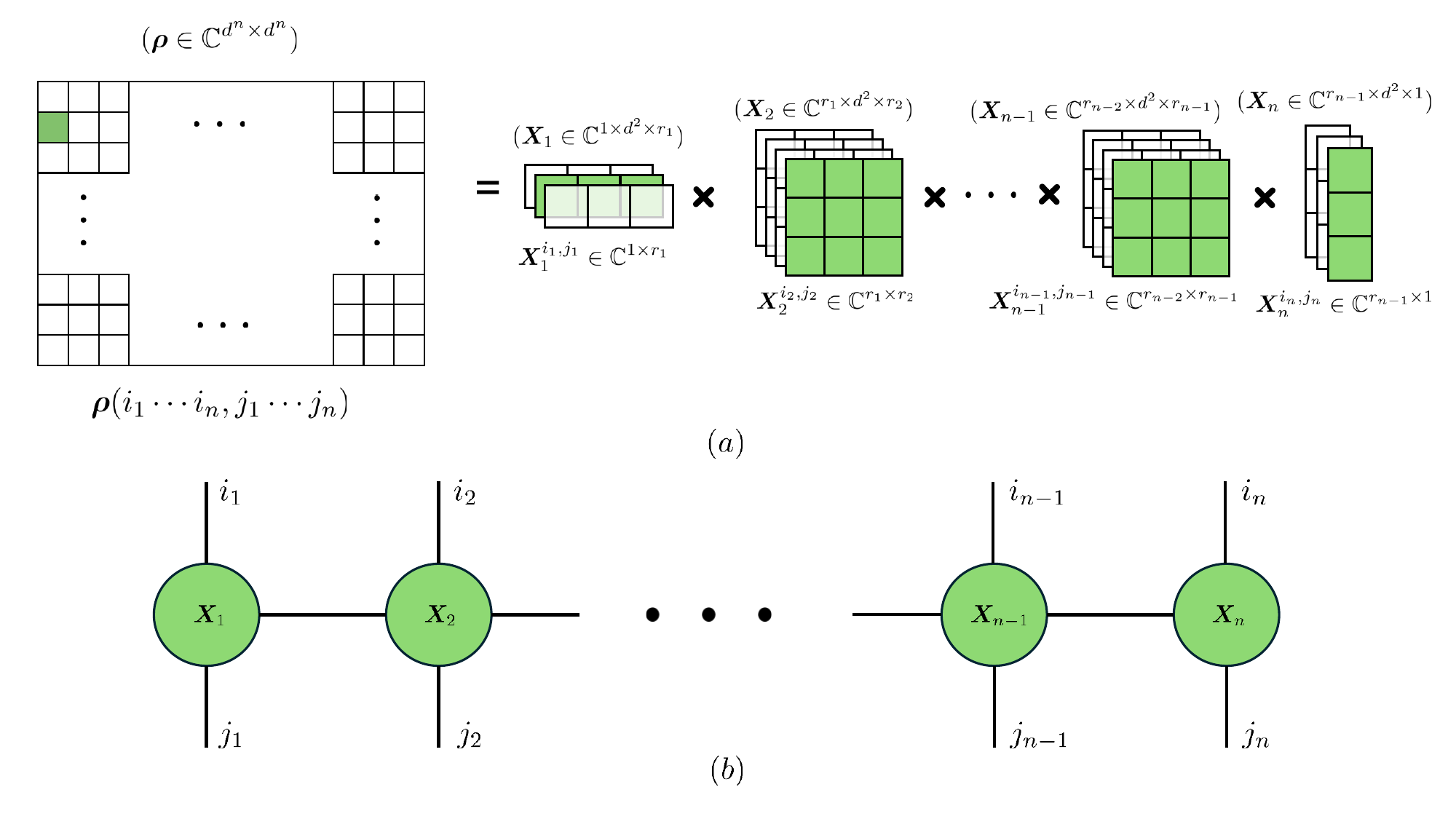}
\vspace{-0cm}
\caption{Illustration of the MPO in \eqref{DefOfMPO} from two perspectives: (a) each entry of the density matrix can be represented as products of $n$ matrices, where green represents one entry and the corresponding $n$ matrices, and (b) each element of the density matrix is illustrated in a diagrammatic form, where the line connecting two circles signifies the tensor contraction operation \cite{cichocki2014tensor}, and unconnected line segments denote indices.}
\label{TheMPOfig}
\end{figure*}

\subsection{Matrix Product Operator (MPO)}
\label{sec:MPO}

For a density matrix $\vrho\in\C^{d^\nqbit\times d^ \nqbit}$ corresponding to an $\nqbit$-qudit quantum system, we use a single index-array $i_1\cdots i_\nqbit$ ($j_1\cdots j_\nqbit$) to specify the indices of rows (columns), where $i_1,\ldots,i_\nqbit\in[d]$.\footnote{ Specifically, $i_1\cdots i_n$ represents the $(i_1+\sum_{\ell=2}^nd^{\ell-1}(i_\ell-1))$-th row.
} Then we say $\vrho$ is an MPO if we can express its $(i_1\cdots i_\nqbit,j_1\cdots j_\nqbit)$-element as the following matrix product~\cite{werner2016positive}
\begin{eqnarray}
\label{DefOfMPO}
\vrho(i_1 \cdots i_\nqbit, j_1 \cdots j_\nqbit)  =  \mX_1^{i_1,j_1} \mX_2^{i_2,j_2} \cdots \mX_\nqbit^{i_\nqbit,j_\nqbit},
\end{eqnarray}
where $\mX_\ell^{i_\ell,j_\ell}\in \C^{r_{\ell-1}\times r_\ell}$ with $r_0 = r_\nqbit = 1$. See Figure~\ref{TheMPOfig} for an illustration. The dimensions $\vr = (r_1,\ldots,r_{n-1})$ are often called the {\it bond dimensions}\footnote{It is also common to simply call $\overline{r} = \max\{r_1,\ldots,r_{n-1}\}$ the bond dimension.} of the MPO in quantum physics, though we may also call them the {\it MPO ranks}. These dimensions can indeed be viewed as the ranks of certain matrices that are obtained by reshaping the density matrix $\vrho$ in various ways.

\paragraph*{Connection to the tensor train (TT) format}
An MPO is equivalent to a tensor train (TT) used to describe high-dimensional tensors \cite{Oseledets11}. To see this, we first reshape $\vrho$ into an $n$-th order tensor $\calX$ of size $d^2\times d^2\times \cdots \times d^2$ by mapping each pair $(i_\ell,j_\ell)$ into a single index $s_\ell=i_\ell+d(j_\ell-1),\ell=1,\ldots,n$
such that the elements of $\calX$ are given by
\begin{eqnarray}
\label{DefinitionOfMPO}
\calX(s_1,\dots,s_n)
=\vrho(i_1 \cdots i_\nqbit, j_1\cdots j_\nqbit).
\end{eqnarray}
Note that $\calX$ is just a reshaping of $\vrho$ and that both objects contain exactly the same entries. Then, according to \eqref{DefOfMPO}, the $(s_1,\ldots,s_n)$-th element of $\calX$ can also be represented as a matrix product
\begin{eqnarray}
\label{DefTT}
\calX(s_1, \ldots, s_\nqbit)  =  \mX_1^{s_1} \mX_2^{s_2} \cdots \mX_\nqbit^{s_\nqbit},
\end{eqnarray}
where with abuse of notation we denote $\mX_\ell^{s_\ell} = \mX_\ell^{i_\ell,j_\ell}$. The decomposition in \eqref{DefTT} is known as the TT decomposition and has been widely studied in the literature~\cite{Oseledets11,holtz2012manifolds,rauhut2017low,yuan2019high,zhao2016tensor,yuan2019tensor}.

\paragraph*{Canonical form} When $n = 2$, the decomposition \eqref{DefTT} is equivalent to the standard matrix factorization of the form $\mA = \mB\mC$, where $\mA\in\R^{d^2\times d^2}, \mB\in\R^{d^2\times r},\mC\in\R^{r\times d^2}$, the rows of $\mB$ correspond to $\mX_1^{s_1}$ and the columns of $\mC$ correspond to $\mX_2^{s_2}$. There exist infinitely many possible choices of $(\mB,\mC)$ such that $\mB\mC = \mA$, but all of them require $r \ge \rank(\mA)$. Among all these possible factorizations, if $\rank(\mB) = \rank(\mC) = r$, then $r = \rank(\mB\mC) = \rank(\mA)$, implying that this is the \emph{minimal} $r$ allowed for the factorization $\mA = \mB\mC$. Moreover, one can always construct a factorization (say by the singular value decomposition) such that $\mB$ is orthogonal with $\mB^\top \mB = \mId_r$, or $\mC$ is orthogonal  with $\mC\mC^\top = \mId_r$.

Likewise, the decomposition of the tensor $\calX$ into the form of \eqref{DefTT} is generally not unique: not only are the factors $\{\mX_{\ell}^{i_\ell,j_\ell}\}$ not unique, but also the dimensions of these factors can vary. To introduce the factorization with the smallest possible dimensions $\vr = (r_1,\ldots,r_{n-1})$, for convenience, for each $\ell$, we put $\mX_\ell = \{\mX_{\ell}^{i_\ell,j_\ell}\}_{i_\ell,j_\ell}$ together into the following two forms
\[
L(\mX_\ell)=\begin{bmatrix}\mX_{\ell}^{1,1} \\ \vdots\\  \mX_{\ell}^{d,d} \end{bmatrix}, \ R(\mX_\ell)=\begin{bmatrix}\mX_{\ell}^{1,1} &  \cdots &  \mX_{\ell}^{d,d} \end{bmatrix},
\]
where $L(\mX_\ell)$ and $R(\mX_\ell)$ are often called the left unfolding and right unfolding of $\mX_\ell$, respectively, if we view $\mX_\ell$ as a tensor. We say the decomposition \eqref{DefTT} is \emph{minimal} if the rank of the left unfolding matrix $L(\mX_\ell)$ is $r_\ell$ and the rank of the right unfolding matrix $R(\mX_\ell)$ is $r_{\ell-1}$. The dimensions $\vr = (r_1,\dots, r_{n-1})$ of such a minimal decomposition are called the \emph{TT ranks} of $\calX$. According to \cite{holtz2012manifolds}, there is exactly one set of ranks $\vr$ that $\calX$ admits a minimal TT decomposition. Moreover, in this case, $r_\ell$ equals the rank of the $\ell$-th unfolding matrix $\mX^{\< \ell\>}\in\C^{  d^{2\ell} \times d^{2\nqbit-2\ell} }$ of the tensor $\calX$, where the $(s_1\cdots s_\ell, s_{\ell+1}\cdots s_{\nqbit})$-th element of $\mX^{\< \ell\>}$ is given by $\mX^{\< \ell\>}(s_1\cdots s_\ell, s_{\ell+1}\cdots s_{\nqbit}) = \calX(s_1,\dots, s_{\nqbit})$. This can also serve as an alternative way to define the TT rank.
As for the matrix case, for any MPO $\vrho$ of the form \eqref{DefOfMPO},
there always exists a factorization such that $L(\mX_\ell)$ are unitary matrices for all $\ell = 1,\ldots,n-1$; that is
\begin{eqnarray}
\label{eq:left-canonical}
L(\mX_\ell)^\dagger L(\mX_\ell) &\!\!\!\!=\!\!\!\!& \sum_{i_\ell,j_\ell}\parans{\mX_\ell^{i_\ell,j_\ell}}^\dagger \mX_\ell^{i_\ell,j_\ell} = \mId_{r_\ell}, \nonumber\\
&\!\!\!\!\!\!\!\!&\ell = 1,\ldots,n-1,
\end{eqnarray}
which is called the left-canonical form\footnote{The right-canonical form refers to the case where $R(\mX_\ell)$ are unitary matrices for all $\ell = 2,\ldots,n$.} \cite{perez2007matrix}. According to \cite[Theorem 1]{holtz2012manifolds}, such a  canonical form is unique up to the insertion of orthogonal matrices between the factors. Thus, we will denote by $\setX_{\ol{r}}$ the set of MPOs with maximum MPO rank equal to $\ol r$:
\begin{align}
\label{SetOfMPO}
&\setX_{\ol{r}}= \Big\{\vrho\in\C^{d^n\times d^n}: \ \vrho = \vrho^\dagger, \vrho(i_1 \cdots i_\nqbit, j_1\cdots j_\nqbit)= \nonumber\\
&\hspace{0.2cm}\Pi_{\ell=1}^{\nqbit}\mX_\ell^{i_\ell,j_\ell}, \mX_\ell^{i_\ell,j_\ell}\in\C^{r_{\ell-1}\times r_\ell},
\sum_{i_\ell,j_\ell} \parans{\mX_\ell^{i_\ell,j_\ell}}^\dagger \mX_\ell^{i_\ell,j_\ell} = \mId_{r_\ell},\nonumber\\
&\hspace{0.2cm}\ell=1,\dots,n-1, r_0=r_n=1, \ol r = \max\{r_\ell\}\Big\}.
\end{align}
Note that the set \eqref{SetOfMPO} contains not only PSD matrices but also non-PSD matrices.
In \Cref{sec:Stable embedding MPO}, we establish stable embedding guarantees that hold for measurements of any $\vrho\in \setX_{\ol{r}}$; these results are then used on our analysis of a properly constrained least-squares minimization in \Cref{sec:stable recovery}.
While the set $\setX_{\ol{r}}$ does include some non-physical matrices, we stress that it does contain {\em all} physical MPOs (with maximum MPO rank equal to $\ol r$). So all physical MPOs are covered by our stable embedding guarantees.
We also note that one could consider imposing additional structure, such as in~\cite[eq. (3)]{verstraete2004matrix},  on the factors $\{\mX_{\ell}^{i_\ell,j_\ell}\}$ to ensure $\vrho$ is PSD. However, the condition in \cite[eq. (3)]{verstraete2004matrix} is only sufficient rather than necessary for ensuring $\vrho$ is PSD, and adding the PSD and trace constraints does not significantly reduce the number of degrees of freedom of elements in the set $\setX_{\ol{r}}$.

\paragraph*{Efficiency of MPO representation} Due to the curse of dimensionality, the number of elements in the density matrix $\vrho$ grows exponentially in the number of qudits $n$. In contrast, the MPO form \eqref{DefOfMPO} can represent $\vrho$ using only $O(n d^2 \overline r^2)$ elements, where $\ol r = \max\{r_1,\ldots,r_{n-1}\}$. This makes the MPO form remarkably effective in combatting the curse of dimensionality as its number of parameters scales only linearly in terms of $n$. The concise representation provided by MPO is remarkably useful in QST since it may allow us to reconstruct a quantum state with both experimental and computational resources that are only \emph{polynomial} rather than \emph{exponential} in the number of qudits~\cite{verstraete2006matrix,verstraete2008matrix,schollwock2011density,ohliger2013efficient}. Beyond applications in quantum information processing, the equivalent form of TT decomposition mentioned above has also been widely used for image compression \cite{latorre2005image,bengua2017efficient},
analyzing theoretical properties of deep networks \cite{khrulkov2018expressive}, network compression (or tensor networks)~\cite{stoudenmire2016supervised,novikov2015tensorizing,yang2017tensor,tjandra2017compressing,yu2017long,ma2019tensorized}, recommendation systems \cite{frolov2017tensor}, probabilistic model estimation \cite{novikov2021tensor}, and learning of Hidden Markov Models \cite{kuznetsov2019tensor} to mention a few usages.\footnote{See \cite{novikov2020tensor} for a python library for TT decomposition.}

\paragraph*{Linear combination of MPOs} In linear algebra, the (matrix) rank of the sum of two matrices is less than or equal to the sum of the (matrix) ranks of these matrices. This also holds for MPO ranks. In particular, for any two MPOs $\wt \vrho, \wh \vrho \in\C^{d^n \times d^n}$ of the form \eqref{DefOfMPO} with factors $\{\wt \mX^{i_\ell,j_\ell}\in \C^{\wt r_{\ell-1}\times \wt r_\ell}\}$ and $\{\wh \mX^{i_\ell,j_\ell}\in \C^{\wh r_{\ell-1}\times \wh r_\ell}\}$, respectively, the elements of their summation $\vrho = \wt \vrho + \wh \vrho$ can be expressed by
\begin{eqnarray}
\label{eq:sum of MPOs}
\vrho(i_1 \cdots i_\nqbit, j_1 \cdots j_\nqbit) &\!\!\!\!=\!\!\!\!& \begin{bmatrix}\wt\mX_1^{i_1,j_1} & \wh\mX_1^{i_1,j_1} \end{bmatrix}
\begin{bmatrix}\wt\mX_2^{i_2,j_2}& {\bm 0} \\ {\bm 0} & \wh\mX_2^{i_2,j_2} \end{bmatrix}\nonumber\\
&\!\!\!\!\!\!\!\!&\hspace{-1.5cm}\cdots \begin{bmatrix}\wt\mX_{n-1}^{i_{n-1},j_{n-1}}& {\bm 0} \\ {\bm 0} & \wh\mX_{n-1}^{i_{n-1},j_{n-1}} \end{bmatrix} \begin{bmatrix}\wt\mX_n^{i_n,j_n} \\ \wh\mX_n^{i_n,j_n} \end{bmatrix},
\end{eqnarray}
implying that the MPO ranks $r_\ell$ of $\vrho$ satisfy $r_\ell \le \wh r_\ell + \wt r_\ell$ for all $\ell = 1,\ldots, n-1$.

\section{Stable Embeddings of  Matrix Product Operators}
\label{sec:Stable embedding MPO}

\subsection{Background}
Measurements must satisfy certain properties to enable recovery of quantum states. One desirable property known as a \emph{stable embedding} has been widely studied and popularized in the compressive sensing literature~\cite{donoho2006compressed,candes2006robust,
candes2008introduction,recht2010guaranteed,eftekhari2015new}. In this section, we will study the embedding of MPOs from various measurement types including quantum measurements. Towards that goal, we will first consider population measurements, and in the next section, we will study stable recovery with empirical measurements.

As described in \Cref{sec:POVM measurements}, the population measurements from one POVM are linear measurements that can be described through a linear map $\calA:  \C^{d^n\times d^n} \rightarrow \R^K$ of the form
\e
\calA(\vrho)= \begin{bmatrix}
          \< \mA_1, \vrho  \> \\
          \vdots \\
          \< \mA_K, \vrho  \>
        \end{bmatrix}.
\label{eq:linear map A}\ee
According to the discussion in \Cref{sec:POVM measurements}, the choice of $\{\mA_k\}$ can vary.  Our goal is to study the properties of the associated measurement operators.

Our study of stable embeddings of MPOs from population measurements concerns the quantity $\|\calA(\vrho)\|_2^2$. As described in Section~\ref{sec: RIP gaussian}, a favorable situation is when
$\calA$ satisfies the restricted isometry property (RIP), where $\|\calA(\vrho)\|_2^2$ is guaranteed to be proportional to $\|\vrho\|_F^2$ for any MPO $\vrho$. In some cases, only a lower bound on this proportionality can be established. 
In particular, in Section~\ref{sec:rankOnePOVMpop}, we establish a guarantee of the form
\e\|\calA(\vrho)\|_2^2 \geq C_{d,n,K} \|\vrho\|_F^2,
\label{eq:embedding}\ee
where $C_{d,n,K}$ is a positive constant depending on $d,n,K$, and the guarantee holds uniformly for all MPOs up to some maximum rank. When this holds, then for any two MPOs $\vrho_1$ and $\vrho_2$, noting that $\vrho_1 - \vrho_2$ is also an MPO according to \eqref{eq:sum of MPOs}, we have
\[
\norm{\calA(\vrho_1) - \calA(\vrho_2)}{2}^2 \geq C_{d,n,K} \norm{\vrho_1 - \vrho_2}{F}^2,
\]
which ensures distinct measurements (i.e., $\calA(\vrho_1)\neq \calA(\vrho_2)$) as long as $\vrho_1 \neq \vrho_2$.

In compressive sensing of sparse signals and low-rank matrices~\cite{donoho2006compressed,candes2006robust,
candes2008introduction,recht2010guaranteed,eftekhari2015new}, uniform stable embeddings of all possible signals of interest can often be achieved by choosing the measurement operators randomly from a certain distribution. Thus,
random matrices and projections  have played a central role in the analysis of the associated inverse problems~\cite{tropp2015convex}. In this section, we will study the embeddings of MPOs from linear measurements where the measurement matrices $\{\mA_k\}$ are generated from certain random distributions. Specifically, we will first study perhaps the most generic random distribution where all the elements of $\mA_k$ are independently generated from a Gaussian distribution. We will then study rank-one random POVM measurements of the form $\mA_k = \va_k\va_k^\dagger$ with each $\va_k$ randomly generated from a multivariate complex normal distribution. Finally, we will study the physically realizable (though inefficient) measurements consisting multiple Haar random projective measurements.

\paragraph*{Normalized set of MPOs} Since $\calA(\cdot)$ is a linear map, without loss of generality, we will focus on MPOs $\vrho \in\setX_{\ol r}$ with unit Frobenius norm. By the left-canonical form in \eqref{eq:left-canonical}, we have
\begin{align}
\|\vrho\|_F^2&= \hspace{-0.1cm}\sum_{i_1,j_1}\hspace{-0.1cm}\cdots\hspace{-0.1cm} \sum_{i_\nqbit,j_\nqbit}\parans{\mX_\nqbit^{i_\nqbit,j_\nqbit}}^\dagger\hspace{-0.1cm} \cdots \parans{\mX_1^{i_1,j_1}}^\dagger\!\! {\mX_1^{i_1,j_1}}\hspace{-0.1cm} \cdots\hspace{-0.1cm} \mX_\nqbit^{i_\nqbit,j_\nqbit} \nonumber\\
&= \sum_{i_2,j_2}\cdots \sum_{i_\nqbit,j_\nqbit}\parans{\mX_\nqbit^{i_\nqbit,j_\nqbit}}^\dagger\cdots \parans{\mX_2^{i_2,j_2}}^\dagger \nonumber\\
&\cdot \underbrace{ \parans{\sum_{i_1,j_1}\parans{\mX_1^{i_1,j_1}}^\dagger{\mX_1^{i_1,j_1}}} }_{\mId_{r_1}}
{\mX_2^{i_2,j_2}} \cdots \mX_\nqbit^{i_\nqbit,j_\nqbit}  \nonumber\\
&= \hspace{-0.1cm}\sum_{i_2,j_2}\hspace{-0.1cm}\cdots\hspace{-0.1cm} \sum_{i_\nqbit,j_\nqbit}\parans{\mX_\nqbit^{i_\nqbit,j_\nqbit}}^\dagger\hspace{-0.1cm}\cdots \hspace{-0.1cm} \parans{\mX_2^{i_2,j_2}}^\dagger\!\! {\mX_2^{i_2,j_2}}\hspace{-0.1cm} \cdots\hspace{-0.1cm} \mX_\nqbit^{i_\nqbit,j_\nqbit} \nonumber\\
&= \cdots = \sum_{i_\nqbit,j_\nqbit}{\mX_\nqbit^{i_\nqbit,j_\nqbit}}^\dagger \mX_\nqbit^{i_\nqbit,j_\nqbit},
\end{align}
which also leads to $\sum_{i_\nqbit,j_\nqbit}\parans{\mX_\nqbit^{i_\nqbit,j_\nqbit}}^\dagger \mX_\nqbit^{i_\nqbit,j_\nqbit}=1$ together with $\|\vrho\|_F^2 = 1$. Thus, the set of all the MPOs $\vrho \in \setX_{\ol r}$ with unit norm, denoted by $\ol \setX_{\ol r}$, can also be expressed by
\begin{eqnarray}
\label{SetOfMPO1}
& \!\!\!\!\!\!\!\!&\hspace{-0.2cm}\ol\setX_{\ol{r}} = \Big\{\vrho\in\C^{d^n\times d^n}:\ \vrho = \vrho^\dagger, \vrho(i_1 \cdots i_\nqbit, j_1\cdots j_\nqbit) = \nonumber\\
& \!\!\!\! \!\!\!\!&\hspace{0cm}\Pi_{\ell=1}^{\nqbit}\mX_\ell^{i_\ell,j_\ell},   \mX_\ell^{i_\ell,j_\ell}\in\C^{r_{\ell-1}\times r_\ell},
\sum_{i_\ell,j_\ell} \parans{\mX_\ell^{i_\ell,j_\ell}}^\dagger \mX_\ell^{i_\ell,j_\ell} = \mId_{r_\ell},\nonumber\\
& \!\!\!\! \!\!\!\!&\hspace{0.3cm}\ell=1,\dots,n, r_0=r_n=1, \ol r = \max\{r_\ell\}\Big\}.
\end{eqnarray}

\subsection{Restricted isometry property with generic Gaussian measurements}
\label{sec: RIP gaussian}

To provide a baseline for the sample complexity of population measurements, we begin by studying perhaps the most generic type of random measurements, where each entry of $\mA_k$ is i.i.d. standard complex Gaussian random variable $X = \mathscr{R}(X) + i \mathscr{I}(X)$ with $\mathscr{R}(X)$ and $\mathscr{I}(X)$ being independent and following $ \calN(0,\frac{1}{2})$, the Gaussian distribution with mean $0$ and variance $\frac{1}{2}$. Such measurements do not form a POVM and thus cannot be physically implemented in quantum measurement systems. However, as Gaussian measurements provide the ``gold standard'' for random linear measurement operators in many compressive sensing and low-rank matrix recovery problems, their sample complexity for stable embeddings of MPOs provides useful insight.

Gaussian measurements  can be shown to satisfy a strong type of stable embedding guarantee known as the restricted isometry property (RIP).

\begin{defi}[Restricted isometry property (RIP)]
A linear operator $\calA: \C^{d^n\times d^n} \rightarrow \C^K$ is said to satisfy the $\delta_{\ol{r}}$-restricted isometry property ($\delta_{\ol{r}}$-RIP) if
\begin{equation}
    \label{eq:ripdef}
    (1-\delta_{\ol{r}}) \| \vrho \|_F^2 \le \frac{1}{K}\| \calA(\vrho) \|_2^2 \le (1+\delta_{\ol{r}}) \| \vrho \|_F^2
\end{equation}
holds for any density operator $\vrho\in\C^{d^n\times d^n}$ which has the MPO format with MPO ranks $\vr = (r_1,\ldots,r_{n-1}), r_i \le \ol r$.
\label{def:RIP}\end{defi}

The following result establishes the RIP for Gaussian measurements.
\begin{theorem}
Suppose that each entry of $\mA_k$ in
the linear map $\calA: \C^{d^n \times d^n} \rightarrow \C^K$ defined in \eqref{eq:linear map A} is an i.i.d.\ standard complex Gaussian random variable. Then, with probability at least $1-\bar\epsilon$, $\calA$ satisfies the $\delta_{\ol{r}}$-RIP as in \eqref{eq:ripdef} for MPOs given that
\begin{equation}
K \ge C \cdot \frac{1}{\delta_{\ol{r}}^2} \cdot \max\left\{ nd^2{\ol{r}}^2 (\log n\ol{r}), \log(1/\bar\epsilon)\right \},
\label{eq:mrip}
\end{equation}
where $C$ is a universal constant.
\label{thm:SubgaussianRIP}
\end{theorem}
In {Appendix} \ref{Proof of the RIP}, we extended this result to generic subgaussian measurements. We note that a similar result for TT-format tensors in the real domain was given in \cite{rauhut2017low}, and we share similar techniques for proving the RIP by using tools involving the $\epsilon$-net and covering arguments \cite{baraniuk2007simple,vershynin2010introduction} and deviation bounds for the supremum of a chaos process \cite{krahmer2014suprema,dirksen2015tail}. While MPOs are equivalent in form to TT-format tensors as discussed in \Cref{sec:MPO}, we provide the proof in {Appendix} \ref{Proof of the RIP} for the sake of completeness and because here we consider the complex domain.
Also, the sampling complexity in \cite{rauhut2017low} is $K \gtrsim \frac{1}{\delta_{\ol{r}}^2} \cdot \max\left\{ ( (n-1){\ol{r}}^3 + nd^2\ol{r} )(\log n\ol{r}), \log(1/\bar\epsilon)\right \}$, which is slightly different from \eqref{eq:mrip}.
Considering a qubit system with $d=2$, the main order $n{\ol{r}}^2$ in \eqref{eq:mrip} is slightly better than the order $(n-1){\ol{r}}^3 $ from \cite{rauhut2017low} when the bond dimension $\ol{r}$ is large.

Although the Gaussian measurements are not POVMs and cannot be directly used for quantum measurements, \Cref{thm:SubgaussianRIP} indicates that it is possible to estimate an MPO state with
$\wt\Omega(nd^2{\ol{r}}^2)$ linear measurements. In comparison, for a state with low (matrix) rank structure, say rank $r$, $\wt\Omega(d^nr)$ measurements are needed even with Gaussian measurements \cite{candes2011tight}.

\subsection{Stable embeddings with rank-one measurements}
\label{sec:rankOnePOVMpop}
We now study the population measurements arising from structured rank-one measurement ensembles with PSD matrices $\mA_k = \vpsi_k \vpsi_k^\dagger$  as introduced in Section~\ref{sec:POVM measurements}. We first consider the case where we omit the constraint \eqref{The defi of POVM 1} that the matrices $\mA_k$ sum to the identity matrix. Rather, we simply generate the $\vpsi_k =\va_k$ independently and randomly from a certain distribution, specifically, $\va_k\sim\calC\calN({\bm 0},\mId_{d^n})$. The independence among $\{\va_k\}$ will simplify the analysis and help derive a tight bound for stable embedding. We call such measurements {\em rank-one independent POVM measurements}.
We then consider the practical case ($\vpsi_k = \vphi_k$) where $\{\vphi_k\}$ are generated from a Haar-distributed random unitary matrix, which results in  {\em Haar random projective measurements}.

\paragraph*{Rank-one Gaussian measurements}
It is known that  rank-one measurements do not obey the RIP condition  for low-rank matrices \cite{zhong2015efficient,KuengACHA17}. Since we expect this to also be true for MPOs, we instead aim to establish a lower bound on the isometry of the form \eqref{eq:embedding}. Towards that goal, we will use Mendelson's small ball method \cite{mendelson2015learning, koltchinskii2015bounding,tropp2015convex} for establishing a lower bound on a nonnegative empirical process. 
\begin{lemma} (\cite{mendelson2015learning, koltchinskii2015bounding,tropp2015convex})
\label{Small Ball Method_original}
Fix a set $E\subset \C^{D}$. Let $\vb$  be a random vector on $\C^{D}$
and let $\vb_{1},\dots,\vb_{K}$ be  independent copies of $\vb$.
Introduce the marginal tail function
\begin{eqnarray}
    \label{marginal_tail_function_original}
    H_{\xi}(E;\vb) = \inf_{\vu\in E} \mathbb{P} \{|\<\vb,\vu \>|\geq \xi  \},\ \text{for} \ \xi>0.
\end{eqnarray}

Let $\epsilon_k,k=1,\dots,K,$ be independent Rademacher random variables, independent from everything else. Define the
mean empirical width of the set $E$ as
\begin{eqnarray}
    \label{mean empirical width_original}
    W_{K}(E;\vb) = \E \sup_{\vu\in E}\<\vh ,\vu\>, \ \text{where} \ \vh =\frac{1}{\sqrt{K}}\sum_{k=1}^K \epsilon_{k}\vb_{k}.
\end{eqnarray}

Then, for any $\xi >0$ and $t >0$, with probability at least $1-e^{-\frac{t^2}{2}}$ we have
\begin{align}
    \label{Final_Conclusion_original}
    \inf_{\vu\in E}\bigg(\sum_{k=1}^K|\<\vb_{k}, \vu  \>|^2  \bigg)^{\frac{1}{2}}\hspace{-0.1cm} \geq\xi\sqrt{K}H_{\xi}(E;\vb)-2W_{}(E;\vb) -t \xi.
\end{align}
\end{lemma}
This result delivers
an effective lower bound for a nonnegative empirical process defined in the left-hand side of \eqref{Final_Conclusion_original}. This result is also utilized for studying stable embeddings for low-rank matrices \cite{KuengACHA17,rauhut2019low}. Noting the similar forms between \eqref{Final_Conclusion_original} and \eqref{eq:embedding}, we apply \Cref{marginal_tail_function_original} for our case where the set $E$ becomes $\ol\setX_{\ol{r}}$ and $\vb$ becomes a random measurement matrix of form $\mA = \va\va^\dagger$ with $\va\sim\calC\calN({\bm 0},\mId_{d^n})$. We then need to analyze the following marginal tail function and mean empirical width
\begin{align*}
H_{\xi}(\ol\setX_{\ol{r}};\mA) & = \inf_{\vrho\in \ol\setX_{\ol{r}}} \mathbb{P} \{|\<\mA,\vrho \>|\geq \xi  \},\\
W_{K}(\ol\setX_{\ol{r}};\mA) & = \frac{1}{\sqrt{K}}\E \sup_{\vrho\in \ol\setX_{\ol{r}}}\sum_{k=1}^K\<\epsilon_{k}\mA_k ,\vrho\>.
\end{align*}
As in \cite{tropp2015convex,KuengACHA17},
 we can use the Payley-Zygmund inequality to obtain a lower bound for the marginal tail function $H_{\xi}(\ol\setX_{\ol{r}};\mA)$. In terms of the mean empirical width $W_{K}(\ol\setX_{\ol{r}};\mA)$, the work \cite{tropp2015convex,KuengACHA17} uses an inequality that directly upper bounds the supremum of $\langle \mA,\vrho \rangle $ over rank-$r$ matrices $\vrho$ by $2\sqrt{r}\|\mA\|$. Unfortunately, it is difficult to extend this approach to our case.
 Instead, we use an $\epsilon$-net argument to provide a uniform upper bound for $\langle \mA,\vrho \rangle$. With the detailed analysis in {Appendix} \ref{Proof_Of_MPO_Subgaussian}, we establish the following result.

\begin{theorem}
\label{Small_Method_ROSubgaussian_Measurement}
Let $\{ \va_{1},\dots, \va_{K} \}$ be selected independently and randomly  from the multivariate standard complex normal distribution $\mId_{d^n}$. Given
\begin{eqnarray}
    \label{number of sample rank one gaussian}
    K\gtrsim n d^2\ol{r}^2\log n,
\end{eqnarray}
then the induced linear map $\calA$ with measurement operators $\{\mA_k=\va_{k}\va_{k}^\dagger\}$ satisfies
\begin{eqnarray}
\label{Small_Method_ROSubgaussian_Measurement_Conclusion_Theorem}
     \|\calA(\vrho)\|_2  = \bigg(\sum_{k=1}^K |\<\va_{k} \va_{k}^\dagger, \vrho  \>|^2 \bigg)^{\frac{1}{2}}  \gtrsim \sqrt{K}, \ \forall \vrho \in \ol\setX_{\ol r}
\end{eqnarray}
with probability at least $1-e^{-\alpha_1 K}$,
where $\alpha_1$ is a positive constant.
\end{theorem}

Under the same setup, one requires $K \gtrsim d^nr$ measurement operators for the induced linear map $\calA$ to obey the stable embedding property for rank-$r$ matrices \cite{KuengACHA17}.  Fortunately, due to the extremely low-dimensional structure of the MPO format, the number of measurement operators only needs to scale linearly in terms of the number of qudits $n$ (if we ignore the logarithmic term).

\paragraph*{Haar random  projective measurements}

We now study practical measurements consisting of an ensemble of Haar random  projective measurements as described in \Cref{sec:POVM measurements}. Let $\begin{bmatrix}\vphi_{i,1} &  \cdots & \vphi_{i,d^n}\end{bmatrix}, \ i = 1,\ldots, Q$ be $Q$ randomly generated Haar-distributed unitary matrices.
According to \Cref{sec:POVM measurements}, each unitary matrix induces a linear operator $\calA_i:\C^{d^n\times d^n} \rightarrow \R^{K}$ that generates population measurements for a quantum state $\vrho$ as
\begin{eqnarray}
\label{The defi of POVM element in Q cases (K measurements)}
 \calA_i(\vrho) =  \begin{bmatrix}
          \< \mA_{i,1}, \vrho  \> \\
          \vdots \\
          \< \mA_{i,K}, \vrho  \>
        \end{bmatrix} =  \begin{bmatrix}
          \< \vphi_{i,1} \vphi_{i,1}^\dagger, \vrho  \> \\
          \vdots \\
          \< \vphi_{i,K} \vphi_{i,K}^\dagger, \vrho  \>
        \end{bmatrix},
\end{eqnarray}
where in practice we will use $K = d^n$, but for generality we can choose any $K \le d^n$. We note that for each $i$, even though $\begin{bmatrix}\vphi_{i,1} & \cdots  & \vphi_{i,d^n}\end{bmatrix}$ is unitary and $\sum_{k=1}^{d^n} \vphi_{i,k} \vphi_{i,k}^\dagger = \mId$, $\calA_i$ is not an identity mapping in $\C^{d^n\times d^n}$ even with $K = d^n$; this is because $\calA_i$ collects at most $d^n$ measurements of an object $\vrho$ that contains $d^{2n}$ entries.
We now stack all the population measurements together as
\begin{eqnarray}
\label{The defi of population measurement in Q cases (K measurements)}
 \calA^Q(\vrho)  = \begin{bmatrix}
          \calA_1(\vrho) \\
          \vdots \\
          \calA_Q(\vrho)
        \end{bmatrix},
\end{eqnarray}
where $ \calA^Q:\C^{d^n\times d^n} \rightarrow \R^{KQ}$ denotes the linear operator corresponding to the $Q$ POVMs.

For any $i$, since $\vphi_{i,k}$ and $\vphi_{i,k'}$ may not be independent for any $k\neq k'$,  we cannot directly apply \Cref{Small Ball Method_original} to study stable embeddings via $\calA^Q$. To address this issue, we modify Mendelson's small ball method as follows.
\begin{lemma}
\label{Small Ball Method_abs_New}
Consider a fixed set $E\subset \C^{D}$. Let $\{\vb_{1},\!\dots,\!\vb_K\!\}$ represent a collection of random columns in $\C^{D}$, which may not be mutually independent. Additionally, let $\{\vb_{i,1},\dots,\vb_{i,K}\}_{i=1}^Q$ denote a set of independent copies of $\{\vb_{1},\dots,\vb_K\}$.
Introduce the marginal tail function
\begin{eqnarray}
    \label{marginal_tail_function_New}
    H_{\xi}(E;\vb) = \inf_{\vu\in E} \frac{1}{K}\sum_{k=1}^K\mathbb{P} \{|\<\vb_{k},\vu \>|\geq \xi  \},\ for \ \xi>0.
\end{eqnarray}
Let $\epsilon_i,i=1,\dots,Q$ be independent Rademacher random variables, independent from everything else, and define the
mean empirical width of the set:
\begin{align}
    \label{mean empirical width_New}
    W_{QK}(E;\vb) = \E \sup_{\vu\in E}\<\vh ,\vu\>, \text{where} \ \vh =\frac{1}{\sqrt{QK}}\sum_{i=1}^Q\sum_{k=1}^K \epsilon_{i}\vb_{i,k}.
\end{align}

Then, for any $\xi >0$ and $t >0$
\begin{eqnarray}
    \label{Final_Conclusion_New}
    \hspace{-0.5cm}\inf_{\vu\in E}\bigg(\sum_{i=1}^Q\sum_{k=1}^K|\<\vb_{i,k}, \vu  \>|^2  \bigg)^{\frac{1}{2}} &\!\!\!\!\geq\!\!\!\!& \xi\sqrt{QK}H_{\xi}(E;\vb)\nonumber\\
    &\!\!\!\!\!\!\!\!&\hspace{-0.5cm} -2W_{QK}(E;\vb) -t \xi\sqrt{K},
\end{eqnarray}
with probability at least $1-e^{-\frac{t^2}{2}}$.

\end{lemma}

The proof has been provided in {Appendix} \ref{Proof_Lemma_abs_small ball new}. Note that when $K = 1$, \Cref{Small Ball Method_abs_New} reduces to \Cref{Small Ball Method_original} (by setting $Q = K$ in \Cref{Small Ball Method_abs_New}). In other words, $Q$ plays the same role as $K$ in \Cref{Small Ball Method_original}.
To effectively apply the modified method, we need to generalize the linear map in \eqref{eq:linear map A}. With \Cref{Small Ball Method_abs_New}, we now establish the stable embedding of \eqref{The defi of population measurement in Q cases (K measurements)} in the following theorem.

\begin{theorem} [Stable embedding of multiple Haar random projective measurements]
\label{Small_Method_ROHaar_Measurement}
Let $\calA^{Q}:\C^{d^n\times d^n} \rightarrow \R^{KQ}$ be the linear mapping defined in \eqref{The defi of population measurement in Q cases (K measurements)} that is induced by $Q$ random unitary matrices. For any $K\ge 1$, assuming
\begin{eqnarray}
    \label{number of sample rank one haar}
    Q\gtrsim nd^2\ol{r}^2 \log n, \ \ol{r}=\max_{i=1,\dots n-1}r_i,
\end{eqnarray}
then with probability at least $1-e^{-\alpha_2Q}$ (where $\alpha_2$ is a positive constant.), $\calA^Q$ obeys
\begin{eqnarray}
    \label{Small_Method_ROHaar_Measurement_Conclusion_Theorem}
    \|\calA^Q(\vrho)\|_2  = \bigg(\sum_{i=1}^Q\sum_{k=1}^K |\<\vphi_{i,k} \vphi_{i,k}^\dagger, \vrho  \>|^2 \bigg)^{\frac{1}{2}}  \gtrsim\frac{\sqrt{QK}}{d^n}
\end{eqnarray}
for any $\vrho\in \ol\setX_{\ol{r}}$.
\end{theorem}

The proof is given in {Appendix} \ref{ProofOfSmallBallMethod_For_Haar_Measurement}. First note that in \Cref{Small_Method_ROHaar_Measurement}, the requirement on $Q$ in \eqref{number of sample rank one haar} and the failure probability $e^{-\alpha_3 Q}$ are similar to those in \Cref{Small_Method_ROSubgaussian_Measurement} on $K$. This is because, as we explained before, $Q$ in \Cref{Small Ball Method_abs_New} plays the same role as $K$ in \Cref{Small Ball Method_original}, and likewise $Q$ in \Cref{Small_Method_ROHaar_Measurement} is equivalent to $K$ in \eqref{number of sample rank one gaussian}. Thus, \Cref{Small_Method_ROHaar_Measurement} holds for any $K\ge 1$. On the other hand, without exploiting the randomness between different columns within a random unitary matrix, the number of POVMs $Q$ is required to be relatively large as stated in \eqref{number of sample rank one haar}. Considering that the local correlations between the columns in the unitary matrix are very weak because the orthogonality is a global property \cite{tropp2012comparison}, we conjecture that the requirement on $Q$ can be significantly reduced, even to $Q=1$. Indeed, according to \cite[Theorem 3]{jiang2006many}, when $n\to \infty$, in an ``in probability"
sense, all elements (scaled by $\sqrt{d^n}$) of $o(\frac{d^{n}}{n\log d})$ columns in a Haar-distributed random unitary matrix can be approximated by entries generated independently from  a standard complex normal distribution. As $o(\frac{d^{n}}{n\log d})$ independent columns from a multivariate complex normal distribution are sufficient for \Cref{Small_Method_ROSubgaussian_Measurement}, this suggests that it is highly possible to ensure stable embedding \eqref{Small_Method_ROHaar_Measurement_Conclusion_Theorem} with a single POVM $Q = 1$. While we leave a formal analysis as future work, we conduct a numerical experiment to support this conjecture. Set $d=2,Q =1, K = d^n, r_1=\cdots = r_{n-1} =2$. Then for each $n$, we randomly generate a unitary matrix (i.e., $Q = 1$), randomly sample many  MPOs $\vrho$  with $\|\vrho\|_F=1$, and compute the minimum of $\|\calA^Q(\vrho)\|_2$ among all the generated MPOs.
In Figure~\ref{Test of Theorem 4}, we  compare the minimum of $\|\calA^Q(\vrho)\|_2$ (averaged over $50$ Monte Carlo trails) with $\frac{1}{\sqrt{d^n}}$. We observe that $\|\calA^Q(\vrho)\|_2$ is of the same order as $\frac{1}{\sqrt{d^n}}$. Furthermore, as the number of qudits increases, $\|\calA^Q(\vrho)\|_2$ approaches $\frac{1}{\sqrt{d^n}}$. This is consistent with \eqref{Small_Method_ROHaar_Measurement_Conclusion_Theorem}, where the right hand side becomes $\frac{1}{\sqrt{d^n}}$ when $K = d^n, Q = 1$.

\begin{figure}[thbp]
\centering
\includegraphics[width=8.5cm, keepaspectratio]%
{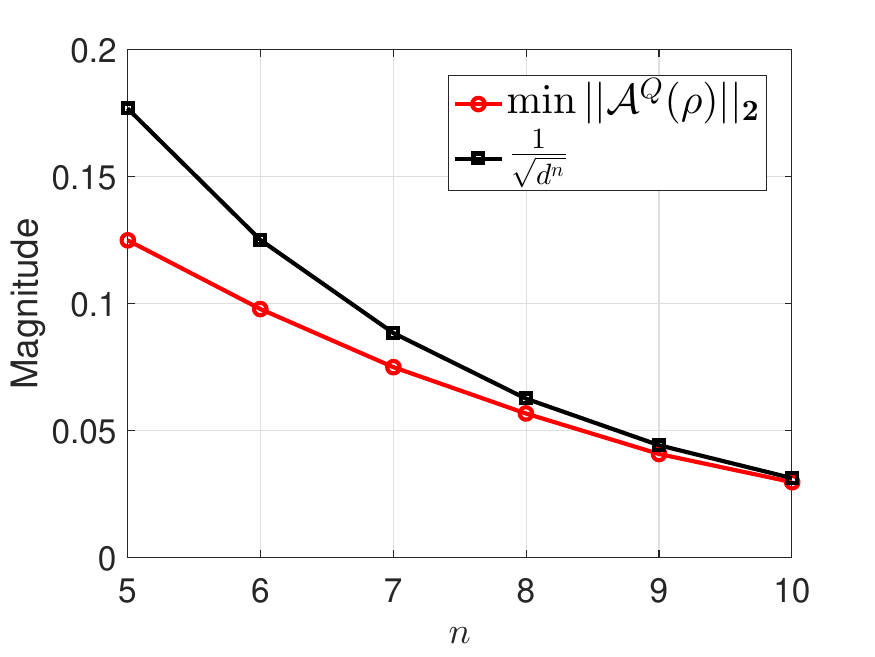}
\caption{Numerical computation of $\min_{\vrho\in\ol{\mathbb{X}}_{\ol{r}}}\|\calA^Q(\vrho)\|_2$ with $Q = 1$ and $K = d^n$.}
\label{Test of Theorem 4}
\end{figure}

\section{Stable recovery with empirical measurements}
\label{sec:stable recovery}

The results of Section~\ref{sec:rankOnePOVMpop} ensure a distinct set of population measurements $\calA^Q(\vrho)$ for any ground-truth MPO $\vrho^\star$ under multiple Haar random projective measurements. Based on these results, in this section, we study the stable recovery of $\vrho$ from empirical measurements obtained by multiple Haar random  projective measurements.
With $Q$ randomly generated Haar-distributed unitary matrices $\begin{bmatrix}\vphi_{i,1} & \cdots & \vphi_{i,d^n}\end{bmatrix}, \ i = 1,\ldots, Q$, according to \eqref{The defi of population measurement in Q cases (K measurements)}, we can generate $Qd^n$ population measurements through the linear measurement operator $ \calA^Q:\C^{d^n\times d^n} \rightarrow \R^{Qd^n}$ (set $K = d^n$ in \eqref{The defi of population measurement in Q cases (K measurements)}) as
\begin{eqnarray}
\label{The defi of population measurement in Q cases sec:4}
{{\bm p}^Q}= \calA^Q(\vrho^\star) = \begin{bmatrix}
          {\bm p}_{1} \\
          \vdots \\
          {\bm p}_{Q}
        \end{bmatrix} = \begin{bmatrix}
          \calA_1(\vrho^\star) \\
          \vdots \\
          \calA_Q(\vrho^\star)
        \end{bmatrix},
\end{eqnarray}
where $\calA_i$ is as defined in~\eqref{The defi of POVM element in Q cases (K measurements)} with $K = d^n$ and with $\mA_{i,k} = \vphi_{i,k} \vphi_{i,k}^\dagger$.
Denote by $p_{i,k}$ the $k$-th element in $\vp_i$.

For each POVM, suppose we repeat the measurement process $M$ times and take the average of the outcomes to generate empirical probabilities
\begin{align}
\wh p_{i,k} = \frac{f_{i,k}}{M}, \ i=1,\dots,Q, \ k=1,\dots,d^n,
\label{eq:empirical-prob in Q cases sec:4}\end{align}
where $f_{i,k}$ denotes the number of times the $k$-th output is observed when using the $i$-th POVM $M$ times. Denote by ${\widehat{\bm p}_i}= \begin{bmatrix}
          \widehat{p}_{i,1} & \cdots &
          \widehat{p}_{i,d^n}
        \end{bmatrix}^\top$ the empirical measurements obtained by the $i$-th POVM and stack all the total empirical measurements together as
${\widehat{\bm p}^Q}= \begin{bmatrix}
          \widehat{\bm p}_{1}^\top & \cdots &
          \widehat{\bm p}_{Q}^\top
        \end{bmatrix}^\top$, which are unbiased estimators of the population measurements $\vp^Q$.
We denote by $\veta$ the measurement error as
\begin{eqnarray}
    \label{Noise_Haar Measurement restate1}
    \veta = \wh{\vp}^Q -  \vp^Q = \wh{\vp}^Q - \calA^Q(\vrho^\star) =  \begin{bmatrix}
          \veta_{1}^\top,
          \cdots,
          \veta_{Q}^\top\end{bmatrix}^\top,
\end{eqnarray}
where $\eta_{i,k}$ is the $k$-th element in $\veta_i$.

With empirical measurements $\wh \vp^Q$, for simplicity, we consider minimizing the following constrained least squares objective:
\begin{eqnarray}
    \label{The loss function in QST}
    \wh{\vrho} = \argmin_{\vrho\in\setX_{\ol{r}}}\|\calA^Q(\vrho) - {\widehat{\bm p}^Q}\|_2^2,
\end{eqnarray}
where $\calA^Q$ is the induced linear map as defined in \eqref{The defi of population measurement in Q cases sec:4}. Supposing one can find a  global solution of \eqref{The loss function in QST}, our goal is to study how the recovery error $\|\wh{\vrho}  -  \vrho^\star\|_F$ scales with the size of the MPO (particularly with respect to the number of qudits $n$) and the total number of measurements $QM$. To enable a stable estimate of the state by measuring it only a polynomial number of times in terms of $n$, we desire the recovery error to grow only polynomially rather than exponentially in $n$.

\subsection{Challenge: Abundant but extremely noisy measurements}
\label{sec:challenge}
Before presenting the main result, we first discuss the challenge and hope of obtaining a recovery error that only grows polynomially in terms of the number of qudits. Recall that $(f_{i,1},\ldots,f_{i,d^n})$ in \eqref{eq:empirical-prob in Q cases sec:4} follows a multinomial distribution $\operatorname{Multinomial}(M, \vp_i)$  with parameters $M$ and $\vp_i$. Thus, $\wh\vp_i - \vp_i$ has mean zero and covariance matrix $\mSigma_i$, where $\mSigma_i$ has elements given by ${\mSigma_i}[l,j] = \begin{cases} \frac{p_{i,l}(1-p_{i,l})}{M}, & l=j \\ -\frac{p_{i,l}p_{i,j}}{M}, & l\neq j \end{cases}$. With this observation, we have
\begin{align}
    \label{Expectation of noise square}
  \E \|\veta\|_2^2 =  \E \bigg[\sum_{i=1}^Q\sum_{k=1}^{d^n}  \eta_{i,k}^2\bigg]  = \sum_{i=1}^Q\sum_{k=1}^{d^n} \frac{p_{i,k} (1 - p_{i,k})}{M}\leq \frac{Q}{M}.
\end{align}
Note that $\frac{p_{i,k} (1 - p_{i,k})}{M}$ could be as small as 0 which can be achieved, although rarely, when $p_{i,k}\in\{0,1\}$ (i.e., when $\{p_{i,k},k =1,\ldots,d^n\}$ has a spiky distribution). However, the above bound on the order of $\frac{Q}{M}$ is tight when the distribution of $\{p_{i,k},k =1,\ldots,d^n\}$ is not spiky (e.g., when each $p_{i,k}$ is on the order of $\frac{1}{d^n}$). To see this, denote the eigenvalue decomposition of $\vrho^\star$ as $\vrho^\star=\sum_{i=1}^{d^n}\lambda_i \vu_i \vu_i^\dagger$, where $\sum_{i=1}^{d^n} \lambda_i=1$.
Now for any $i$ and $k$, we can compute $\E[p_{i,k}^2]$ as
\begin{eqnarray}
    \label{Unitary Distrubuted_2 another}
    &\!\!\!\!\!\!\!\!&\E[|\< {\vphi_{i,k}}{\vphi^\dagger_{i,k}},\vrho^\star \>|^2]\nonumber\\
    &\!\!\!\!=\!\!\!\!& \sum_{j=1}^{d^n}\sum_{l=1}^{d^n}\lambda_j\lambda_l \E[|{\vphi^\dagger_{i,k}} \vu_j|^2|{\vphi^\dagger_{i,k}} \vu_l|^2]\nonumber\\
    &\!\!\!\!=\!\!\!\!& \sum_{l\neq j}\lambda_j\lambda_l(\E[| {\vphi_{i,k}}[1]|^2])^2 + \sum_{l} \lambda_l^2\E[| {\vphi_{i,k}}[1]|^4]\nonumber\\
    &\!\!\!\!=\!\!\!\!& \sum_{l\neq j}\frac{\lambda_j\lambda_l}{d^{2n}} + 2\sum_{l}\frac{\lambda_l^2}{d^n(d^n+1)}\nonumber\\
    &\!\!\!\!=\!\!\!\!& \sum_{j=1}^{d^n}\sum_{l=1}^{d^n}\frac{\lambda_j\lambda_l}{d^{2n}} + \sum_{l=1}^{d^n}\frac{d^n-1}{d^{2n}(d^n+1)}\lambda_l^2\nonumber\\
    &\!\!\!\! = \!\!\!\!&\frac{1}{d^{2n}} + \frac{d^n-1}{d^{2n}(d^n+1)}\|\vrho^\star\|_F^2,
\end{eqnarray}
where ${\vphi_{i,k}}[1]$ is the first element of ${\vphi_{i,k}}$,  the second line utilizes the rotation invariance of the unitary matrix in \Cref{Stastitic property of UV}, and the third line uses \Cref{Unitary_matrix_Expectation}.

Noting that $\|\vrho^\star\|_F^2 \le (\sum_{i=1}^{d^n} \lambda_i)^2 = 1$, we further have
\e\label{eq:2nd-moment-pk}
\frac{1}{d^{2n}} \le \E[p_{i,k}^2] \le \frac{2}{d^{2n}}, \ \forall 1\le i\le Q \ \ \text{and} \ \ \forall 1\le k\le d^n.
\ee
In other words, if $\begin{bmatrix}\vphi_{i,1} & \cdots & \vphi_{i,d^n}\end{bmatrix}$ is a randomly generated unitary matrix, then each $p_{i,k}$ has the same second moment of order $1/d^{2n}$. This suggests that the distribution of $\{p_{i,k},k =1,\ldots,d^n\}$ is more uniform than spiky.

In addition, \eqref{eq:2nd-moment-pk} also gives the energy of the clean measurements or population measurements as
\begin{align}
    \label{Expectation of clean energy square}
    \frac{Q}{d^n} \le \E  \bigg[\sum_{i=1}^Q\sum_{k=1}^{d^n}  p_{i,k}^2\bigg]  = \sum_{i=1}^Q\sum_{k=1}^{d^n}\E \< \vphi_{i,k}\vphi_{i,k}^\dagger, \vrho^\star  \>^2\leq \frac{2Q}{d^n}.
\end{align}
To summarize, the above discussion gives the following comparison between the energy of the clean measurements and the noise in the measurements:
\begin{align*}
\text{Clean measurements: } & \E{\|\vp^Q\|_2^2} = O\parans{ \frac{Q}{d^n} },\\
\text{Statistical error: } & \E \|\veta\|_2^2 =  O\left(\frac{Q}{M}\right),
\end{align*}
which indicates that the statistical error or measurement noise is exponentially larger than the clean measurements. This seems to suggest that $M$ has to be on the order of $d^n$ to obtain measurements with suitable signal-to-noise ratio for stable recovery.

Fortunately, though each measurement could be extremely noisy, we have an exponentially large number of such measurements $\{\wh p_{i,1},\ldots,\wh p_{i,d^n}\}_i$, from $Q$ POVMs.
This setting is slightly different from some common inverse problems \cite{beck2009fast,chandrasekaran2012convex,tropp2015convex}, where the number of measurements matches the number of degrees of freedom behind the underlying signal but the measurements are not overwhelmed by noise. In addition, conditioned on the selected POVM, the measurement noise $\veta$ is random and behaves close to a multivariate Gaussian distribution~\cite{carter2002deficiency,muirhead2009aspects,ouimet2021precise}. By exploiting these observations together with the stable embeddings established in the last section, we anticipate stable recovery even when $M$ is only polynomially large in $n$.

\subsection{Stable recovery with empirical measurements}
We now provide a formal analysis of the recovery error $\|\wh{\vrho}  -  \vrho^\star\|_F$, where  $\wh{\vrho}$ is a global solution of \eqref{The loss function in QST}. Using \eqref{The loss function in QST} and the fact that $\vrho^\star \in \setX_{\bar{r}}$, we have
\begin{eqnarray}
    \label{whrho and rho star relationship}
    \hspace{-0.1cm}0 &\!\!\!\!\!\leq\!\!\!\!\!& \|\calA^Q(\vrho^\star) - \wh{\vp}^Q \|_2^2  - \|\calA^Q(\wh{\vrho}) - \wh{\vp}^Q \|_2^2\nonumber\\
&\!\!\!\!\!=\!\!\!\!\!&\|\calA^Q(\vrho^\star)-\calA^Q(\vrho^\star) - \veta\|_2^2 - \|\calA^Q(\wh{\vrho})-\calA^Q(\vrho^\star) - \veta\|_2^2\nonumber\\
&\!\!\!\!\!=\!\!\!\!\!& 2\<\calA^Q(\vrho^\star)+\veta, \calA^Q(\wh{\vrho} - \vrho^\star)  \> + \|\calA^Q(\vrho^\star)\|_2^2 - \|\calA^Q(\wh{\vrho})\|_2^2\nonumber\\
&\!\!\!\!\!=\!\!\!\!\!& 2\<  \veta, \calA^Q(\wh{\vrho} - \vrho^\star) \> - \|\calA^Q(\wh{\vrho} - \vrho^\star)\|_2^2,
\end{eqnarray}
which further implies that
\begin{eqnarray}
    \label{whrho and rho^star relationship_1}
    \|\calA^Q(\wh{\vrho} - \vrho^\star)\|_2^2 \leq 2\<  \veta, \calA^Q(\wh{\vrho} - \vrho^\star) \>.
\end{eqnarray}
The left-hand side of the above equation can be further lower bounded by order of $\frac{Q}{d^n} \|\wh{\vrho} - \vrho^\star\|_F^2$ according to \Cref{Small_Method_ROHaar_Measurement}. The challenging part is to deal with the right-hand side of \eqref{whrho and rho^star relationship_1}. A simple Cauchy–Schwarz inequality  $\<  \veta, \calA^Q(\wh{\vrho} - \vrho^\star) \> \le \|\veta\|_2 \cdot \|\calA^Q(\wh{\vrho} - \vrho^\star)\|_2$ is insufficient to provide a tight result since $\|\veta\|_2$ scales as $\frac{1}{\sqrt{M}}$ as discussed after \eqref{Expectation of noise square}.
Instead, we exploit the randomness of $\veta$ and use the following concentration bound for multinomial random variables, which is proved in \Cref{General bound of multinomial distribution Q cases} of {Appendix} \ref{sec:app-aux} and is derived based on \cite{kawaguchi2022robustness}.

\begin{lemma}
\label{General bound of multinomial distribution Q cases1}
Suppose $\{(f_{i,k},\dots, f_{i,K})\},i=1,\dots,Q$ are mutually independent and follow the multinomial distribution \\ $\operatorname{Multinomial}(M,\vp_i)$  where $\sum_{k=1}^{K}f_{i,k} =M $ and $\vp_i = [p_{i,1}, \cdots, p_{i,K}]$.
Let $a_{i,1},\dots, a_{i,K}$ be fixed. Then, for any $t>0$,
\begin{align}
    \label{General bound of multinomial distribution for all constant Q cases1}
    &\P{\sum_{i=1}^Q\sum_{k=1}^Ka_{i,k}(\frac{f_{i,k}}{M} - p_{i,k}) > t   }\nonumber\\
    &\leq  e^{-\frac{Mt}{4a_{\max}}\! \min\bigg\{\! 1, \frac{a_{\max}t }{4\sum_{i=1}^Q\sum_{k=1}^K a_{i,k}^2p_{i,k}} \! \bigg\}}\! + \! e^{-\frac{Mt^2 }{8\sum_{i=1}^Q\sum_{k=1}^K a_{i,k}^2p_{i,k}}},
\end{align}
where $a_{\max} = \max_{i,k}|a_{i,k}|$.
\end{lemma}
One may not be able to directly apply the above result for
$\<  \veta, \calA^Q(\wh{\vrho} - \vrho^\star) \>$ since $\wh \vrho$ depends on $\veta$. We address this issue by using the covering argument to bound $\<  \veta, \calA^Q({\vrho} - \vrho^\star) \>$ for all possible $\vrho$. We refer to {Appendix} \ref{Proof of Statistical error in Haar} for the detailed analysis. We now summarize the main result as follows.
\begin{theorem}
\label{Statistical Error_Haar_Measurement}
Given an MPO state $\vrho^\star\in\C^{d^n \times d^n}$ of the form \eqref{DefOfMPO} with MPO ranks $\vr$, independently generate $Q$ Haar-distributed random unitary matrices $\begin{bmatrix}  \vphi_{i,1} & \cdots & \vphi_{i,d^n}\end{bmatrix}, i=1,\dots,Q$. Use each induced rank-one POVM $\{\vphi_{i,k}\vphi_{i,k}^\dagger\}_{k=1}^{d^n}$ to measure the state $M$ times and get the empirical measurements $\wh \vp_i$. For any $\epsilon>0$, suppose  $Q\gtrsim nd^2\ol{r}^2(\log n)$ and
\begin{align}
    \label{Statistical Error_Haar_Measurement_Conclusion_Theorem_M}
    QM \gtrsim \frac{n d^2\ol{r}^2\log n (\log Q+n\log d)^2}{\epsilon^2}, \! \quad \ol{r}=\max_{i=1,\dots n-1}r_i.
\end{align}
Then any global solution $\wh \vrho$ of \eqref{The loss function in QST} satisfies
\begin{eqnarray}
    \label{Statistical Error_Haar_Measurement_Conclusion_Theorem}
    \|\wh{\vrho}-\vrho^\star\|_F \le \epsilon
\end{eqnarray}
with probability at least $\min\{1-e^{-\alpha_3 (\log Q + n\log d)} - e^{-\alpha_4 \nqbit d^2 \ol{r}^2 \log n}, 1-e^{-\alpha_2Q} \}$, where $\alpha_3 $ and $\alpha_4$ are positive constants,  $\alpha_2$  corresponds to constants of the probability in \Cref{Small_Method_ROHaar_Measurement}.
\end{theorem}

\Cref{Statistical Error_Haar_Measurement} ensures a stable recovery of the ground-truth state when the total number of state copies $QM$ only grows polynomially ($n^3$) in terms of the number of qudits $n$, as the order specified in \eqref{Statistical Error_Haar_Measurement_Conclusion_Theorem_M}. If we ignore the $\log Q$ term, which exists due a to proof artifact and which we conjecture can be removed, then \eqref{Statistical Error_Haar_Measurement_Conclusion_Theorem_M} only requires $QM$ to be sufficiently large, without any requirement on the number of measuring times $M$ for each POVM. In other words,  \Cref{Statistical Error_Haar_Measurement} provides theoretical support for the practical use of
single-shot measurements (i.e., $M = 1$ where each POVM is measured only once) that are used in \cite{wang2020scalable,huang2020predicting}. Note that the orders of the polynomial in \eqref{Statistical Error_Haar_Measurement_Conclusion_Theorem_M}, particularly in terms of $n$, are fairly large compared to the number $O(nd^2\ol{r}^2)$ of degrees of freedom of the MPO and may not be optimal. For this reason, we conjecture that the bound in \eqref{Statistical Error_Haar_Measurement_Conclusion_Theorem_M} could be further improved, such as by removing the term $(\log Q+n\log d)^2$ that extends the bound beyond the number $O(nd^2\ol{r}^2)$ of degrees of freedom of the MPO. We refer to \Cref{sec: Discussion and conclusion} for additional detailed discussion.

The requirement $Q\gtrsim nd^2\ol{r}^2\log n$ and failure probability $e^{-\alpha_2Q}$ are inherited from \Cref{Small_Method_ROHaar_Measurement} for a stable embedding via the $Q$ POVMs $\{\vphi_{i,k} \vphi_{i,k}^\dagger\}$. As discussed right after \Cref{Small_Method_ROHaar_Measurement}, we conjecture that \Cref{Small_Method_ROHaar_Measurement} holds with $Q = 1$ by setting $K = d^n$. If this is the case, then the requirement $Q\gtrsim nd^2\ol{r}^2\log n$ can also be dropped, and \Cref{Statistical Error_Haar_Measurement} would also hold for $Q =1$. In the next section, we will use experiments to demonstrate that a single POVM is sufficient to stably recover  $\vrho^\star$.

Based on the stable embedding results in \Cref{Small_Method_ROHaar_Measurement}, the recovery guarantee \eqref{Statistical Error_Haar_Measurement_Conclusion_Theorem} is established in the Frobenius norm instead of the trace norm.
However, if the MPO state $\vrho^\star$ has a further low matrix rank, we can also establish recovery guarantee in the trace norm by using a bound between trace distance and Hilbert-Schmidt distance for low-rank states \cite{coles2019strong}, i.e.,  $\|\wh\vrho - \vrho^\star\|_1\leq 2\sqrt{\text{rank}(\vrho^\star)}\| \wh\vrho - \vrho^\star \|_F$, which is also used in \cite{haah2017sample} for obtaining guarantees in trace norm for low-rank states. While this approach provides a vacuous bound when the matrix rank is high, we conjecture \Cref{Statistical Error_Haar_Measurement}  can be extended for the trace norm, regardless of the matrix rank of $\vrho^\star$, by directly analyzing it, but we leave this as future work.

Finally, we note that while the solution $\wh{\vrho}$ of \eqref{The loss function in QST} may be non-physical, we can impose additional constraints to obtain a physical state without sacrificing the recovery guarantee in \Cref{Statistical Error_Haar_Measurement}. Specifically, let $\vrho^\diamond$ be the global solution to the following minimization problem with additional PSD and trace constraints:
\begin{eqnarray}
    \label{The loss function in QST-PSD}
    \vrho^\diamond = \argmin_{\vrho\in\setX_{\ol{r}}, \vrho\succeq \vzero, \trace(\vrho) = 1}\|\calA^Q(\vrho) - {\widehat{\bm p}^Q}\|_2^2.
\end{eqnarray}
Then $\vrho^\diamond$ has the same guarantee as $\wh \vrho$ under the same setup as \Cref{Statistical Error_Haar_Measurement}. Alternatively, we can simply project $\wh \vrho$ onto the set of physical states $\setS_+:=\{\vrho\in\C^{d^n\times d^n}: \vrho \succeq \vzero, \trace(\vrho) = 1\}$. Denote by $P_{\setS_+}$ the projection onto the set $\setS_+$, which can be efficiently computed by projecting the eigenvalues onto a simplex \cite{condat2016fast}. Since the set $\setS_+$ is convex, the corresponding projector is non-expansive and hence
\[
\|P_{\setS_+}(\wh{\vrho} )-\vrho^\star\|_F = \|P_{\setS_+}(\wh{\vrho} )-P_{\setS_+}(\vrho^\star)\|_F \le \|\wh{\vrho}-\vrho^\star\|_F \le \epsilon,
\]
which implies that the projection step ensures that the state becomes physically valid while preserving or even improving the recovery guarantee.

\section{Numerical Experiments}
\label{Numerical Experiments}

In this section, we perform numerical experiments on quantum state tomography for MPOs to illustrate our theoretical results. Due to computational constraints, we conduct experiments on real matrix product states (MPSs, which are pure states) of the form $\vrho^\star = \vu^\star {\vu^\star }^\top$, where $\vu^\star \in \R^{d^n\times 1}$ satisfies $\|\vu^\star\|_2=1$ and its $(i_1\cdots i_n)$-element can be represented in a matrix product form similar to the MPO form in~\eqref{DefOfMPO}:
\[
\vu^\star(i_1\cdots i_n) = {\mU_1^\star}^{i_1}\cdots {\mU_n^\star}^{i_n}.
\]
Here, each matrix ${\mU_{\ell}^\star}^{i_{\ell}}$ has size $r \times r$, except for ${\mU_1^\star}^{i_1}$ and ${\mU_n^\star}^{i_n}$ that have dimension of $1 \times r$ and $r \times 1$, respectively. We generate each MPS $\vu^\star$ by first generating a random Gaussian vector of length $d^n$ and then applying the sequential SVD \cite{Oseledets11} to truncate it to an MPS, which we finally normalize to have unit length.
 As a consequence, entry $\vrho^\star(i_1 \cdots i_\nqbit, j_1 \cdots j_\nqbit)$ can be expressed as
\begin{eqnarray}
    &\!\!\!\!\!\!\!\!&\vrho^\star(i_1 \cdots i_\nqbit, j_1 \cdots j_\nqbit)\nonumber\\
    &\!\!\!\!=\!\!\!\!&  {\mU_1^\star}^{i_1}\cdots {\mU_n^\star}^{i_n}  {\mU_1^\star}^{j_1}\cdots {\mU_n^\star}^{j_n} \nonumber\\
    &\!\!\!\!=\!\!\!\!& ({\mU_1^\star}^{i_1}\cdots {\mU_n^\star}^{i_n}) \otimes ( {\mU_1^\star}^{j_1}\cdots {\mU_n^\star}^{j_n} )\nonumber\\
    &\!\!\!\!=\!\!\!\!& (\underbrace{{\mU_1^\star}^{i_1}\otimes {\mU_1^\star}^{j_1}}_{{\mX_1^\star}^{i_1,j_1}}) \cdots (\underbrace{{\mU_n^\star}^{i_n}\otimes {\mU_n^\star}^{j_n}}_{{\mX_\nqbit^\star}^{i_\nqbit,j_\nqbit}}),\nonumber
\end{eqnarray}
where $\otimes$ denotes the Kronecker product. Thus, $\vrho^\star = \vu^\star {\vu^\star }^\top$ is also an MPO with MPO ranks $r_1=\dots=r_{n-1} = r^2$.

To illustrate that \Cref{Statistical Error_Haar_Measurement} might hold even with $Q=1$, we only use a single Haar random  projective in the experiments.
We generate a real Haar-distributed random unitary matrix $\begin{bmatrix}  \vphi_{1} & \cdots & \vphi_{d^n}\end{bmatrix} \in \R^{d^n \times d^n}$. Each population measurement \eqref{The defi of POVM 2} can then be rewritten as
\[
p_k = \trace(\vphi_{k}\vphi_k^\top \vrho^\star) = \abs{\vphi_k^\top \vu^\star}^2.
\]
This is our reason for considering a pure state $\vrho^\star$ as it reduces the complexity for computing $ \trace(\vphi_{k}\vphi_k^\top \vrho^\star)$ from $O(d^{2n})$ to $O(d^{n})$.

We use the induced POVM to measure the state $\vrho^\star$ $M$ times to get the empirical measurements $\wh \vp$. With the obtained measurements, as in \eqref{The loss function in QST}, we attempt to recover the MPS $\vu^\star$ (and hence $\vrho^\star$) by minimizing the following constrained mean squared error loss
\begin{eqnarray}\begin{split}
    \wh\vu &= \argmin_{\vu \in \setU_r} \frac{1}{2}\sum_{i=1}^{d^n}(|\vphi_i^\top\vu|^2 - \wh p_i)^2,\\
   \setU_r &= \bigg\{\vu\in\R^{d^n}: \vu(i_1\cdots i_n) = {\mU_1}^{i_1}\cdots {\mU_n}^{i_n}, \\
    &\mU_1^{i_1} \in \R^{1\times r}, \mU_n^{i_n} \in \R^{r\times 1}, \mU_\ell^{i_\ell} \in \R^{r\times r}, 1<i<n\bigg\},
\end{split} \label{Loss function of MPS}\end{eqnarray}
which has the same form as \eqref{The loss function in QST}.

As in \cite{Rauhut17}, we solve \eqref{Loss function of MPS} by the following iterative hard thresholding (IHT, i.e., projected gradient descent):
\begin{eqnarray}
    \label{Iterative equ of GDwithTT_SVD_1}
    {\vu}_{t+1} =\calP_{\setU_r}\parans{ \vu_t - \mu \sum_{i=1}^{d^n} (|\vphi_i^\top\vu_t|^2 - \wh p_i) \vphi_i\vphi_i^\top\vu_t },
\end{eqnarray}
where $\mu$ is the step size and
$\calP_{\setU_r}$ denotes the projection onto the MPS set $\setU_r$, which can be approximately computed via a sequential SVD algorithm \cite{Oseledets11}. We adopt this approach for computing an approximate projection in the following experiments.

Since our goal is to verify how the global solution $\wh \vu$ behaves, to ensure the convergence to a global solution, we use a good initialization $\vu_0 = \frac{\vu^\star + \lambda\vv}{\|\vu^\star + \lambda\vv\|_2}$ where $\vv$ is randomly generated from the unit sphere of $\R^{d^n}$. In all the experiments, we set $\lambda = 0.7$ so that the initialization $\vu_0$ is still not very close to the ground truth $\vu^\star$. Since the gradient becomes exponentially small in $n$, which can be observed by using the same argument in \eqref{Expectation of clean energy square} for
$\vphi_i^\top\vu_t$, we set the step size $\mu=0.01\times d^n$. The solution $\wh{\vu}$ is obtained by running the IHT algorithm \eqref{Iterative equ of GDwithTT_SVD_1} until convergence. Since the factorization $\vrho^\star = \vu^\star {\vu^\star }^\top$ is not unique as $\vrho^\star = (-\vu^\star) (-\vu^\star)^\top$ also holds,
we measure the quality of the recovered $\wh \vu$ by the following distance 
\begin{eqnarray}
    \label{THE definition of MSE}
    \dist(\wh\vu,\vu^\star) = \min\left\{\|\wh{\vu} - \vu^\star\|_2^2, \|\wh{\vu} + \vu^\star\|_2^2 \right\}.
\end{eqnarray}
For each experiment, we conduct $10$ Monte Carlo trials and take the average recovery distance over the 10 trials.

\begin{figure}[!htbp]
\centering
\subfigure[]{
\begin{minipage}[t]{0.45\textwidth}
\centering
\includegraphics[width=8.5cm]{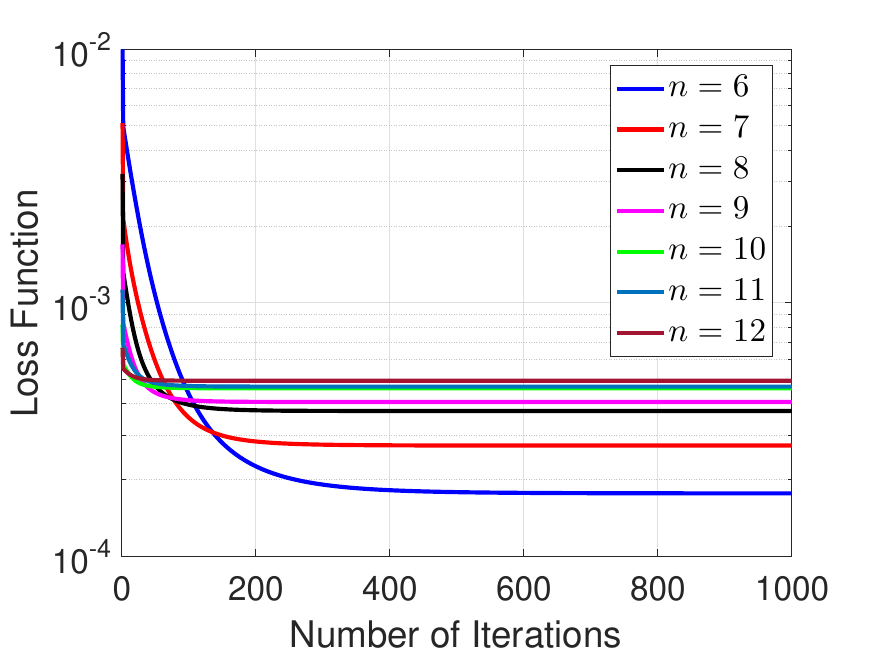}
\end{minipage}
\label{Performance Comparison_M_loss}
}
\subfigure[]{
\begin{minipage}[t]{0.45\textwidth}
\centering
\includegraphics[width=8.5cm]{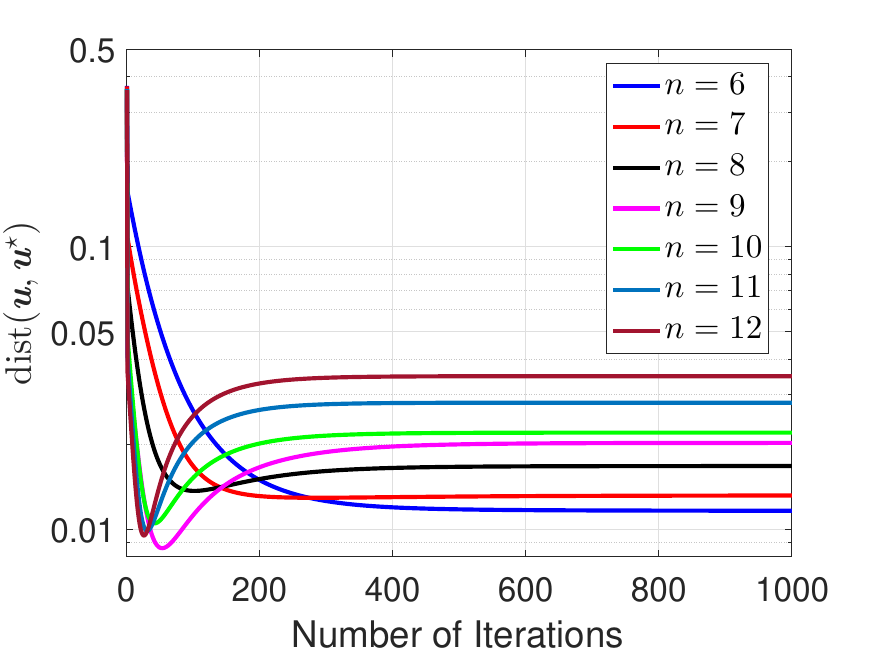}
\end{minipage}
\label{different m2}
}
\caption{Illustration of convergence of IHT in \eqref{Iterative equ of GDwithTT_SVD_1} in terms of (a)  loss function defined in \eqref{Loss function of MPS}, and (b) recovery error defined in \eqref{THE definition of MSE} for different $\nqbit$ with $M=1000$, $r=2$, and $d=2$.  }
\label{Performance Comparison_Convergence}
\end{figure}

\begin{figure}[htbp]
\centering
\includegraphics[width=8.5cm, keepaspectratio]%
{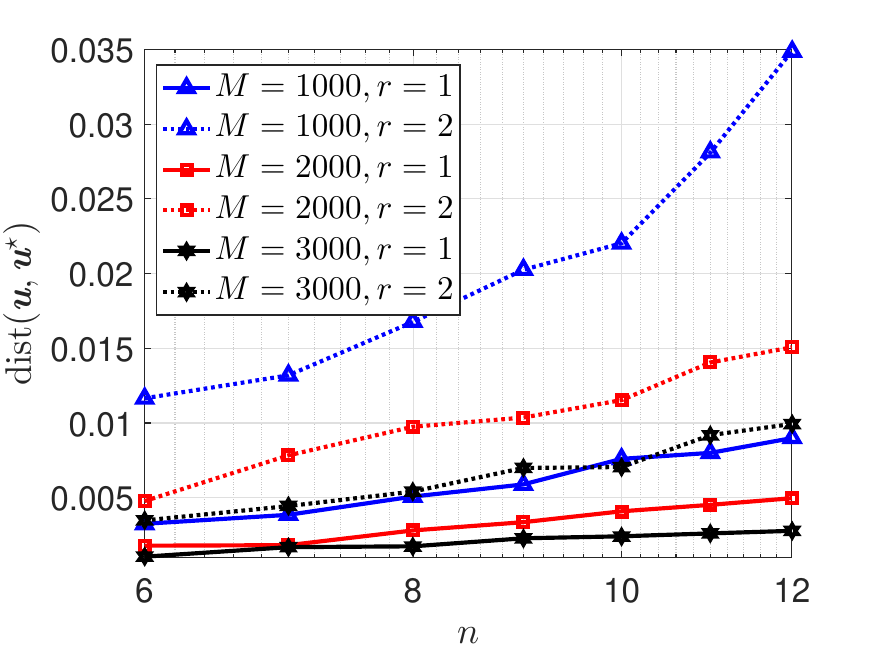}
\caption{IHT recovery error as the number of qudits $n$ increases with several choices of $M$ and $r$.}
\label{Performance Comparison_M}
\end{figure}

\paragraph*{Experimental results}
We first set $M=1000$, $r=2$, and $d = 2$  and examine the convergence of the IHT algorithm defined in \eqref{Iterative equ of GDwithTT_SVD_1}. \Cref{Performance Comparison_M_loss} shows the convergence of the algorithm in minimizing the loss function defined in \eqref{Loss function of MPS}; it can be observed that the IHT algorithm converges relatively fast.
\Cref{different m2} plots the learning curves in terms of the recovery error for the ground-truth $\vu^\star$ as defined in \eqref{THE definition of MSE}. We first note that the initialization $\vu_0$ is not close to the ground truth $\vu^\star$, which is consistent with our choice of initialization described above. After convergence, the algorithm produces a much better recovery of $\vu^\star$ and the recovery error increases steadily as $n$ increases. It is also of interest to note that when $n$ increases while $M$ remains the same, the recovery error curve exhibits a ``U-shape" that first decreases, followed by an increasing trend. In other words, if the algorithm stops appropriately at the initial phase, it produces an iterate much closer to the ground truth than the final one. This is sometimes called algorithmic bias and can be exploited to produce a better solution~\cite{stoger2021small,wang2021early,
ding2022validation}. But we highlight that we do not use this early-stopping approach here, and instead run the algorithm until convergence and use the final iterate since our goal is to verify the properties of the global minimizer.

Next, in \Cref{Performance Comparison_M}, we plot the recovery error as a function of $n$ for various values of $M$ and $r$.
As expected, the recovery error increases when $M$ decreases or $r$ increases, but the IHT algorithm produces stable performance in all cases. We also observe that the recovery error increases only polynomially rather than exponentially with $n$, which is consistent with  \Cref{Statistical Error_Haar_Measurement}.

\section{Discussion and Conclusion}
\label{sec: Discussion and conclusion}

In this paper, we have studied sampling bounds for recovering structured quantum states that can be represented as matrix product operators (MPOs). We first established a non-asymptotic lower bound on the number of
requisite measurements to ensure a stable embedding of MPOs under several choices of random measurement ensembles, including generic subgaussian measurements,  rank-one Gaussian measurement ensembles, and Haar random  projective measurements. We then established theoretical bounds on the accuracy of a constrained least-squares estimator for recovering an MPO by using its empirical measurements obtained from multiple Haar random  projective measurements. Our research shows that a stable recovery guarantee requires only \emph{polynomial} growth in the total number of state copies relative to the number of qudits. Thus, these results support the growing evidence for using MPOs for quantum state tomography and may have implications for the advancement of more efficient quantum state tomography methods in the future. Our findings suggest interesting directions for enhancing the current results or expanding our research to a more practical context. We elaborate on these possibilities below.

\paragraph*{Stable embedding for MPOs with a single Haar random projective} As discussed right after \Cref{Small_Method_ROHaar_Measurement}, we conjecture that a single Haar random  projective is sufficient to establish stable embeddings for MPOs. This is supported by our numerical experiments with measurements from a single POVM to recover the MPO state. One possible approach is to exploit the fact that the local correlations between the columns in the unitary matrix are very weak because orthogonality is a global property \cite{tropp2012comparison}. Incorporating this property into Mendelson's small ball method presents a challenge, however. Another approach is to exploit the connection between the unitary matrix and the Gaussian distribution, as used in \cite{rauhut2019low} for studying rank-one tight frame measurements.

\paragraph*{Improving sampling complexity for the number of state copies}

In \Cref{Statistical Error_Haar_Measurement}, we established a recovery guarantee for MPOs from Haar random projective measurements. The result requires a total number of state copies
$QM = \wt\Omega\big(\frac{n^3 d^2\ol{r}^2}{\epsilon^2}\big)$. This sampling complexity is probably not optimal; one may compare it to $O(nd^2\ol{r}^2)$, the number of degrees of freedom in the MPOs. Below we consider rank-one Gaussian measurements and use an alternative approach to establish a recovery guarantee.

Consider the rank-one Gaussian measurement ensembles $\{\va_k \va_k^\dagger\}_{k=1}^K$. The Chernoff bound \cite{ahlswede2002strong} implies that for sufficiently large $K$, $\frac{1}{K}\sum_{k=1}^K \va_k \va_k^\dagger \approx \mId_{d^n}$. Hence, we may view $\{\frac{1}{K}\va_k \va_k^\dagger\}_{k=1}^K$ as being similar to a POVM, though the rank-one measurement matrices do not exactly sum to the identity. Then, we may define the population measurements as $p_k = \<\frac{1}{K}\va_k \va_k^\dagger, \vrho^\star\>, k = 1,\ldots,K$, and denote by $\calA$ the associated linear measurement operator such that $\calA(\vrho^\star) = \{p_k\}$. Also, the empirical measurements obtained by measuring the states $M$ times are denoted by $\wh p_k = f_k/M$, where $f_1,\ldots,f_K$  follow the multinomial distribution $\operatorname{Multinomial}(M, \vp)$  with parameters $M$ and $\vp = \begin{bmatrix} p_1 & \cdots & p_K\end{bmatrix}^\top$. Denote by $\veta$ the measurement errors with entries $\eta_k = \wh p_k - p_k,k = 1,\ldots,K$.

Suppose we solve the same problem \eqref{The loss function in QST} (with $\calA^Q$ replaced by $\calA$) and denote its global solution as $\wh \vrho$. It follows that
\[
\|\calA(\wh{\vrho}) - \calA(\vrho^\star)- \veta\|_2\leq \|\calA(\vrho^\star) - \calA(\vrho^\star) - \veta\|_2 = \|\veta\|_2.
\]
On the other hand, $\|\calA(\wh{\vrho}) - \calA(\vrho^\star)- \veta\|_2 \ge \|\calA(\wh{\vrho}) - \calA(\vrho^\star)\|_2  - \|\veta\|_2$, which together with the above equation gives
\e
\|\calA(\wh{\vrho}) - \calA(\vrho^\star)\|_2 \le 2\|\veta\|_2.
\label{eq:gaussian-recovery-1}\ee

According to \Cref{Small_Method_ROSubgaussian_Measurement}, the left-hand side can be further lower bounded by $\|\calA(\wh{\vrho}) - \calA(\vrho^\star)\|_2 \gtrsim  \|\wh{\vrho} - \vrho^\star\|_F/\sqrt{K}$. Note that the scaling is different from \eqref{Small_Method_ROSubgaussian_Measurement_Conclusion_Theorem}, which is due to the scaling difference between the measurement operator $\calA$ and the one defined in \Cref{Small_Method_ROSubgaussian_Measurement}. On the other hand, since $\E{\veta}=0$ and $\E{\|\veta\|_2^2} = \frac{1}{M^2}\sum_{k=1}^K Mp_k(1-p_k) \leq \frac{1}{M}$, we can use a concentration inequality such as Chebyshev’s inequality to obtain $\|\veta\|_2 \lesssim \frac{1}{\sqrt{M}}$ with high probability. Plugging these equations into \eqref{eq:gaussian-recovery-1} gives
\begin{eqnarray}
    \label{Statistical Error of rank one sub-gaussian_conclusion_2}
    \|\wh{\vrho} - \vrho^\star\|_F\lesssim \sqrt{\frac{K}{M}}.
\end{eqnarray}

The above derivation is commonly used in studying recovery accuracy for inverse problems. On one hand, since \Cref{Small_Method_ROSubgaussian_Measurement} only requires $K=\wt\Omega(nd^2\ol{r}^2)$, by taking $K=\Omega(nd^2\ol{r}^2\log n)$ in \eqref{Statistical Error of rank one sub-gaussian_conclusion_2}, we observe that $M = \wt\Omega(nd^2\ol{r}^2/\epsilon^2)$ is sufficient to ensure $ \|\wh{\vrho} - \vrho^\star\|_F \le \epsilon$. As demonstrated in this context, this approach often leads to an optimal, or nearly optimal, recovery bound when using a minimal yet sufficient number of measurements. As another example, the work~\cite{haah2017sample} employs this approach to establish a recovery bound for low-rank states. But on the other hand, recall that the above derivation is based on the assumption that $\frac{1}{K}\sum_{k=1}^K \va_k \va_k^\dagger \approx \mId_{d^n}$, which holds only when $K\gtrsim d^n$. If we plug this into \eqref{Statistical Error of rank one sub-gaussian_conclusion_2}, then $M\gtrsim d^n/\epsilon^2$ is required to ensure $\epsilon$-accuracy. This also illustrates the challenge described in \Cref{sec:challenge}. Nevertheless, the discussion above suggests that it may still be possible to improve upon our current bound for the total number of state copies.

\paragraph*{Convergence of IHT and other efficient optimization algorithms}
The algorithmic aspect is not the focus of this work. In our experiments, we employ an IHT algorithm, but we do not provide a formal guarantee for that algorithm. It will be of interest to develop a convergence guarantee for this algorithm. As discussed in \cite{Rauhut17}, one potential challenge is to find a good initialization that allows IHT to converge quickly to the target solution. On the other hand, the IHT algorithm requires performing a sequential SVD algorithm in each iteration, which could be computationally expensive, especially for large quantum systems.
Consequently, exploring alternative optimization algorithms that offer computational efficiency without the need for a projection step and can effectively handle an increasing number of qudits has become an area of great interest.

\paragraph*{Extension to local measurements}
In this paper, we primarily focus on rank-one POVMs with Haar-distributed unitary matrices. Such measurements are known as global measurements since the unitary matrix will rotate the entire system of qudits simultaneously. This poses challenges in performing these measurements with practical quantum circuits. Therefore, an important future direction we will pursue is to study other measurement settings, such as the unitary t-design \cite{scott2006tight,brandao2016local} or even local measurements \cite{brandao2020fast, lancien2013distinguishing} that can be conducted efficiently on current quantum computers. Such measurement settings may also reduce the cost for computing the gradient of the least squares loss \eqref{The loss function in QST}. It is also interesting to design measurement operators that can improve efficiency in both performing experimental measurements and post-processing for estimating the state, which is often achieved by using certain iterative algorithms.


\section*{Acknowledgment}
We acknowledge funding support from NSF Grants No.\ CCF-1839232, PHY-2112893, CCF-2106834 and CCF-2241298, as well as the W. M. Keck Foundation. We thank the Ohio Supercomputer Center for providing the computational resources needed in carrying out this work. Finally, we are grateful to Rungang Han, Holger Rauhut, Gongguo Tang, and Roman Vershynin for many valuable discussions and to Alireza Goldar for helpful comments on the manuscript.

\section*{Appendices}
\appendices
To simplify the notations, universal constants in each proof may share the same symbols (e.g., $c_0$), but they could represent different values.

\section{Proof of \Cref{thm:SubgaussianRIP}}
\label{Proof of the RIP}

\subsection{Generic subgaussian measurements}
We first extend the statement of \Cref{thm:SubgaussianRIP} to generic subgaussian measurements and then prove this more general result.

\begin{defi}[Subgaussian measurement ensembles~\cite{bierme2015modulus}]
A complex random variable $X$ is called $L$-subgaussian if there exists a constant $L > 0$ such that $\E e^{\mathscr{R}(t X)} \le e^{L^2|t|^2/2}
$ holds for all $t \in \C$. We say that $\calA: \C^{d^n\times d^n} \rightarrow \C^K$  is an $L$-subgaussian measurement ensemble if all the elements of $\mA_k,k=1,\dots,K$ are independent $L$-subgaussian random variables with mean zero and variance one.
\label{def:Subgaussian}\end{defi}

Note that a complex-valued random variable $X$ is subgaussian if and only if its both real part $\mathscr{R}(X)$ and imaginary part $\mathscr{I}(X)$ are real subgaussian random variables. We define  the subgaussian norm of $X$ as
\begin{equation}
    \label{Subgaussian norm}
    \|X\|_{\psi_2} = \inf \left\{t>0, \ \E e^{\frac{|X|^2}{t^2}} \leq 2  \right\}.
\end{equation}
Here  are some classical examples of subgaussian distributions.
\begin{itemize}
\item{(Gaussian)} A standard complex Gaussian random variable $X = \mathscr{R}(X) + i \mathscr{I}(X)$ with $\mathscr{R}(X)$ and $\mathscr{I}(X)$ being independent and following $ \calN(0,\frac{1}{2})$, is a subgaussian random variable with $\|X\|_{\psi_2}\leq C$, where $C$ is an absolute constant.
\item{(Bernoulli)} A Bernoulli random variable $X$ that takes values $-1$ and $1$ with equal probability is a subgaussian random variable with $\|X\|_{\psi_2}= \frac{1}{\sqrt{\ln 2}}$.
\end{itemize}

The following result establishes the RIP for an $L$-subgaussian measurement ensemble.
\begin{theorem}
Suppose the linear map $\calA: \C^{d^n \times d^n} \rightarrow \C^K$  is an $L$-subgaussian measurement ensemble defined in {Definition} \ref{def:Subgaussian}. Then, with probability at least $1-\bar\epsilon$, $\calA$ satisfies the $\delta_{\ol{r}}$-RIP as in \eqref{eq:ripdef} for MPOs given that
\begin{equation}
K \ge C \cdot \frac{1}{\delta_{\ol{r}}^2} \cdot \max\left\{ nd^2{\ol{r}}^2 (\log n\ol{r}), \log(1/\ol\epsilon)\right \},
\label{eq:mrip1}
\end{equation}
where $C$ is a universal constant depending only on $L$.
\label{thm:SubgaussianRIP-appendix}
\end{theorem}

\subsection{Covering number for MPOs}
To study the RIP or similar properties for MPOs, we need to first compute the covering number for the set of unit-norm MPOs. Since a unit-norm MPO can always be written in the canonical form, we consider the set $\ol \setX_{\ol r}$ defined in \eqref{SetOfMPO1}.
Note that the definition of \eqref{SetOfMPO1} is different from \eqref{SetOfMPO} since in \eqref{SetOfMPO1}, we assume $\|\vrho\|_F=1$ such that $\sum_{i_n,j_n}{\mX_n^{i_n,j_n}}^\dagger \mX_n^{i_n,j_n} = 1$.

We say $\widetilde\setX_{\ol{r}}$ is an $\epsilon$-net of $\ol\setX_{\ol{r}}$ if for any $\vrho\in \ol\setX_{\ol{r}}$, there exists $\ol\vrho\in \widetilde\setX_{\ol{r}}$ such that $\norm{\vrho - \ol\vrho}{F} \le \epsilon$. Here we use the Frobenius norm to quantify the distance, but one may use other metrics depending on the application. We now provide the covering number of $\widetilde\setX_{\ol{r}}$.
\begin{lemma}
There exists an $\epsilon$-net $\widetilde\setX_{\ol{r}}$ for $\ol\setX_{\ol{r}} $ in \eqref{SetOfMPO1} under the Frobenius norm obeying
\begin{eqnarray}
    \label{lem:cover-number-mps-equ}
    \abs{\widetilde\setX_{\ol{r}}} \le \bigg(\frac{3n\ol{r}}{\epsilon}\bigg)^{nd^2{\ol{r}}^2},
\end{eqnarray}
where $\abs{\widetilde\setX_{\ol{r}}}$ denotes the number of elements in the set $\widetilde\setX_{\ol{r}}$.
\label{lem:cover-number-mps}\end{lemma}

\begin{proof} Denote by
\begin{eqnarray}
L(\mX_\ell) &\!\!\!\!=\!\!\!\!& \begin{bmatrix}\mX_\ell^{1,1} \\ \vdots \\ \mX_\ell^{d,d} \end{bmatrix} \in \C^{d^2r_{\ell-1}\times r_{\ell}}, \text{which satisfies}\nonumber\\
&\!\!\!\!\!\!\!\!&  \hspace{-0.5cm} \ L(\mX_\ell)^\dagger L(\mX_\ell) = \sum_{i_\ell,j_\ell}{\mX_\ell^{i_\ell,j_\ell}}^\dagger \mX_\ell^{i_\ell,j_\ell} = \mId_{r_\ell}.
\end{eqnarray}

For covering
\[
\setO_{d^2r_{\ell-1}, r_\ell} = \bigg\{L(\mX_\ell) \in \C^{d^2r_{\ell-1}\times r_\ell}, L(\mX_\ell)^\dagger \\ \cdot L(\mX_\ell)  =  \mId_{r_\ell}  \bigg\},\]
which contains matrices with unit-norm and orthogonal column vectors, it is beneficial to use the $\norm{\cdot}{1,2}$ norm which counts the largest energy of each column, i.e.,  $\norm{\mA}{1,2} = \max_i\norm{\mA(:,i)}{2}$. By relaxing $\setO_{d^2r_{\ell-1}, r_\ell}$ to the set of matrices with unit-norm vectors, the standard result on the covering number of unit ball implies that there exists an $\epsilon_\ell$-net $\ol\setO_{d^2r_{\ell-1}, r_\ell}$ for $\setO_{d^2r_{\ell-1}, r_\ell}$ obeying \[\abs{\ol\setO_{d^2r_{\ell-1}, r_\ell}} \le \bigg(\frac{3}{\epsilon_\ell}\bigg)^{d^2r_{\ell-1} r_\ell}\le \bigg(\frac{3}{\epsilon_\ell}\bigg)^{d^2\ol{r}^2}.\]
Then we define the set
\begin{eqnarray}
\widetilde\setX_{\ol{r}}:&\!\!\!\!=\!\!\!\!& \bigg\{\ol\vrho: \ol\vrho(i_1 \cdots i_\nqbit, j_1\cdots j_\nqbit) = \Pi_{\ell=1}^n\ol\mX_\ell^{i_\ell,j_\ell}, \nonumber\\
&\!\!\!\!\!\!\!\!& L(\ol\mX_\ell) = \begin{bmatrix}\ol\mX_\ell^{1,1} \\ \vdots \\ \ol\mX_\ell^{d,d} \end{bmatrix} \in \ol\setO_{d^2r_{\ell-1}, r_\ell}, \ \forall \ell\in[n] \bigg\},
\end{eqnarray}
which obeys
\begin{eqnarray}
\label{covering number olW}
\abs{\widetilde\setX_{\ol{r}}} \le \prod_\ell \bigg(\frac{3}{\epsilon_\ell}\bigg)^{d^2\ol{r}^2}.
\end{eqnarray}

We now verify that $\widetilde\setX_{\ol{r}}$ is an $\epsilon$-net for $\ol\setX_{\ol{r}}$ by appropriately selecting $\epsilon_\ell$. For any $\vrho \in \ol\setX_{\ol{r}}$ with $\vrho(i_1 \dots i_\nqbit, j_1\dots j_\nqbit) = \Pi_{\ell=1}^n\mX_\ell^{i_\ell,j_\ell}$, we construct $\ol\vrho$ with $\ol\vrho(i_1 \dots i_\nqbit, j_1\dots j_\nqbit) = \Pi_{\ell=1}^n\ol\mX_\ell^{i_\ell,j_\ell}$ where $\norm{L(\mX_\ell) - L(\ol\mX_\ell)}{1,2} \le \epsilon_\ell$. Then we have
\begin{eqnarray}
&\!\!\!\!\!\!\!\!&\hspace{-0.5cm}\Gamma= \norm{\vrho - \ol\vrho }{F}^2 \nonumber\\
&\!\!\!\!=\!\!\!\!& \hspace{-0.3cm}  \sum_{\substack{i_1, \dots, i_n, \\ j_1,\dots, j_n}}
    \hspace{-0.1cm}\bigg|\mX_1^{i_1,j_1}\mX_2^{i_2,j_2}\cdots\mX_n^{i_n,j_n}-  \ol\mX_1^{i_1,j_1}\ol\mX_2^{i_2,j_2}\cdots\ol\mX_n^{i_n,j_n}\bigg|^2\nonumber\\
    & \!\!\!\!=\!\!\!\!&\sum_{\substack{i_1, \dots, i_n, \\ j_1,\dots, j_n}}
    \bigg|\sum_{\ell=1}^n  \Big(\ol\mX_1^{i_1,j_1}\cdots\ol\mX_{\ell-1}^{i_{\ell-1},j_{\ell-1}}\mX_{\ell}^{i_{\ell},j_{\ell}} \mX_n^{i_n,j_n} \nonumber\\
    &\!\!\!\!\!\!\!\!& \hspace{1.5cm} -  \ol\mX_1^{i_1,j_1} \cdots\ol\mX_{\ell}^{i_{\ell},j_{\ell}}\mX_{\ell+1}^{i_{\ell+1},j_{\ell+1}} \cdots\mX_n^{i_n,j_n}\Big) \bigg|^2\nonumber\\
    &\!\!\!\! \le\!\!\!\!& n\sum_{\ell=1}^n  \Gamma_\ell \leq n \sum_{\ell =1}^n r_\ell\epsilon_\ell^2,
\end{eqnarray}
where in the first inequality we define
\begin{align}
\label{the definition of gamma_l}
\Gamma_\ell=& \sum_{\substack{i_1, \dots, i_n,\\ j_1,\dots, j_n}}
    \bigg|\ol\mX_1^{i_1,j_1}\cdots\ol\mX_{\ell-1}^{i_{\ell-1},j_{\ell-1}}  \cdot\mX_{\ell}^{i_{\ell},j_{\ell}} \cdots\mX_n^{i_n,j_n} \nonumber\\
    & -\ol\mX_1^{i_1,j_1}\cdots\ol\mX_{\ell}^{i_{\ell},j_{\ell}}\mX_{\ell+1}^{i_{\ell+1},j_{\ell+1}} \cdots\mX_n^{i_n,j_n}\bigg|^2,
\end{align}
and the second inequality uses the inequalities $\Gamma_{n} \leq \epsilon_n^2$ and $\Gamma_{\ell} \le r_\ell\epsilon_\ell^2, \ell=1,\dots, n-1$  which can be proved as follows. First note that $\Gamma_{n} $ can be written as
\begin{align*}
\Gamma_{n}  &= \sum_{\substack{i_1, \dots, i_n\\ j_1,\dots, j_n}}
    \bigg|\ol\mX_1^{i_1,j_1}\cdots\ol\mX_{n-1}^{i_{n-1},j_{n-1}}\mX_n^{i_n,j_n} \\
    & \hspace{0.5cm}  - \ol\mX_1^{i_1,j_1} \cdots\ol\mX_{n-1}^{i_{n-1},j_{n-1}}\ol\mX_n^{i_n,j_n}\bigg|^2 \\
&=\!\!\!  \sum_{\substack{i_1, \dots, i_n\\ j_1,\dots, j_n}} \!\!\!
\bigg|\ol\mX_1^{i_1,j_1} \underbrace{ \ol\mX_2^{i_2,j_2}\cdots\ol\mX_{n-1}^{i_{n-1},j_{n-1}}(\mX_n^{i_n,j_n} - \ol \mX_n^{i_n,j_n})}_{\vxi_{\substack{i_2 \dots, i_n\\ j_2,\dots, j_n}}} \bigg|^2 \\
& = \sum_{\substack{i_2, \dots, i_n\\ j_2,\dots, j_n}} \vxi_{\substack{i_2 \dots, i_n\\ j_2,\dots, j_n}}^\dagger  \cdot \underbrace{\sum_{i_1,j_1}({\ol\mX_1^{i_1,j_1}}^\dagger\ol\mX_1^{i_1,j_1})}_{=\mId_{r_1}} \cdot \vxi_{\substack{i_2 \dots, i_n\\ j_2,\dots, j_n}} \\
& = \sum_{\substack{i_2, \dots, i_n\\ j_2,\dots, j_n}} \vxi_{\substack{i_2 \dots, i_n\\ j_2,\dots, j_n}}^\dagger \vxi_{\substack{i_2 \dots, i_n\\ j_2,\dots, j_n}}.
\end{align*}
Repeating the above process $(n-2)$ more times yields
\begin{align*}
\Gamma_{n} & = \sum_{i_n,j_n} (\mX_n^{i_n,j_n} - \ol \mX_n^{i_n,j_n})^\dagger(\mX_n^{i_n,j_n} - \ol \mX_n^{i_n,j_n}) \\
&=  \norm{L(\mX_n) - L(\ol\mX_n)}{2}^2 \le \epsilon_n^2.
\end{align*}
Using the same approach, we can obtain
\begin{align*}
\Gamma_{\ell} & =  \sum_{\substack{i_\ell, \dots, i_n, \\ j_\ell,\dots, j_n}} \Bigg( {\mX_n^{i_n,j_n}}^\dagger\cdots{\mX_{\ell+1}^{i_{\ell+1},j_{\ell+1}}}^\dagger(\mX_{\ell}^{i_{\ell},j_{\ell}} - \ol\mX_{\ell}^{i_{\ell},j_{\ell}})^\dagger\nonumber\\
& \hspace{0.5cm}\cdot(\mX_{\ell}^{i_{\ell},j_{\ell}} - \ol\mX_{\ell}^{i_{\ell},j_{\ell}})\underbrace{\mX_{\ell+1}^{i_{\ell+1},j_{\ell+1}} \cdots\mX_n^{i_n,j_n}}_{\vxi_{\substack{i_{\ell+1} \dots, i_n\\ j_{\ell+1},\dots, j_n}}} \Bigg)\nonumber\\
&\leq \norm{L(\mX_{\ell}) - L(\ol \mX_{\ell}) }{F}^2  \underbrace{\sum_{\substack{i_{\ell+1}, \dots, i_n, \\ j_{\ell+1},\dots, j_n}}\vxi_{\substack{i_{\ell+1} \dots, i_n\\ j_{\ell+1},\dots, j_n}}^\dagger\vxi_{\substack{i_{\ell+1} \dots, i_n\\ j_{\ell+1},\dots, j_n}}}_{=1}\nonumber\\
&\le r_\ell \epsilon_{\ell}^2
\end{align*}
for all $\ell \le n-1$.
Therefore, we can choose $\epsilon_\ell = \frac{\epsilon}{\ol{r} n}$ in \eqref{covering number olW} to ensure $\widetilde\setX_{\ol{r}}$ as an $\epsilon$-net for $\ol\setX_{\ol{r}}$ (as $\norm{\vrho - \ol\vrho }{F}^2\leq n \sum_{\ell=1}^n r_\ell\epsilon_\ell^2 \leq \epsilon^2$) and such  $\widetilde\setX_{\ol{r}}$ obeys
\begin{eqnarray}
\abs{\widetilde\setX_{\ol{r}}} \le  \bigg(\frac{3n\ol{r}}{\epsilon}\bigg)^{nd^2{\ol{r}}^2}.
\end{eqnarray}
This completes the proof of Lemma~\ref{lem:cover-number-mps}.
\end{proof}

\subsection{Proof of \Cref{thm:SubgaussianRIP-appendix}}
Using the covering number established in Lemma~\ref{lem:cover-number-mps}, we can now follow the arguments in \cite{rauhut2017low} to establish the RIP for MPOs $\vrho$ under subgaussian measurements.
Because of the linearity of the measurement operator $\calA$, we note that there exists a complex-valued matrix $\mA$ of size $K \times d^{2n}$ such that
\begin{eqnarray}
\calA(\vrho) = \mA \operatorname{vec}(\vrho),
\end{eqnarray}
where $\operatorname{vec}(\vrho)\in\C^{d^{2n}}$ denotes the vectorization (in any predetermined order) of the MPO format $\vrho$.
Note that our goal is to study the quantity $\frac{1}{K}\| \calA(\vrho) \|_2^2 = \| \frac{1}{\sqrt{K}}\calA(\vrho) \|_2^2$.
Equivalently, there exists a vector $\vxi \in \C^{Kd^{2n}}$ (containing the row-wise vectorization of $\mA$) such that \begin{equation} \frac{1}{\sqrt{K}}\calA(\vrho)  = \mV_{\vrho} \vxi, \label{eq:mtxi} \end{equation}
where $\mV_{\vrho}$ is an $K \times Kd^{2n}$ matrix given by
\begin{equation}
\mV_{\vrho} = \frac{1}{\sqrt{K}} \begin{bmatrix} \operatorname{vec}(\vrho)^\dagger & & & \\ & \operatorname{vec}(\vrho)^\dagger & & \\ & & \ddots & \\ & & & \operatorname{vec}(\vrho)^\dagger \end{bmatrix}.
\label{eq:vt}
\end{equation}

Now we begin to prove \Cref{thm:SubgaussianRIP-appendix}.
\begin{proof} For any $\vrho \in \ol\setX_{\ol{r}}$ in \eqref{SetOfMPO}, we recall~\eqref{eq:ripdef}. Because $\vxi$ is a random vector, $\mV_{\vrho}\vxi$ is also a random vector. We can compute the expectation of the energy of this random vector:
\begin{align}
\E \| \mV_{\vrho}\vxi \|_2^2 &= \E (\vxi^\dagger \mV_{\vrho}^\dagger \mV_{\vrho}\vxi) = \E \trace(\vxi^\dagger \mV_{\vrho}^\dagger \mV_{\vrho}\vxi) \nonumber \\ &= \E \trace(\mV_{\vrho}^\dagger \mV_{\vrho}\vxi \vxi^\dagger)= \trace(\mV_{\vrho}^\dagger \mV_{\vrho}\E(\vxi \vxi^\dagger))\nonumber\\
& = \trace(\mV_{\vrho}^\dagger \mV_{\vrho} \mId) = \|\vrho\|_F^2.
\label{eq:ripexp}
\end{align}
Here we used the fact that $\E \vxi \vxi^\dagger=\mId$ since, by assumption, all elements of $\vxi$ are independent mean-zero, variance one, $L$-subgaussian variables. Using  \eqref{eq:mtxi}, and~\eqref{eq:ripexp}, we note that proving that $\calA$ satisfies the $\delta_{\ol{r}}$-RIP is equivalent to proving
\begin{equation}
\sup_{\vrho \in \ol\setX_{\ol{r}}} \left| \| \mV_{\vrho}\vxi \|_2^2 - \E \| \mV_{\vrho}\vxi \|_2^2 \right| \le \delta_{\ol{r}}.
\label{eq:ripequiv}
\end{equation}
We can view $\left| \| \mV_{\vrho}\vxi \|_2^2 - \E \| \mV_{\vrho}\vxi \|_2^2 \right|$ as a random process indexed by the variable $\vrho$, and our goal is to bound the supremum of this random process over the set $\ol\setX_{\ol{r}}$.  \cite[Theorem 3]{rauhut2017low} gives a mechanism to bound this supremum. Specifically, let $\calB := \{ \mV_{\vrho}: ~ \vrho \in \ol\setX_{\ol{r}}\}$ and note that
\begin{equation}
\sup_{\mB \in \calB} \left| \| \mB \vxi \|_2^2 - \E \| \mB \vxi \|_2^2 \right| =
 \sup_{\vrho \in \ol\setX_{\ol{r}}} \left| \| \mV_{\vrho}\vxi \|_2^2 - \E \| \mV_{\vrho}\vxi \|_2^2 \right|.
\label{eq:bvequiv}
\end{equation}
\cite[Theorem 3]{rauhut2017low} states that there exist constants $c_1,c_2$ (depending on $L$) such that for $t > 0$,
\begin{eqnarray}
    &\!\!\!\!\!\!\!\!&\P{ \sup_{\mB \in \calB} \left| \| \mB \vxi \|_2^2 - \E \| \mB \vxi \|_2^2 \right| \ge c_1 E + t }\nonumber\\
    &\!\!\!\!\le\!\!\!\!& 2 e^{ -c_2 \min \left\{\frac{t^2}{V^2},\frac{t}{U}\right\} },
    \label{eq:rauhutthm3}
\end{eqnarray}
where $E$, $U$, and $V$ are quantities defined as
\begin{align}\label{eq:euvdef}
    E & := \gamma_2(\calB,\|\cdot\|_{2\rightarrow 2}) \left(\gamma_2(\calB,\|\cdot\|_{2\rightarrow 2})+ d_F(\calB)\right)\nonumber\\
    & \hspace{0.5cm}+ d_F(\calB)d_{2\rightarrow 2}(\calB), \nonumber \\
    V &:= d_4^2(\calB), \\
    U &:= d_{2\rightarrow 2}^2(\calB)\nonumber,
\end{align}
and $d_F(\calB)$, $d_{2\rightarrow 2}(\calB)$, $d_4^2(\calB)$, and $\gamma_2(\calB,\|\cdot\|_{2\rightarrow 2})$ are quantities that we define and bound in the next paragraph.

In this paragraph, we bound the quantities $E$, $U$, and $V$ appearing in~\eqref{eq:rauhutthm3}. To do this, we define and bound $d_F(\calB)$, $d_{2\rightarrow 2}(\calB)$, $d_4^2(\calB)$, and $\gamma_2(\calB,\|\cdot\|_{2\rightarrow 2})$ which appear in the definitions of $E$, $U$, and $V$ in~\eqref{eq:euvdef}. First,
\begin{equation}
d_F(\calB) := \sup_{\mB \in \calB} \| \mB \|_F^2 =\sup_{\vrho \in \ol\setX_{\ol{r}}} \| \vrho \|_F^2 = 1,
\end{equation}
since every MPO format $\vrho \in \ol\setX_{\ol{r}}$ has unit norm. Second,
\begin{equation}
d_{2\rightarrow 2}(\calB) := \sup_{\mB \in \calB} \| \mB \|_{2\rightarrow 2} = \sup_{\vrho \in \ol\setX_{\ol{r}}} \frac{1}{\sqrt{K}} \| \vrho \|_F^2 = \frac{1}{\sqrt{K}},
\label{eq:twotwo}
\end{equation}
due to the block diagonal structure of $\mV_{\vrho}$ (see~\eqref{eq:vt}) and the normalization of all $\vrho \in \ol\setX_{\ol{r}}$. Third,
\begin{equation}
d_4(\calB) := \sup_{\mB \in \calB} \left( \trace(\mB^\dagger \mB)^2 \right)^{1/4} = K^{-1/4},
\end{equation}
see \cite[Eqn. (65) ]{rauhut2017low} for an analogous derivation.
Fourth,
\begin{align}
\gamma_2(\calB,\|\cdot\|_{2\rightarrow 2}) &\le C \int_{0}^{d_{2\rightarrow 2}(\calB)} \sqrt{\log \calN(\calB, \|\cdot\|_{2\rightarrow 2}, u)} \; du \nonumber \\
&= C \int_{0}^{\frac{1}{\sqrt{K}}} \sqrt{\log \calN(\calB, \|\cdot\|_{2\rightarrow 2}, u)} \; du,
\label{eq:dudley1}
\end{align}
where the covering number $\calN(\calB, \|\cdot\|_{2\rightarrow 2}, u)$ denotes the minimum cardinality of a $u$-net for $\calB$ with respect to the norm $\|\cdot\|_{2\rightarrow 2}$. As suggested by~\eqref{eq:twotwo}, the $\|\cdot\|_{2\rightarrow 2}$ distance on $\calB$ is equivalent to $\frac{1}{\sqrt{K}}$ times the squared Frobenius distance on $\ol\setX_{\ol{r}}$. Therefore,
\begin{eqnarray}
\calN(\calB, \|\cdot\|_{2\rightarrow 2}, u)& \!\!\!\!= \!\!\!\!&
\calN(\ol\setX_{\ol{r}}, \frac{1}{\sqrt{K}} \|\cdot\|_F^2, u)\nonumber\\
&\!\!\!\! =\!\!\!\!& \calN(\ol\setX_{\ol{r}}, \|\cdot\|_F, K^{1/4} \sqrt{u}).
\end{eqnarray}
Changing variables by letting $\epsilon = K^{1/4} \sqrt{u}$, \eqref{eq:dudley1} becomes
\begin{eqnarray}
\gamma_2(\calB,\|\cdot\|_{2\rightarrow 2}) &\hspace{-0.1cm}\!\!\!\!\le\!\!\!\!& 2 C \frac{1}{\sqrt{K}} \int_{0}^{1} \epsilon \sqrt{\log \calN(\ol\setX_{\ol{r}}, \|\cdot\|_F, \epsilon)} \; d\epsilon\nonumber\\
&\hspace{-0.1cm} \!\!\!\! \le \!\!\!\!& C \frac{1}{\sqrt{K}} \int_{0}^{1} \sqrt{\log \calN(\ol\setX_{\ol{r}}, \|\cdot\|_F, \epsilon)} \; d\epsilon,
\end{eqnarray}
where the factor of $2$ has been absorbed into the universal constant $C$. Now, by directly applying Lemma~\ref{lem:cover-number-mps}, we have that
\[
\calN(\ol\setX_{\ol{r}}, \|\cdot\|_F, \epsilon) \le \bigg(\frac{3n\ol{r}}{\epsilon}\bigg)^{nd^2{\ol{r}}^2}.
\]
Therefore,
\begin{align}
\gamma_2(\calB,\|\cdot\|_{2\rightarrow 2}) &\le  C \frac{1}{\sqrt{K}} \int_{0}^{1} \sqrt{\log \calN(\ol\setX_{\ol{r}}, \|\cdot\|_F, \epsilon)} \; d\epsilon \nonumber \\
&\le C \frac{1}{\sqrt{K}} \int_{0}^{1} \sqrt{\log \bigg(\frac{3n\ol{r}}{\epsilon}\bigg)^{nd^2{\ol{r}}^2}} \; d\epsilon \nonumber \\
&\le C \frac{1}{\sqrt{K}} \int_{0}^{1} \sqrt{{nd^2{\ol{r}}^2} \log \bigg(\frac{3n\ol{r}}{\epsilon}\bigg)} \; d\epsilon \nonumber \\
&\le C \sqrt{\frac{nd^2{\ol{r}}^2 }{K}} \int_{0}^{1} \sqrt{ \log \bigg(\frac{3n\ol{r}}{\epsilon}\bigg)} \; d\epsilon \nonumber \\
&\le C \sqrt{\frac{nd^2{\ol{r}}^2 \log n\ol{r} }{K}}, \label{eq:finalgamma2}
\end{align}
where the last line follows from the fact that
\[
\int_{0}^{1} \sqrt{ \log \bigg(\frac{3n\ol{r}}{\epsilon}\bigg)} \; d\epsilon \le C + \sqrt{\log (3n\ol{r})} \le C \sqrt{\log n\ol{r}},
\]
and each appearance of $C$ denotes an unspecified universal constant that may change from instance to instance. Putting together the above quantities, we have the following three numbers which appear in~\eqref{eq:rauhutthm3}:
\begin{align}\label{eq:euv}
    E :&= \gamma_2(\calB,\|\cdot\|_{2\rightarrow 2}) \left(\gamma_2(\calB,\|\cdot\|_{2\rightarrow 2}) + d_F(\calB)\right)\nonumber\\
    & \hspace{0.5cm} + d_F(\calB)d_{2\rightarrow 2}(\calB) \nonumber \\ &= \gamma^2_2(\calB,\|\cdot\|_{2\rightarrow 2}) + \gamma_2(\calB,\|\cdot\|_{2\rightarrow 2}) + \frac{1}{\sqrt{K}}, \nonumber \\
    V :&=  d_4^2(\calB) = \frac{1}{\sqrt{K}}, \\
    U :&= d_{2\rightarrow 2}^2(\calB) = \frac{1}{K}\nonumber.
\end{align}

Plugging~\eqref{eq:bvequiv} and~\eqref{eq:euv} into~\eqref{eq:rauhutthm3}, we have
\begin{eqnarray}
&\!\!\!\!\!\!\!\!&\hspace{-0.5cm}\mathbb{P}\bigg( \sup_{\vrho \in \ol\setX_{\ol{r}}} \left| \| \mV_{\vrho}\vxi \|_2^2 - \E \| \mV_{\vrho}\vxi \|_2^2 \right| \ge c_1(\gamma^2_2(\calB,\|\cdot\|_{2\rightarrow 2})\nonumber\\
&\!\!\!\!\!\!\!\!&\hspace{-0.5cm}  + \gamma_2(\calB,\|\cdot\|_{2\rightarrow 2}) + \frac{1}{\sqrt{K}}) + t \bigg) \le 2 e^{-c_2 \min \left\{Kt^2,Kt\right\} }.
\end{eqnarray}
Our goal is to find a value of $K$ such that~\eqref{eq:ripequiv} holds with probability at least $1-\bar{\epsilon}$.

Let $t = \delta/2$ and recall that $\delta < 1$, so $\min \left\{Kt^2,Kt\right\} = K\delta^2/4$. If we choose $K > C \delta^{-2} \log(1/\bar{\epsilon})$ for an appropriately chosen constant $C$, we have $2 e^{-c_2 \min \left\{Kt^2,Kt\right\}} \le \bar\epsilon$. Next, using the bound on $\gamma^2_2(\calB,\|\cdot\|_{2\rightarrow 2})$ from~\eqref{eq:finalgamma2}, we see that by choosing
\begin{equation}
K \ge C \cdot \frac{nd^2\ol{r}^2 (\log n\ol{r})}{\delta^2},
\end{equation}
for an appropriately chosen constant $C$, then we guarantee that
\begin{equation}
c_1(\gamma^2_2(\calB,\|\cdot\|_{2\rightarrow 2}) + \gamma_2(\calB,\|\cdot\|_{2\rightarrow 2}) + \frac{1}{\sqrt{K}}) \le \frac{\delta}{2}.
\end{equation}
Putting all of the pieces together, we conclude that when~\eqref{eq:mrip} is satisfied, we have
\begin{equation}
 \P{\sup_{\vrho \in \ol\setX_{\bar{r}}} \left| \| \mV_{\vrho}\vxi \|_2^2 - \E \| \mV_{\vrho}\vxi \|_2^2 \right| \ge \delta } \le
\bar\epsilon.
\end{equation}
We have thus proved that~\eqref{eq:ripequiv} holds with probability at least $1-\bar\epsilon$. This completes the proof of \Cref{thm:SubgaussianRIP-appendix}.
\end{proof}

\section{Proof of \Cref{Small_Method_ROSubgaussian_Measurement}}
\label{Proof_Of_MPO_Subgaussian}

\begin{proof} In this section, we will apply Mendelson's small ball method to derive \Cref{Small_Method_ROSubgaussian_Measurement}. According to \Cref{Small Ball Method_original}, and supposing that $\{ \va_{1},\dots, \va_{K} \}$ are selected independently from the standard complex normal distribution $\va\sim\calC\calN({\bm 0},\mId_{d^n})$,  we need to bound
\begin{eqnarray}
    \label{ROSubgaussian_Measurement_H}
    H_{\xi}(\ol\setX_{\ol{r}}) = \inf_{\vrho\in\ol\setX_{\ol{r}}} \mathbb{P} \{|\< \va\va^\dagger,\vrho \>|\geq \xi  \}
\end{eqnarray}
and
\begin{eqnarray}
    \label{ROSubgaussian_Measurement_W}
    W(\ol\setX_{\ol{r}}) = \E \sup_{\vrho\in \ol\setX_{\ol{r}}}\frac{1}{\sqrt{K}}\sum_{k=1}^K \<\epsilon_{k}\va_{k}\va_{k}^\dagger, \vrho  \>,
\end{eqnarray}
where $\{\epsilon_k\}$ is a Rademacher sequence independent from everything else.

\begin{itemize}
  \item{Lower bound of $H_{\xi}(\ol\setX_{\ol{r}})$:} To bound $H_{\xi}(\ol\setX_{\ol{r}})$, we use the Paley-Zygmund inequality (\Cref{Foucart13MathIntro_lemma716}). Specifically, we can get
\begin{align}
    \label{ROSubgaussian_Measurement_H_1}
&\hspace{-0.3cm}\!\!H_{\xi}(\ol\setX_{\ol{r}})= \inf_{\vrho\in\ol\setX_{\ol{r}}} \P{|\< \va\va^\dagger,\vrho \>|\geq \xi  }\nonumber\\
    &=\inf_{\vrho\in\ol\setX_{\ol{r}}} \P{|\< \va\va^\dagger,\vrho \>|^2\geq \xi^2  }\nonumber\\
    &\geq\inf_{\vrho\in\ol\setX_{\ol{r}}} \P{|\< \va\va^\dagger,\vrho \>|^2\geq \frac{1}{2}\E[|\< \va\va^\dagger,\vrho \>|^2]  }\nonumber\\
    &\geq\inf_{\vrho\in\ol\setX_{\ol{r}}} \frac{(\E[|\< \va\va^\dagger,\vrho \>|^2])^2}{4\E[|\< \va\va^\dagger,\vrho \>|^4]}, \ \forall \xi\le \sqrt{ \frac{1}{2}\E[|\< \va\va^\dagger,\vrho \>|^2] },
\end{align}
where the first inequality follows because $\P{|\< \va\va^\dagger,\vrho \>|^2\\ \geq \xi^2  }$ is a decreasing function with respect to $\xi$, and the second inequality uses \Cref{Foucart13MathIntro_lemma716}.

Next we start to analyze $\frac{(\E[|\< \va\va^\dagger,\vrho \>|^2])^2}{4\E[|\< \va\va^\dagger,\vrho \>|^4]}$.
By the fact that $\langle \va\va^\dagger,\vrho \rangle$ is a second-order polynomial in the entries of Gaussian random vector $\va$, we can obtain $\|\< \va\va^\dagger,\vrho\>\|_{\psi_1}\leq c\|\va\|_{\psi_2}^2\|\vrho\|_F\leq O(1)$ \cite{zajkowski2020bounds} for some constant $c$; thus, $\langle \va\va^\dagger,\vrho \rangle$ is a subexponential random variable. Hence, there exists $\alpha$ such that $\E e^{\alpha |\langle \va\va^\dagger,\vrho \rangle|}$ is finite. It follows from \Cref{lem:hypercontractivity} that there exists a constant $C_0$ such that
\begin{eqnarray}
    \label{ROSubgaussian_Measurement_H_fourth and second}
\E [|\langle \va\va^\dagger,\vrho \rangle|^4] \le C_0 \parans{\E [|\langle \va\va^\dagger,\vrho \rangle|^2]}^2.
\end{eqnarray}

We need obtain $\xi$ to finish the analysis. To that end, we bound the expectation $\E[|\langle \va\va^\dagger,\vrho \rangle|^2]$. Since $\vrho$ is Hermitian, it has the eigenvalue decomposition $\vrho = \sum_{i=1}^{d^n} \lambda_{i} \vu_i \vu_i^\dagger$. Using the same argument as in \cite{KuengACHA17}, we can obtain that
\begin{eqnarray}
    \label{lower bound of second expectation in rank-one gaussian}
\E[|\langle \va\va^\dagger,\vrho \rangle|^2]  \ge 1.
\end{eqnarray}

Thus, we can set $\xi=\frac{1}{2}$. There exists a universal constant $c_0$ such that
\begin{eqnarray}
    \label{ROSubgaussian_Measurement_H_conclusion}
&\!\!\!\!\!\!\!\!&\P{|\langle \va\va^\dagger,\vrho \rangle|^2 \ge \frac{1}{2} } \ge c_0, \ \forall \vrho\in\ol\setX_{\ol{r}} \nonumber\\
&\!\!\!\!\!\!\!\!& \hspace{0.5cm} \Rightarrow \quad H_{\xi}(\ol\setX_{\ol{r}}) \ge c_0.
\end{eqnarray}

  \item{Upper bound of $W(\ol\setX_{\ol{r}})$:} As discussed above, $ \{\<\epsilon_{k}\va_{k}\va_{k}^\dagger, \vrho  \>\}_{k=1}^{K}$ are independent subexponential random variables since $\|\<\epsilon_{k}\va_{k}\va_{k}^\dagger, \vrho  \>\|_{\psi_1} \leq c_1\|\vrho\|_F\|\va_k\|_{\psi_2}^2  \leq c_2$ \cite{zajkowski2020bounds} where $c_1, c_2$ are some universal constants and the second inequality follows from $\|\vrho\|_F=1$ and $\|\va_k\|_{\psi_2}\leq O(1)$. In addition, we have $\E \<\epsilon_{k}\va_{k}\va_{k}^\dagger, \vrho  \> = 0$ because of the Rademacher random variables $\epsilon_{k}$.

  By the analysis in the covering argument of  {Appendix} \ref{ProofOf<H,X>forSubgaussian}, when $K = \Omega(n d^2 \ol{r}^2\log n)$,  we have
\begin{eqnarray}
    \label{ROSubgaussian_Measurement_W_1}
    W(\ol\setX_{\ol{r}}) &\!\!\!\!=\!\!\!\!&\E \sup_{\vrho\in \ol\setX_{\ol{r}}}\frac{1}{\sqrt{K}}\sum_{k=1}^K \<\epsilon_{k}\va_{k}\va_{k}^\dagger, \vrho  \>\nonumber\\
    &\!\!\!\!\leq\!\!\!\!& c_3 d\ol{r}\sqrt{n \log n},
\end{eqnarray}
where $c_3$ is a positive constant.

  \item{Contraction:} Combining \eqref{ROSubgaussian_Measurement_H_conclusion} and \eqref{ROSubgaussian_Measurement_W_1}, we set $t=\frac{c_0\sqrt{K}}{2}$ and $K\geq \frac{256c_3^2 nd^2 \ol{r}^2\log n }{c_0^2}$ in \eqref{Final_Conclusion_original}, then get
\begin{eqnarray}
    \label{ROSubgaussian_Measurement_1_final}
    &\!\!\!\!\!\!\!\!&\inf_{\vrho\in\ol\setX_{\ol{r}}}\bigg(\sum_{k=1}^K|\< \va_{k}\va_{k}^\dagger, \vrho \>  |^2   \bigg)^{\frac{1}{2}}\nonumber\\
    &\!\!\!\!\geq\!\!\!\!& \xi \sqrt{K} H_{\xi}(\ol\setX_{\ol{r}}) -2W(\ol\setX_{\ol{r}}) -t\xi\nonumber\\
    &\!\!\!\!\geq\!\!\!\!&\frac{c_0\sqrt{K}}{2}-2c_3 d\ol{r}\sqrt{n\log n}-\frac{t}{2}\geq\frac{c_0\sqrt{K}}{8}.
\end{eqnarray}
with probability $1-e^{-\frac{c_0K}{8}}$.

\end{itemize}
This completes the proof of \Cref{Small_Method_ROSubgaussian_Measurement}.
\end{proof}

\subsection{Proof of the upper bound for $W(\ol\setX_{\ol{r}})$ in \eqref{ROSubgaussian_Measurement_W_1}}
\label{ProofOf<H,X>forSubgaussian}

\begin{proof}

In this section, we apply a covering argument to prove \eqref{ROSubgaussian_Measurement_W_1}.
For an MPO $\vrho$ of the form \eqref{DefOfMPO}, for simplicity, we denote it by $\vrho = [\mX_1,\dots, \mX_n]\in\ol\setX_{\ol{r}}$.
Also denote $\mA_k = \epsilon_{k}\va_{k}\va_{k}^\dagger$.
For each set of matrices $\{L(\mX_\ell)\in\R^{d^2r_{\ell-1}\times r_\ell}: \|L(\mX_\ell)\|\leq 1\}$ ($r_0=1$), according to \cite{zhang2018tensor}, we can construct an $\epsilon$-net $\{L(\mX_\ell^{(1)}), \dots, L(\mX_\ell^{(N_\ell)})  \}$ with the covering number $N_\ell\leq (\frac{4+\epsilon}{\epsilon})^{d^2r_{\ell-1}r_\ell}$ such that
\begin{eqnarray}
    \label{ProofOf<H,X>forSubgaussian_proof1}
    \sup_{L(\mX_\ell): \|L(\mX_\ell)\|\leq 1}~\min_{p_\ell\leq N_\ell} \|L(\mX_\ell)-L(\mX_\ell^{(p_\ell)})\|\leq \epsilon,
\end{eqnarray}
for all $\ell=1,\dots, n-1$.
Also, we can construct an $\epsilon$-net $\{ L(\mX_n^{(1)}), \dots, L(\mX_n^{(N_n)}) \}$ for $\{L(\mX_n)\in\R^{d^2r_{n-1}\times 1}: \|L(\mX_n)\|_F\leq 1  \}$ such that
\begin{eqnarray}
    \label{ProofOf<H,X>forSubgaussian_proof2}
    \sup_{L(\mX_n): \|L(\mX_n)\|_F\leq 1}\min_{p_n\leq N_n} \|L(\mX_n)-L(\mX_n^{(p_n)})\|_F\leq \epsilon,
\end{eqnarray}
with the covering number $N_n\leq (\frac{2+\epsilon}{\epsilon})^{d^2r_{n-1}}$. Note that different from \Cref{lem:cover-number-mps} that uses the $\|\cdot\|_{1,2}$ norm, here we use the spectral norm $\|\cdot\|$ and Frobenius norm $\|\cdot\|_F$ to define the covering numbers.

For simplicity, we use $\calI$ to denote the index set $[N_1]\times \cdots \times [N_n]$. Denote by
\begin{align*}
    [\mX_1^\star,\dots, \mX_n^\star] &:= \hspace{-0.9cm}
     \argmax_{\mbox{\tiny$\begin{array}{c} L(\mX_\ell)\in\R^{d^2r_{\ell-1}\times r_\ell}\\
     \|L(\mX_\ell)\|\leq 1, \ell=1,\dots, n-1\\
     \|L(\mX_n)\|_F\leq 1 \end{array}$}}\hspace{-0.5cm}  \frac{1}{\sqrt{K}}\sum_{k=1}^K\<\mA_k, [\mX_1,\dots, \mX_n]  \>,\\
    T & := \frac{1}{\sqrt{K}}\sum_{k=1}^K\<\mA_k,  [\mX_1^\star,\dots, \mX_N^\star]  \>.
\end{align*}

According to the construction of the $\epsilon$-nets, there exists $p=(p_1,\dots, p_n)\in\calI$ such that
\begin{eqnarray}
    \label{ProofOf<H,X>forSubgaussian_proof5}
    &&\|L(\mX_\ell^\star) - L(\mX_\ell^{(p_\ell)}) \|\leq\epsilon, \ \  \ell=1,\dots,n-1 \nonumber\\
    && \hspace{-0cm} \text{and}  \ \ \  \|L(\mX_n^\star) - L(\mX_n^{(p_n)})\|_F\leq\epsilon.
\end{eqnarray}

Now taking $\epsilon=\frac{1}{2n}$ gives
\begin{align} \label{ProofOf<H,X>forSubgaussian_proof6}
    T&=\frac{1}{\sqrt{K}}\sum_{k=1}^K\<\mA_k, [\mX_1^{(p_1)},\dots, \mX_n^{(p_n)}]  \>\nonumber\\
    &+\frac{1}{\sqrt{K}}\sum_{k=1}^K\bigg\<\mA_k, [\mX_1^\star,\dots, \mX_n^\star] - [\mX_1^{(p_1)},\dots, \mX_n^{(p_n)}]  \bigg\>\nonumber\\
    &=\frac{1}{\sqrt{K}}\sum_{k=1}^K\<\mA_k, [\mX_1^{(p_1)},\dots, \mX_n^{(p_n)}]  \>+ \frac{1}{\sqrt{K}}\sum_{k=1}^K\bigg\<\mA_k,\nonumber\\
    &\sum_{a_1=1}^n[\mX_1^{(p_1)},\dots,\! \mX_{a_1-1}^{(p_1)}, \mX_{a_1}^{(p_{a_1})}\! -\! \mX_{a_1}^\star, \mX_{a_1+1}^\star,  \dots, \! \mX_n^\star] \bigg\>\nonumber\\
    &\leq \frac{1}{\sqrt{K}}\sum_{k=1}^K\<\mA_k, [\mX_1^{(p_1)},\dots, \mX_n^{(p_n)}]  \> + n \epsilon T\nonumber\\
    &=\frac{1}{\sqrt{K}}\sum_{k=1}^K\<\mA_k, [\mX_1^{(p_1)},\dots, \mX_n^{(p_n)}]  \>+\frac{T}{2},
\end{align}
where we write $[\mX_1^\star,\dots, \mX_n^\star] - [\mX_1^{(p_1)},\dots, \mX_n^{(p_n)}]$ in the second line as the sum of $n$ terms
according to \Cref{EXPANSION_A1TOAN-B1TOBN_1}.

Notice that for any $\{L(\mX_\ell)\}_{\ell\leq n-1}$ and $L(\mX_n) $, where $\|L(\mX_\ell)\|\leq 1$ and $\|L(\mX_n)\|_F\leq 1$, we have $\|\vrho\|_F=\|[\mX_1,\dots, \mX_n]\|_F\leq 1$. As in the discussion in {Appendix} \ref{Proof_Of_MPO_Subgaussian}, $\< \mA_k, \vrho\>=\<\epsilon_{k}\va_{k}\va_{k}^\dagger, \vrho  \>$ is a centered subexponential random variable with subexponential norm of order $O(1)$, so we can use \Cref{Sub-exponentialInequality} to get
\begin{eqnarray}
    \label{ProofOf<H,X>forSubgaussian_proof8}
    &\!\!\!\!\!\!\!\!&\P{\bigg|\frac{1}{\sqrt{K}}\sum_{k=1}^K\<\mA_k,  [\mX_1^{(p_1)},\dots, \mX_n^{(p_n)}] \>\bigg|\geq t}\nonumber\\
    &\!\!\!\!\leq\!\!\!\!& e^{1-c_1\min\{ \frac{t^2}{c_2^2}, \frac{t\sqrt{K}}{c_2}  \}},
\end{eqnarray}
where $c_1$ and $c_2$ are constants.

Combining \eqref{ProofOf<H,X>forSubgaussian_proof6} and \eqref{ProofOf<H,X>forSubgaussian_proof8} together yields
\begin{eqnarray} \label{ProofOf<H,X>forSubgaussian_proof9}
    &\!\!\!\!\!\!\!\!&\P{T\geq t}\nonumber\\
    &\!\!\!\!\leq\!\!\!\!& \P{\max_{p_1,\dots, p_n}\bigg|\frac{1}{\sqrt{K}}\sum_{k=1}^K\<\mA_k,  [\mX_1^{(p_1)},\dots, \mX_n^{(p_n)}] \>\bigg|\geq \frac{t}{2}}\nonumber\\
    &\!\!\!\!\leq\!\!\!\!& \parans{\prod_{i=1}^n N_i} e^{1-c_1\min\{ \frac{t^2}{c_2^2}, \frac{t\sqrt{K}}{c_2}  \}}\nonumber\\
    &\!\!\!\!\leq\!\!\!\!&\bigg(\frac{4+\epsilon}{\epsilon}\bigg)^{d^2r_1+\sum_{i=2}^{n-1}d^2r_{i-1}r_i+d^2r_{n-1}}e^{1-c_1\min\{ \frac{t^2}{c_2^2}, \frac{t\sqrt{K}}{c_2}  \}}\nonumber\\
    &\!\!\!\!\leq\!\!\!\!&e^{1-c_1\min\{ \frac{t^2}{c_2^2}, \frac{t\sqrt{K}}{c_2}  \}+C n d^2\ol{r}^2\log n},\nonumber
\end{eqnarray}
where $\ol{r}=\max_{i=1,\dots,n-1}r_i$, $C$ is a universal constant, and the last line uses $\frac{4+\epsilon}{\epsilon}=\frac{4+\frac{1}{2n}}{\frac{1}{2n}}=8n+1$ based on the assumption $\epsilon=\frac{1}{2n}$ in \eqref{ProofOf<H,X>forSubgaussian_proof6}. Now choosing $K = c_3 n d^2 \ol{r}^2\log n$ with a positive constant $c_3$ and plugging this into the above equation, we can find constants $c_4$ and $c_5$ such that
\begin{align*}
 \P{T\geq t} \le e^{-c_4 t\sqrt{n d^2 \ol{r}^2\log n}  },  \ \forall \ t \ge c_5 \sqrt{n d^2 \ol{r}^2\log n},
\end{align*}
which further implies that
\begin{align}
\label{ProofOf<H,X>forSubgaussian_proof9_1}
& W(\ol\setX_{\ol{r}})= \E T \nonumber\\
    &\leq c_5 \sqrt{n d^2 \ol{r}^2\log n} +  \int_{c_5 \sqrt{n d^2 \ol{r}^2\log n}}^{\infty}\P{T\geq t} dt\nonumber\\
    &\le c_5 \sqrt{n d^2 \ol{r}^2\log n} +  \int_{c_5 \sqrt{n d^2 \ol{r}^2\log n}}^{\infty} e^{-c_4 t\sqrt{n d^2 \ol{r}^2\log n} } dt\nonumber\\
    &\leq  c_6d\bar{r}\sqrt{n\log n},
\end{align}
where $c_6$ is a positive constant.
\end{proof}

\section{Proof of \Cref{Small Ball Method_abs_New}}
\label{Proof_Lemma_abs_small ball new}

\begin{proof}
First, we introduce a directional version of the marginal tail function:
\begin{align}
    \label{marginal_tail_function_New11}
    H_{\xi}(E;\vb) = \frac{1}{K} \sum_{k=1}^K\mathbb{P} \{ |\<\vb_{k},\vu \>|\geq \xi  \},\ \text{for}\ \vu\in E \ \text{and} \ \xi>0.
\end{align}

Lyapunov’s inequality and Markov’s inequality give the following bounds
\begin{eqnarray}
    \label{LyapunovAndMarkov_New}
    &\!\!\!\!\!\!\!\!&\bigg(\frac{1}{QK}\sum_{i=1}^Q\sum_{k=1}^K|\<\vb_{i,k}, \vu  \>|^2  \bigg)^{\frac{1}{2}}\nonumber\\
    &\!\!\!\!\geq\!\!\!\!& \frac{1}{QK}\sum_{i=1}^Q\sum_{k=1}^K|\<\vb_{i,k}, \vu  \>|\nonumber\\
    &\!\!\!\!\geq\!\!\!\!& \frac{\xi}{QK} \sum_{i=1}^Q\sum_{k=1}^K\mathbbm{1}(|\<\vb_{i,k}, \vu  \>|\geq \xi  ),
\end{eqnarray}
where we write $\mathbbm{1}(A)$ for the $0-1$ random variable that indicates whether the event $A$ takes place.
Add and subtract $H_{2\xi}(E;\vb)$ inside the sum, and then take the infimum over $\vu\in E$ to reach the inequality
\begin{eqnarray}
    \label{LyapunovAndMarkov_Q_New}
    &&\hspace{-1.5cm}\inf_{\vu\in E}\bigg(\frac{1}{QK}\sum_{i=1}^Q\sum_{k=1}^K|\<\vb_{i,k}, \vu  \>|^2  \bigg)^{\frac{1}{2}} \geq \xi\inf_{\vu\in E}H_{2\xi}(E;\vb)\nonumber\\
    &&\hspace{-1.5cm} -  \frac{\xi}{Q}\sup_{\vu\in E}\sum_{i=1}^Q\bigg[H_{2\xi}(E;\vb) - \frac{1}{K}\sum_{k=1}^K\mathbbm{1}(|\<\vb_{i,k}, \vu  \>|\geq \xi  )\bigg].
\end{eqnarray}
Observe that each summand over index $i$ at the RHS is independent and bounded in magnitude by $1$. Therefore, based on \cite[Section 6.1]{boucheron2013concentration}, we have
\begin{eqnarray}
    \label{LyapunovAndMarkov_Q_Supremum_New}
    &\!\!\!\!\!\!\!\!&\hspace{-0.6cm} \sup_{\vu\in E}\sum_{i=1}^Q\left[H_{2\xi}(E;\vb) - \frac{1}{K}\sum_{k=1}^K\mathbbm{1}(|\<\vb_{i,k}, \vu  \>|\ge\xi  )\right]\nonumber\\
    &\!\!\!\!\!\!\!\!&\hspace{-0.6cm}\leq \E \sup_{\vu\in E}\sum_{i=1}^Q\! \bigg[H_{2\xi}(E;\vb)\! - \! \frac{1}{K}\sum_{k=1}^K\mathbbm{1}(|\<\vb_{i,k}, \vu  \>|\geq \xi  )\bigg]\! +\! t\sqrt{Q},\nonumber\\
\end{eqnarray}
with probability at least $1-e^{-\frac{t^2}{2}}$.

Next, we simplify the expected supremum. Introduce a soft indicator function:
$$\phi_{\xi}: \R\to [0,1] \ \text{where} \ \phi_{\xi}(s)=\begin{cases} 0, & |s|\leq \xi, \\ (|s|-\xi)/\xi, & \xi<|s|\leq 2\xi, \\ 1, &2\xi < |s|. \end{cases}$$

According to \cite{tropp2015convex}, we can derive
\begin{eqnarray}
    \label{LyapunovAndMarkov_Q_Supremum_Expectation_New}
    &\!\!\!\!\!\!\!\!\!\!&\E \sup_{\vu\in E}\sum_{i=1}^Q\bigg[H_{2\xi}(E;\vb) - \frac{1}{K}\sum_{k=1}^K\mathbbm{1}(|\<\vb_{i,k}, \vu  \>|\geq \xi  )\bigg]\nonumber\\
    &\!\!\!\!\!=\!\!\!\!\!& \frac{1}{K}\E \sup_{\vu\in E}\sum_{i=1}^Q\sum_{k=1}^K\bigg[\E \mathbbm{1}(|\<\vb_{k}, \vu  \>|\geq 2\xi  )\! - \! \mathbbm{1}(|\<\vb_{i,k}, \vu  \>|\geq \xi  )\bigg]\nonumber\\
    &\!\!\!\!\!\leq\!\!\!\!\!&\frac{1}{K}\E \sup_{\vu\in E}\sum_{i=1}^Q\bigg[\E \sum_{k=1}^K\phi_{\xi}\left(\<\vb_{k}, \vu  \>\right) - \sum_{k=1}^K\phi_{\xi}\left(\<\vb_{i,k}, \vu  \>   \right)\bigg]\nonumber\\
    &\!\!\!\!\!\leq\!\!\!\!\!& \frac{2}{K}\E \sup_{\vu\in E}\sum_{i=1}^Q\epsilon_{i}\sum_{k=1}^K\phi_{\xi}(\<\vb_{i,k}, \vu  \>)  \nonumber\\
    &\!\!\!\!\!\leq\!\!\!\!\!&\frac{2}{\xi K}\E \sup_{\vu\in E}\sum_{i=1}^Q\sum_{k=1}^K\epsilon_{i}\<\vb_{i,k}, \vu  \>,
\end{eqnarray}
where in the first equation, we write the marginal tail function as an expectation, and then we bound the two indicators using the soft indicator function.
In the second inequality, where $\epsilon_i,i=1,\dots,Q$ are independent Rademacher random variables that are independent from everything else, we  use the Giné–Zinn symmetrization \cite[Lemma 2.3.1]{wellner2013weak} due to the independence of $\sum_{k=1}^K\phi_{\xi}\left(\<\vb_{i,k}, \vu  \>\right)$ for $ i=1,\dots,Q$.
In the last line, due to the contraction of $\xi\phi_{\xi}$, we apply the Rademacher comparison principle \cite[Eqn.(4.20)]{ledoux1991probability}.

Hence, we have
\begin{align}
    \label{Final_Conclusion_1_New}
    &\!\!\!\!\!\!\!\!&\inf_{\vu\in E}\bigg(\frac{1}{QK}\sum_{i=1}^Q\sum_{k=1}^K|\<\vb_{i,k}, \vu  \>|^2  \bigg)^{\frac{1}{2}} \geq \xi\inf_{\vu\in E}H_{2\xi}(E;\vb)\nonumber\\
    &\!\!\!\!\!\!\!\!&-  \frac{\xi}{Q}\bigg[\frac{2}{\xi K}\E \sup_{\vu\in E}\sum_{i=1}^Q\sum_{k=1}^K\epsilon_{i}\<\vb_{i,k}, \vu  \>+t\sqrt{Q}\bigg].
\end{align}
Letting $\vh =\frac{1}{\sqrt{Q K}}\sum_{i=1}^Q\sum_{k=1}^K \epsilon_{i}\vb_{i,k}$, we can finally obtain
\begin{eqnarray}
    \label{Final_Conclusion_2_New}
    &\!\!\!\!\!\!\!\!&\inf_{\vu\in E}\bigg(\sum_{i=1}^Q\sum_{k=1}^K|\<\vb_{i,k}, \vu  \>|^2  \bigg)^{\frac{1}{2}} \geq \xi\sqrt{QK}\inf_{\vu\in E}H_{2\xi}(E;\vb)\nonumber\\
    &\!\!\!\!\!\!\!\!&\hspace{3.7cm} -2\E \sup_{\vu\in E}\<\vh ,\vu\> -t \xi\sqrt{K}.
\end{eqnarray}
This completes the proof of \Cref{Small Ball Method_abs_New}.
\end{proof}

\section{Proof of \Cref{Small_Method_ROHaar_Measurement}}
\label{ProofOfSmallBallMethod_For_Haar_Measurement}

\begin{proof} We prove \Cref{Small_Method_ROHaar_Measurement} using the modified Mendelson's small ball method. Let $\{ \vphi_{1},\dots, \vphi_{K} \}$ be the first $K$ columns of a randomly generated Haar distributed unitary matrix, and let $\{ \vphi_{i,1},\dots, \vphi_{i,K} \}_{i=1}^Q$ be independent copies of $\{ \vphi_{1},\dots, \vphi_{K} \}$. According to \Cref{Small Ball Method_abs_New},  we need to  bound
\begin{eqnarray} \label{ProofOfSmallBallMethod_For_Haar_Measurement_H}
    H_{\xi}(\ol\setX_{\ol{r}}) = \inf_{\vrho\in\ol\setX_{\bar{r}}} \frac{1}{K}\sum_{k=1}^K\mathbb{P} \{|\< \vphi_k\vphi_k^\dagger,\vrho \>|\geq \xi  \}
\end{eqnarray}
and
\begin{eqnarray}
    \label{ProofOfSmallBallMethod_For_Haar_Measurement_W}
    W(\ol\setX_{\ol{r}}) = \E \sup_{\vrho\in \ol\setX_{\bar{r}}}\frac{1}{\sqrt{QK}}\sum_{i=1}^Q\sum_{k=1}^K \<\epsilon_i\vphi_{i,k}\vphi_{i,k}^\dagger, \vrho  \>,
\end{eqnarray}
where $\epsilon_i,i=1,\dots,Q$ are indepdent Rademacher random variables. Below we study the two quantities separately.

\begin{itemize}
  \item{Lower bound of $H_{\xi}(\ol\setX_{\ol{r}})$}: As in {Appendix} \ref{Proof_Of_MPO_Subgaussian}, we also use the Paley-Zygmund inequality (\Cref{Foucart13MathIntro_lemma716}) to bound $H_{\xi}(\ol\setX_{\ol{r}})$. Specifically,
\begin{eqnarray}
    \label{ProofOfSmallBallMethod_For_Haar_Measurement_H_1}
    &\!\!\!\!\!\!\!\!&H_{\xi}(\ol\setX_{\bar{r}})= \inf_{\vrho\in\ol\setX_{\bar{r}}} \frac{1}{K}\sum_{k=1}^K \P{|\< \vphi_k\vphi_k^\dagger,\vrho \>|\geq \xi  }\nonumber\\
    &\!\!\!\!=\!\!\!\!&\inf_{\vrho\in\ol\setX_{\bar{r}}} \frac{1}{K}\sum_{k=1}^K\P{|\< \vphi_k\vphi_k^\dagger,\vrho \>|^2\geq \xi^2  }\nonumber\\
    &\!\!\!\!\geq\!\!\!\!&\inf_{\vrho\in\ol\setX_{\bar{r}}} \frac{1}{K}\sum_{k=1}^K\P{|\< \vphi_k\vphi_k^\dagger,\vrho \>|^2\geq \frac{1}{2}\E[|\< \vphi_k\vphi_k^\dagger,\vrho \>|^2]  }\nonumber\\
    &\!\!\!\!\geq\!\!\!\!&\inf_{\vrho\in\ol\setX_{\bar{r}}} \frac{1}{K}\sum_{k=1}^K\frac{(\E[|\< \vphi_k\vphi_k^\dagger,\vrho \>|^2])^2}{4\E[|\< \vphi_k\vphi_k^\dagger,\vrho \>|^4]} \geq c_0,\nonumber\\
    &\!\!\!\!\!\!\!\!& \forall \xi\le \sqrt{ \frac{1}{2}\E[|\< \vphi_k\vphi_k^\dagger,\vrho \>|^2] },
\end{eqnarray}
where the first inequality follows because $\P{|\< \vphi_k\vphi_k^\dagger,\vrho \>|^2 \geq \xi^2  }$ is a decreasing function with respect to $\xi$, the second inequality uses the Paley-Zygmund inequality (\Cref{Foucart13MathIntro_lemma716}) for $|\< \vphi_k\vphi_k^\dagger,\vrho \>|^2$, and the last inequality uses \Cref{lem:hypercontractivity}. Below we show that $\abs{\< \vphi_k\vphi_k^\dagger,\vrho \>}$ is a subexponential random variable and hence satisfies the requirements for both \Cref{Foucart13MathIntro_lemma716} and \Cref{lem:hypercontractivity}. According to the process of Gram-Schmidt orthogonalization for obtaining a Haar-distributed unitary matrix, $\vphi_1$ can be obtained by normalizing a standard complex normal random vector from the distribution $\calC\calN({\bm 0},\mId_{d^N})$. Using $\|\sqrt{d^n}\vphi_1\|_{\psi_2}\leq O(1)$ \cite{vershynin2018high}, we have
\begin{align}\|\< \vphi_1\vphi_1^\dagger,\vrho \>\|_{\psi_1}\leq \frac{c}{d^n}\|\sqrt{d^n}\vphi_1\|_{\psi_2}^2\|\vrho\|_F = O(\frac{1}{d^n})
\label{eq:subexp-phiRho}\end{align}
for some constant $c$ and hence $\abs{\< \vphi_1\vphi_1^\dagger,\vrho \>}$ is a subexponential random variable. Finally according to \cite{petz2004asymptotics}, because all the entries in a Haar-distributed unitary matrix have the same distribution due to the translation invariance of \Cref{Stastitic property of UV}, we conclude that each $\abs{\< \vphi_k\vphi_k^\dagger,\vrho \>}$ is a subexponential random variable for all $k = 1,\ldots, d^n$.

To complete the proof of this part, we now study $\E[|\< \vphi_k\vphi_k^\dagger,\vrho \>|^2]$, controlling the upper bound of $\xi$ in \eqref{ProofOfSmallBallMethod_For_Haar_Measurement_H_1}. Towards that goal, for any $\vrho\in\ol\setX_{\ol r}$, we denote its eigenvalue decomposition by $\vrho=\sum_{i=1}^{d^n}\lambda_i \vu_i \vu_i^\dagger$, where $\{\vu_i\}$ are unitary vectors and $\{\lambda_i\}$ are the eigenvalues with $\sum_{i=1}^{d^n} \lambda_i^2=1$.
Now, following \eqref{Unitary Distrubuted_2 another}, we have
\begin{eqnarray}
    \label{Unitary Distrubuted_2 other}
    \E[|\< \vphi_k\vphi_k^\dagger,\vrho \>|^2]  &\!\!\!\!=\!\!\!\!& \sum_{j=1}^{d^n}\sum_{l=1}^{d^n}\frac{\lambda_j\lambda_l}{d^{2n}} + \sum_{l=1}^{d^n}\frac{d^n-1}{d^{2n}(d^n+1)}\lambda_l^2\nonumber\\ &\!\!\!\!=\!\!\!\!&\frac{(\sum_{l}\lambda_l)^2}{d^{2n}} + \frac{d^n-1}{d^{2n}(d^n+1)}\|\vrho\|_F^2\nonumber\\
    &\!\!\!\!\geq\!\!\!\!&  \frac{d^n-1}{d^{2n}(d^n+1)}.
\end{eqnarray}
This together with \eqref{ProofOfSmallBallMethod_For_Haar_Measurement_H_1} further implies that
\begin{align}
\label{ProofOfSmallBallMethod_For_Haar_Measurement_H_1_v2}
    H_{\xi}(\ol\setX_{\ol{r}}) \ge c_0, \ \forall \xi \le \frac{c_1}{d^n}
\end{align}
for some positive constant $c_1$.

\item{Upper bound of $W(\ol\setX_{\ol{r}})$:}
Since each $\< \vphi_{i,k}\vphi_{i,k}^\dagger,\vrho \>$ is a subexponetial random variable with $\|\< \vphi_{i,k}\vphi_{i,k}^\dagger,\vrho \>\|_{\psi_1} = O(\frac{1}{d^n})$ according to \eqref{eq:subexp-phiRho}, $\epsilon_i\vphi_{i,k}^\dagger\vrho\vphi_{i,k}$ is a centered subexponential random variable with the subexponential norm $\|\epsilon_i\vphi_{i,k}^\dagger\vrho\vphi_{i,k}\|_{\psi_1}= O(\frac{1}{d^n})$. On the other hand, for any $i$, the random vectors $\vphi_{i,k}$ and $\vphi_{i,k'}$ are not dependent to each other for $k\neq k'$. Thus, we use \Cref{Concentration inequality of sum of dependent sub-exp sum} to obtain its concentration inequality as

\begin{eqnarray}
    \label{ProofOfSmallBallMethod_For_Haar_Measurement_W_1}
     &&\!\!\!\!\!\!\!\!\P{ \frac{1}{\sqrt{QK}}\sum_{i=1}^Q\sum_{k=1}^K \<\epsilon_{i}\vphi_{i,k}\vphi_{i,k}^\dagger, \vrho  \> \geq t}\nonumber\\
     &\!\!\!\!\leq\!\!\!\!& \begin{cases}
    \bigg(\frac{4+\epsilon}{\epsilon}\bigg)^{Cnd^2\ol{r}^2\log n}e^{-\frac{c_2 d^{2n}t^2}{4K}}, & t\leq \frac{c_4\sqrt{QK} }{d^n}\\
    \bigg(\frac{4+\epsilon}{\epsilon}\bigg)^{Cnd^2\ol{r}^2\log n}e^{-\frac{c_3 \sqrt{Q} d^nt}{2\sqrt{K}}}, & t > \frac{c_4\sqrt{QK} }{d^n}
  \end{cases} \nonumber\\
    &\!\!\!\!\leq\!\!\!\!&\begin{cases}
    e^{-\frac{c_2 d^{2n}t^2}{4K} +Cnd^2\ol{r}^2\log n}, & t\leq \frac{c_4\sqrt{QK} }{d^n}\\
    e^{-\frac{c_3 \sqrt{Q} d^nt}{2\sqrt{K}} +Cnd^2\ol{r}^2\log n}, & t > \frac{c_4\sqrt{QK} }{d^n}
  \end{cases}\nonumber\\
    &\!\!\!\!\leq\!\!\!\!& e^{-\min\{\frac{c_2 d^{2n}t^2}{4K},\frac{c_3 \sqrt{Q} d^nt}{2\sqrt{K}} \} + Cnd^2\ol{r}^2\log n },
\end{eqnarray}
where We utilize $d^2r_1+\sum_{i=2}^{n-1}d^2r_{i-1}r_i+d^2r_{n-1}\leq Cnd^2\ol{r}^2\log n$ in the first inequality for a universal constant $C$. Subsequently, we set  $\epsilon=\frac{1}{2n}$ in the second inequality. Furthermore, $\ol{r}=\max_{i=1,\dots,n-1}r_i$, and $c_2$, $c_3$, $c_4$ are positive constants.
Following the same analysis of \eqref{ProofOf<H,X>forSubgaussian_proof9_1}, when $Q = \Omega(n d^2 \ol{r}^2\log n)$, we have
\begin{eqnarray}
    \label{ProofOfSmallBallMethod_For_Haar_Measurement_W_3}
    W(\ol\setX_{\ol{r}}) &\!\!\!\!=\!\!\!\!& \E \sup_{\vrho\in \ol\setX_{\ol{r}}}\frac{1}{\sqrt{QK}}\sum_{i=1}^Q\sum_{k=1}^K \<\epsilon_{i}\vphi_{i,k}\vphi_{i,k}^\dagger, \vrho  \>\nonumber\\
    &\!\!\!\!\leq\!\!\!\!& c_5 \frac{\sqrt{K} d\ol{r}\sqrt{n\log n}}{d^n},
\end{eqnarray}
where $c_5$ is a universal constant.

  \item{Contraction:} Combining \eqref{ProofOfSmallBallMethod_For_Haar_Measurement_H_1_v2}  and \eqref{ProofOfSmallBallMethod_For_Haar_Measurement_W_3}, and setting $t=\frac{c_0\sqrt{Q}}{2}$, $\xi =\frac{c_1}{d^n}$,   and $Q\geq \frac{64c_5^2 nd^2 \ol{r}^2(\log n)}{c_0^2c_1^2}$, we get
\begin{eqnarray}
    \label{ProofOfSmallBallMethod_For_Haar_Measurement_final}
    &\!\!\!\!\!\!\!\!&\inf_{\vrho\in\ol\setX_{\ol{r}}}\bigg(\sum_{i=1}^Q\sum_{k=1}^K|\< \vphi_{i,k}\vphi_{i,k}^\dagger, \vrho \>  |^2   \bigg)^{\frac{1}{2}} \nonumber\\
    &\!\!\!\!\geq\!\!\!\!& \xi \sqrt{QK} H_{\xi}(\ol\setX_{\ol{r}}) -2W(\ol\setX_{\ol{r}}) -t\xi\sqrt{K}\nonumber\\
    &\!\!\!\!\geq\!\!\!\!&\frac{c_0c_1\sqrt{QK}}{d^n}-2c_5 \frac{\sqrt{K} d\ol{r}\sqrt{n\log n}}{d^n}-\frac{c_1 \sqrt{QK}}{d^n}\nonumber\\
    &\!\!\!\!\geq\!\!\!\!&\frac{c_0c_1\sqrt{QK}}{4d^{n}}
\end{eqnarray}
with probability $1-e^{-\Omega(Q)}$.

\end{itemize}
This completes the proof of \Cref{Small_Method_ROHaar_Measurement}.
\end{proof}

\section{Proof of \Cref{Statistical Error_Haar_Measurement}}
\label{Proof of Statistical error in Haar}

\begin{proof}
Before deriving \Cref{Statistical Error_Haar_Measurement}, we restate our model.
We first randomly generate $Q$  Haar distributed unitary matrices $\begin{bmatrix} \vphi_{i,1} & \cdots & \vphi_{i,d^n} \end{bmatrix}$, which induce  $Q$ POVMs of form $
    \big\{\vphi_{i,1}\vphi_{i,1}^\dagger,\dots, \vphi_{i,d^n}\vphi_{i,d^n}^\dagger  \big\}, i=1,\dots,Q.$ Recalling \eqref{The defi of population measurement in Q cases sec:4} and \eqref{eq:empirical-prob in Q cases sec:4}, we have population measurements for the unknown quantum state $\vrho^\star$ and  total empirical measurements given by ${{\bm p}}^Q= \calA^Q(\vrho^\star) $ and ${\widehat{\bm p}}^Q$.
We then define the statistical measurement error as
\begin{eqnarray}
    \label{Noise_Haar Measurement1}
    \veta = \wh{\vp}^Q -  \vp^Q = \wh{\vp}^Q - \calA^Q(\vrho^\star) =  \begin{bmatrix}
          \veta_{1}^\top,
          \cdots,
          \veta_{Q}^\top\end{bmatrix}^\top,
\end{eqnarray}
where $\eta_{i,k}$ is the $k$-th element in $\veta_i$. With $\wh \vp^Q$, we estimate the unknown state $\vrho^\star$ by solving the following constrained least-squares problem
\begin{eqnarray}
    \label{The loss function in QST_haar}
    \wh{\vrho} = \arg\min_{\vrho\in\setX_{\bar{r}}}\|\calA^Q(\vrho) - {\widehat{\bm p}^Q}\|_2^2.
\end{eqnarray}
Following \eqref{whrho and rho star relationship}, we have
\begin{eqnarray}
    \label{whX and X^star relationship_1}
    \|\calA^Q(\wh{\vrho} - \vrho^\star)\|_2^2 \leq 2\<  \veta, \calA^Q(\wh{\vrho} - \vrho^\star) \>.
\end{eqnarray}

According to \Cref{{Small_Method_ROHaar_Measurement}}, given $Q\gtrsim nd^2\ol{r}^2(\log n)$, with probability at least $1-e^{-c_1Q}$, we have $\|\calA^Q(\wh{\vrho} - \vrho^\star)\|_2^2\gtrsim \frac{Q}{d^{n}}\|\wh{\vrho} - \vrho^\star\|_F^2$. Next, we will upper bound $\<  \veta, \calA^Q(\wh{\vrho} - \vrho^\star) \>$. Towards that goal, we first rewrite this term as
\begin{eqnarray}
    \label{upper bound of entire variable_haar}
    \<  \veta, \calA^Q(\wh{\vrho} - \vrho^\star) \> &\!\!\!\!\!\!=\!\!\!\!\!\!& \sum_{i=1}^Q\sum_{k=1}^{d^n} \eta_{i,k}\vphi_{i,k}^\dagger (\wh{\vrho} - \vrho^\star) \vphi_{i,k}\nonumber\\
    &\!\!\!\!\!\!\leq\!\!\!\!\!\!&  \|\wh{\vrho} - \vrho^\star\|_F \max_{\vrho\in \ol\setX_{2\ol{r}}} \sum_{i=1}^Q\sum_{k=1}^{d^n} \eta_{i,k}\vphi_{i,k}^\dagger \vrho \vphi_{i,k}.\nonumber\\
\end{eqnarray}

The rest of the proof is to bound $\max_{\vrho\in \ol\setX_{2\ol{r}}} \sum_{i=1}^Q\sum_{k=1}^{d^n} \eta_{i,k}\vphi_{i,k}^\dagger \vrho \vphi_{i,k}$, which will be achieved by using a covering argument.
First, when conditioned on $\{\vphi_{i,k},\forall i,k \}$, we consider any fixed value of $\widetilde\vrho$ and  apply \Cref{General bound of multinomial distribution Q cases1} to establish a concentration inequality for the expression $\sum_{i=1}^Q\sum_{k=1}^{d^n} \eta_{i,k}\vphi_{i,k}^\dagger \widetilde\vrho \vphi_{i,k}$.
Denote the event $F:=\{ \max_{i,k} |\vphi_{i,k}^\dagger\wt\vrho\vphi_{i,k}| \lesssim \frac{\log Q + n\log d}{d^{\nqbit }}$ $\}$ which holds with probability $\P{F} = 1-e^{-c_2 (\log Q + n\log d) }$ (its proof is given in {Appendix} \ref{Proof of concentration of multinomial dis}). Then we have
\begin{align}
    \label{tail function of of entire variable_haar original}
    \P{ \sum_{i=1}^Q\sum_{k=1}^{d^n} \eta_{i,k}\vphi_{i,k}^\dagger \widetilde\vrho \vphi_{i,k} \geq t \bigg| F  }\leq 2e^{-\frac{d^{2n}Mt^2}{c_3Q(\log Q + n\log d)^2}},
\end{align}
where $c_2$ and $c_3$ are positive constants. The formal proof of \eqref{tail function of of entire variable_haar original} is given in {Appendix} \ref{Proof of concentration of multinomial dis}.

Following the same analysis as in {Appendix} \ref{ProofOf<H,X>forSubgaussian}, there exists an $\epsilon$-net $\widetilde\setX_{2\ol{r}}$ of $\ol\setX_{2\ol{r}}$  such that
\begin{eqnarray}
    \label{Upper bound of entire variable_haar covering arument1}
    &\!\!\!\!\!\!\!\!&\P{\max_{\vrho\in \ol\setX_{2\ol{r}}} \sum_{i=1}^Q\sum_{k=1}^{d^n} \eta_{i,k}\vphi_{i,k}^\dagger \vrho \vphi_{i,k}\geq t \bigg| F  }\nonumber\\
    &\!\!\!\!\!\!\!\!&\hspace{-0.4cm} \leq\P{\max_{\widetilde\vrho\in \widetilde\setX_{2\ol{r}}} \sum_{i=1}^Q\sum_{k=1}^{d^n} \eta_{i,k}\vphi_{i,k}^\dagger \widetilde\vrho \vphi_{i,k}\geq \frac{t}{2} \bigg| F}\nonumber\\
    &\!\!\!\!\!\!\!\!&\hspace{-0.4cm}\leq\bigg(\frac{4+\epsilon}{\epsilon}\bigg)^{\hspace{-0.1cm}2d^2r_1+4\sum_{i=2}^{n-1}d^2r_{i-1}r_i+2d^2r_{n-1}}
    \hspace{-0.2cm}e^{-\frac{d^{2n}Mt^2}{c_3Q(\log Q + n\log d)^2} + \log 2}\nonumber\\
    &\!\!\!\!\!\!\!\!&\hspace{-0.4cm}\leq e^{-\frac{d^{2n}Mt^2}{c_3Q(\log Q + n\log d)^2}+C n d^2\ol{r}^2\log n + \log 2},
\end{eqnarray}
where $\epsilon=\frac{1}{2n}$ is chosen, $\ol{r}=\max_{i=1,\dots,n-1}r_i$, and $C$ is a universal constant in the last line. By taking $t=\hat{t} \triangleq \frac{c_4  \sqrt{Qn \log n}d \ol{r}(\log Q + n\log d)}{\sqrt{M}d^{n}}$ in the above equation, we further obtain
\begin{align}
    \label{Upper bound of entire variable_haar covering arument another1}
    \hspace{-0.1cm}\P{\max_{\vrho\in \ol\setX_{2\ol{r}}} \sum_{i=1}^Q\sum_{k=1}^{d^n} \eta_{i,k}\vphi_{i,k}^\dagger \vrho \vphi_{i,k}\leq \hat{t} \bigg| F  }\geq 1-e^{-c_5 \nqbit d^2 \ol{r}^2 \log n},
\end{align}
where $c_4$ and $c_5$ are  constants.

Now plugging in the probability for the event $F$, we finally get
\begin{eqnarray}
    \label{Upper bound of entire variable_haar covering arument another 11}
    &\!\!\!\!\!\!\!\!&\P{\max_{\vrho\in \ol\setX_{2\ol{r}}} \sum_{i=1}^Q\sum_{k=1}^{d^n} \eta_{i,k}\vphi_{i,k}^\dagger \vrho \vphi_{i,k}\leq \hat{t}}\nonumber\\
    &\!\!\!\!\geq\!\!\!\!& \P{\max_{\vrho\in \ol\setX_{2\ol{r}}} \sum_{i=1}^Q\sum_{k=1}^{d^n} \eta_{i,k}\vphi_{i,k}^\dagger \vrho \vphi_{i,k}\leq \hat{t} \cap F  }\nonumber\\
    &\!\!\!\!=\!\!\!\!&\P{F} \P{\max_{\vrho\in \ol\setX_{2\ol{r}}} \sum_{i=1}^Q\sum_{k=1}^{d^n} \eta_{i,k}\vphi_{i,k}^\dagger \vrho \vphi_{i,k}\leq  \hat{t} \bigg| F  }\nonumber\\
    &\!\!\!\!\geq\!\!\!\!& (1-e^{-c_2 \log(Qd^n) })(1-e^{-c_5 \nqbit d^2 \ol{r}^2 \log n})\nonumber\\
    &\!\!\!\!\geq\!\!\!\!& 1- e^{-c_2 (\log Q + n\log d) } - e^{-c_5 \nqbit d^2 \ol{r}^2 \log n}.
\end{eqnarray}

Hence, for $\<  \veta, \calA^Q(\wh{\vrho} - \vrho^\star) \>$ in \eqref{upper bound of entire variable_haar}, the above equation implies that with probability at least $1- e^{-c_2 (\log Q + n\log d) } - e^{-c_5 \nqbit d^2 \ol{r}^2 \log n}$,
\begin{eqnarray}
    \label{upper bound of entire variable_haar conclusion1}
    &\!\!\!\!\!\!\!\!&\<  \veta, \calA^Q(\wh{\vrho} - \vrho^\star) \> \nonumber\\
    &\!\!\!\!\leq\!\!\!\!& \frac{c_4  \sqrt{Qn \log n}d \ol{r}(\log Q + n\log d)}{\sqrt{M}d^{n}}\|\wh{\vrho} - \vrho^\star\|_F.
\end{eqnarray}
Combining this together with $\|\calA^Q(\wh{\vrho} - \vrho^\star)\|_2^2\gtrsim \frac{Q}{d^{n}}\|\wh{\vrho} - \vrho^\star\|_F^2$, we finally obtain
\begin{eqnarray}
    \label{statistical error bound conclusion1}
    \|\wh{\vrho} - \vrho^\star\|_F\lesssim  \frac{ \sqrt{n \log n}d \ol{r}(\log Q + n\log d)}{\sqrt{MQ}}.
\end{eqnarray}
This completes the proof of \Cref{Statistical Error_Haar_Measurement}.
\end{proof}

\subsection{Proof of \eqref{tail function of of entire variable_haar original}}
\label{Proof of concentration of multinomial dis}

\begin{proof}
Conditioning on $\{\vphi_{i,k},\forall i,k \}$, we use \Cref{General bound of multinomial distribution Q cases} by setting $a_{i,k} = \vphi_{i,k}^\dagger \widetilde\vrho \vphi_{i,k}$ to get
\begin{align}
    \label{Concentration inequality of sum of multinomial}
    &\P{\sum_{i=1}^Q\sum_{k=1}^{d^n} \eta_{i,k}\vphi_{i,k}^\dagger \widetilde\vrho \vphi_{i,k} \geq t \bigg| \{\vphi_{i,k},\forall i,k \}}\nonumber\\
    &\leq  e^{-\frac{Mt}{4\max_{i,k}|\vphi_{i,k}^\dagger \widetilde\vrho \vphi_{i,k}|}\min\bigg\{1, \frac{\max_{i,k}|\vphi_{i,k}^\dagger \widetilde\vrho \vphi_{i,k}| t}{4\sum_{i = 1}^{Q}\sum_{k = 1}^{d^n}|\vphi_{i,k}^\dagger \widetilde\vrho \vphi_{i,k}|^2p_{i,k} } \bigg\}  }\nonumber\\
    & \hspace{0.5cm}+  e^{-\frac{Mt^2}{8\sum_{i = 1}^{Q}\sum_{k = 1}^{d^n} |\vphi_{i,k}^\dagger \widetilde\vrho \vphi_{i,k}|^2p_{i,k} }}\nonumber\\
    &=  e^{-\frac{Mt^2}{16\sum_{i = 1}^{Q}\sum_{k = 1}^{d^n} |\vphi_{i,k}^\dagger \widetilde\vrho \vphi_{i,k}|^2p_{i,k} }} \! +\!  e^{-\frac{Mt^2}{8\sum_{i = 1}^{Q}\sum_{k = 1}^{d^n} |\vphi_{i,k}^\dagger \widetilde\vrho \vphi_{i,k}|^2p_{i,k} }},
\end{align}
where without loss of generality, we assume that $\frac{\max_{i,k}|\vphi_{i,k}^\dagger \widetilde\vrho \vphi_{i,k}| t}{4\sum_{i = 1}^{Q}\sum_{k = 1}^{d^n}|\vphi_{i,k}^\dagger \widetilde\vrho \vphi_{i,k}|^2p_{i,k} }\leq 1$ in the last line.

Notice that for any $i$ and $k$, $\vphi_{i,k}^\dagger\widetilde\vrho\vphi_{i,k}$ is a subexponential random variable with subexponential norm $\|\vphi_{i,k}^\dagger\widetilde\vrho\vphi_{i,k}\|_{\psi_1}  = O(\frac{1}{d^n})$ according to \eqref{eq:subexp-phiRho}. By using the concentration equality for the tail of a subexponential random variable \cite[Proposition 3]{chen2022error}, which states that $\P{|X|\geq t}\leq 2e^{-\frac{c_0t}{\|X\|_{\psi_1}}}$ for any subexponential random variable $X$ with a universal constant $c_0$,
we have
\begin{eqnarray}
    \label{Concentration inequality of single Sub-exponential random variable}
     \P{ |\vphi_{i,k}^\dagger\wt\vrho\vphi_{i,k}| \geq t}\leq e^{1-c_1d^n t},
\end{eqnarray}
where $c_1$ is a universal constant. It follows that
\begin{align}
    \label{Concentration inequality of single Sub-exponential random variable maximum}
     \P{ \max_{i,k}|\vphi_{i,k}^\dagger\wt\vrho\vphi_{i,k}| \leq t}&\geq 1 - Qd^ne^{1-c_1d^n t}\nonumber\\
     & =  1 - e^{1-c_1d^n t + \log(Qd^n)}.
\end{align}
Thus, setting $t= \frac{c_2\log(Qd^n)}{d^{\nqbit }}$,  we obtain
\begin{eqnarray}
    \label{Upper bound of RIP_one random variable1}
    \max_{i,
     k} |\vphi_{i,k}^\dagger\wt\vrho\vphi_{i,k}| \!\leq\! \frac{c_2\log(Qd^n)}{d^{\nqbit }} \!=\! \frac{c_2 (\log Q + n\log d)}{d^{\nqbit }},
\end{eqnarray}
with the probability at least $1-e^{-c_3 \log(Qd^n) }$, where $c_2$ and $c_3$ are positive constants.
Now under the event $F= \{ \max_{i,k} |\vphi_{i,k}^\dagger\wt\vrho\vphi_{i,k}| \lesssim \frac{\log Q + n\log d}{d^{\nqbit }} \}$, we have
\begin{eqnarray}
    \label{upper bound of maximum for inner product}
    \sum_{i = 1}^{Q}\sum_{k = 1}^{d^n} |\vphi_{i,k}^\dagger \widetilde\vrho \vphi_{i,k}|^2p_{i,k}&\!\!\!\!\leq\!\!\!\!& \max_{i,k} |\vphi_{i,k}^\dagger \widetilde\vrho \vphi_{i,k}|^2 \sum_{i=1}^Q\sum_{k = 1}^{d^n}p_{i,k}\nonumber\\
    &\!\!\!\!\leq\!\!\!\!& \frac{c_2^2Q(\log Q + n\log d)^2}{d^{2\nqbit }},
\end{eqnarray}
and thus
\begin{align}
    \label{tail function of of entire variable_haar another1}
    \P{ \sum_{i=1}^Q\sum_{k=1}^{d^n} \eta_{i,k}\vphi_{i,k}^\dagger \widetilde\vrho \vphi_{i,k}\geq t \bigg| F  } \leq 2e^{-\frac{d^{2n}Mt^2}{16c_2^2Q(\log Q + n\log d)^2}}.\nonumber\\
\end{align}
\end{proof}

\section{Auxiliary Materials}
\label{sec:app-aux}

\begin{lemma}
\label{Foucart13MathIntro_lemma716}
(\cite[Lemma 7.16]{Foucart13MathIntro} Paley-Zygmund inequality)
If a nonnegative random variable $Z$ has finite second moment, then we have
\begin{eqnarray}
    \label{Foucart13MathIntro_Lemma716}
    \mathbb{P}(Z> t)\geq \frac{(\E Z-t)^2}{\E Z^2}, 0\leq t \leq \E Z.
\end{eqnarray}
\end{lemma}

\begin{lemma}(\cite[Lemma 3.7]{ledoux1991probability}) Let $d$ be an integer and let $Z$ be a positive random variable. Then the following are equivalent:
\begin{itemize}
    \item there is a constant $C$ such that for any $p\ge 2$,
    \begin{eqnarray}
    \parans{ \E[Z^p] }^{1/p} \le C p^{d/2} \parans{ \E[Z^2] }^{1/2};
    \end{eqnarray}
    \item for some $\alpha>0$,
    \begin{eqnarray}
    \E \exp(\alpha Z^{2/d}) <\infty.
    \end{eqnarray}
\end{itemize}
\label{lem:hypercontractivity}\end{lemma}

\begin{lemma}(\cite[Theorem 2.8.2]{vershynin2018high})
\label{Sub-exponentialInequality}
Let $X_1,\ldots,X_N$ be independent, mean zero, subexponential random variables, and $\va = (a_1,\ldots,a_N)\in\R^N$. Then, for every $t\ge 0$, we have
\begin{align}
\P{\abs{\sum_{i=1}^N a_i X_i  } \ge t \!} \!\le\! 2\exp\bracket{\! -c \min\bracket{\! \frac{t^2}{K^2 \|\va\|_2^2}, \frac{t}{K \|\va\|_\infty} \! } \! }.
\end{align}
where $K = \max_i \|X_i\|_{\psi_1}=\sup_{q\geq 1}\E(|X_i|^q)^{1/q}/q$ and $c$ is a positive constant.
\end{lemma}

\begin{lemma}
\label{Stastitic property of UV}
(\cite[Lemma 2.2]{petz2004asymptotics})
A Haar-distributed random unitary matrix $\mU\in\C^{D\times D}$ can be equivalently generated by applying the Gram-Schmidt orthogonalization procedure to $D$ independent random vectors $\vz_i\in\C^{D},i = 1, 2, ..., D$, where the entries $z_{i,j}$ are mutually independent standard complex normal random variables. For all unitary matrices $\mV\in\C^{D\times D}$, the distributions of $\mU$ and $\mV\mU$ are the same.
\end{lemma}

\begin{lemma}
\label{Unitary_matrix_Expectation}
(\cite[Corollary 1.2]{novak2006truncations})
Let $u_{ij}$ be an element of $n\times n$ Haar-distributed random unitary matrix $\mU$. We have
\begin{eqnarray}
\label{Unitary_matrix_Expectation_1}
\E[| u_{ij}|^{2d} ]=\frac{d!}{n(n+1)\cdots (n+d-1)}.
\end{eqnarray}
\end{lemma}

\begin{lemma}
\label{EXPANSION_A1TOAN-B1TOBN_1}
For any  ${\bm A}_i,{\bm A}^\star_i\in\R^{r_{i-1}\times r_i},i=1,\dots,N$, we have
\begin{eqnarray}
    \label{EXPANSION_A1TOAN-B1TOBN_2}
    &\!\!\!\!\!\!\!\!&{\bm A}_1{\bm A}_2\cdots {\bm A}_N-{\bm A}_1^\star{\bm A}_2^\star\cdots {\bm A}_N^\star\nonumber\\
    &\!\!\!\!=\!\!\!\!& \sum_{i=1}^N \mA_1^\star \cdots \mA_{i-1}^\star (\mA_{i} - \mA_i^\star) \mA_{i+1} \cdots \mA_N.
\end{eqnarray}

\end{lemma}

\begin{proof}
We expand ${\bm A}_1{\bm A}_2\cdots {\bm A}_N-{\bm A}_1^\star{\bm A}_2^\star\cdots {\bm A}_N^\star$ as
\begin{eqnarray}
    \label{EXPANSION_A1TOAN-B1TOBN_3}
    &&\!\!\!\!{\bm A}_1{\bm A}_2\cdots {\bm A}_N-{\bm A}_1^\star{\bm A}_2^\star\cdots {\bm A}_N^\star\nonumber\\
    &\!\!\!\!=\!\!\!\!&{\bm A}_1{\bm A}_2\cdots {\bm A}_N-{\bm A}_1^\star{\bm A}_2{\bm A}_3\cdots {\bm A}_N \nonumber\\
    &\!\!\!\!\!\!\!\!&+{\bm A}_1^\star{\bm A}_2{\bm A}_3\cdots {\bm A}_N-{\bm A}_1^\star{\bm A}_2^\star\cdots {\bm A}_N^\star\nonumber\\
    &\!\!\!\!=\!\!\!\!&({\bm A}_1-{\bm A}_1^\star){\bm A}_2\cdots {\bm A}_N+{\bm A}_1^\star{\bm A}_2{\bm A}_3\cdots {\bm A}_N\nonumber\\
    &\!\!\!\!\!\!\!\!& -{\bm A}_1^\star{\bm A}_2^\star{\bm A}_3\cdots {\bm A}_N+{\bm A}_1^\star{\bm A}_2^\star{\bm A}_3\cdots {\bm A}_N - {\bm A}_1^\star{\bm A}_2^\star\cdots {\bm A}_N^\star\nonumber\\
    &\!\!\!\!=\!\!\!\!&\cdots = \sum_{i=1}^N \mA_1^\star \cdots \mA_{i-1}^\star (\mA_{i} - \mA_i^\star) \mA_{i+1} \cdots \mA_N.
\end{eqnarray}
\end{proof}

\begin{lemma}
\cite[Theorem 3.1]{tanoue2022concentration}
\label{Concentration inequality of sum of dependent sub-exp sum}
Suppose that $X = \sum_{i=1}^Q \sum_{k=1}^K w_k X_{i,k}$, where $w_k,k=1,\dots, K$ are constants, and each $X_{i,k}, i=1,\dots,Q, k=1,\dots,K$ is a zero-mean, subexponential random variable with $\|X_{i,k} \|_{\psi_1}$. In addition, the $Q$ multivariate random variables $(X_{i,1},\ldots,X_{i,K}), i = 1,\ldots, Q$ are mutually independent. However, it is possible for the variables $X_{i,k}$ and $X_{i,k'}, k'\neq k$ within each multivariate random variable to be dependent.
Then
\begin{eqnarray}
    \label{Concentration inequality of sum of dependent sub-exp}
    \P{X>t}\leq \begin{cases}
    e^{-\frac{t^2}{4T^2}}, & t\leq 2T^2H,\\
    e^{-\frac{tH}{2}}, & t>2T^2H.
  \end{cases}
\end{eqnarray}
where $T = \sum_{k=1}^Kw_k \sqrt{\sum_{i=1}^Q c_{i,k} \|X_{i,k} \|_{\psi_1}^2 }$ and $H = \bigg(\min_{k}\frac{ \sqrt{\sum_{i=1}^Q c_{i,k} \|X_{i,k} \|_{\psi_1}^2 }}{ \sum_{k=1}^Kw_k \sqrt{\sum_{i=1}^Q c_{i,k} \|X_{i,k} \|_{\psi_1}^2 }}\bigg)\cdot\bigg(\min_{i} \{\frac{d_{i,k}}{\|X_{i,k} \|_{\psi_1}} \} \bigg)$ with constants $c_{i,k}$ and $d_{i,k}$.
\end{lemma}

Below, we extend the concentration bounds presented in \cite[Lemmas 2\&3]{kawaguchi2022robustness} for a single multinomial random variable to encompass multiple multinomial random variables.
\begin{lemma}
\label{lower bound of multinomial distribution with positive constant conclusion Q cases}
Suppose that the $Q$ multivariate random variables $(f_{i,k},\dots, f_{i,K}),i=1,\dots,Q$ are mutually independent and follow the multinomial distribution $\operatorname{Multinomial}(M,\vp_i)$  with $\sum_{k=1}^{K}f_{i,k} =M $ and $\vp_i = [p_{i,1}.\dots, p_{i,K}]$, respectively.
Let $a_{i,1},\dots, a_{i,K}\geq 0$ be fixed such that $\sum_{k=1}^Ka_{i,k} p_{i,k}\neq 0, i=1,\dots,Q$. Then, for any $t>0$,
\begin{eqnarray}
    \label{lower bound of multinomial distribution with positive constant Q cases}
    &\!\!\!\!\!\!\!\!&\P{\sum_{i=1}^Q\sum_{k=1}^Ka_{i,k}\big(\frac{f_{i,k}}{M} - p_{i,k} \big) > t   }\nonumber\\
    &\!\!\!\!\leq\!\!\!\!& e^{-\frac{Mt}{2a_{\max}}\min\bigg\{1, \frac{a_{\max}t }{2\sum_{i=1}^Q\sum_{k=1}^K a_{i,k}^2p_{i,k}} \bigg\}},
\end{eqnarray}
where $a_{\max} = \max_{i,k}a_{i,k}$.

\end{lemma}

\begin{proof}
For any $v>0$, we have
\begin{eqnarray}
    \label{lower bound of multinomial distribution with positive constant Q cases derivation}
    &\!\!\!\!\!\!\!\!&\P{\sum_{i=1}^Q\sum_{k=1}^Ka_{i,k}\big(\frac{f_{i,k}}{M} - p_{i,k} \big) > t   }\nonumber\\
    &\!\!\!\!=\!\!\!\!& \P{v\sum_{i=1}^Q\sum_{k=1}^Ka_{i,k}\frac{f_{i,k}}{M}  > v\big(t+ \sum_{i=1}^Q\sum_{k=1}^Ka_{i,k} p_{i,k}\big)   }\nonumber\\
    &\!\!\!\!\leq\!\!\!\!& \P{e^{v\sum_{i=1}^Q\sum_{k=1}^Ka_{i,k}\frac{f_{i,k}}{M}}  \geq e^{v(t+ \sum_{i=1}^Q\sum_{k=1}^Ka_{i,k} p_{i,k})}   }\nonumber\\
    &\!\!\!\!\leq\!\!\!\!& e^{-v(t+ \sum_{i=1}^Q\sum_{k=1}^Ka_{i,k} p_{i,k})} \E{e^{v\sum_{i=1}^Q\sum_{k=1}^Ka_{i,k}\frac{f_{i,k}}{M}}}\nonumber\\
    &\!\!\!\!=\!\!\!\!& e^{-v(t+ \sum_{i=1}^Q\sum_{k=1}^Ka_{i,k} p_{i,k})} \prod_{i=1}^Q \E{e^{v\sum_{k=1}^Ka_{i,k}\frac{f_{i,k}}{M}}}\nonumber\\
    &\!\!\!\!\leq\!\!\!\!& e^{-v(t+ \sum_{i=1}^Q\sum_{k=1}^Ka_{i,k} p_{i,k})}\nonumber\\
    &\!\!\!\!\!\!\!\!& \cdot e^{v\sum_{i=1}^Q\sum_{k=1}^K a_{i,k}p_{i,k} + \sum_{i=1}^Q\sum_{k=1}^K p_{i,k}\frac{v^2a_{i,k}^2}{M}}\nonumber\\
    &\!\!\!\!=\!\!\!\!&e^{-vt+ \sum_{i=1}^Q\sum_{k=1}^K p_{i,k}\frac{v^2a_{i,k}^2}{M}}\nonumber\\
    &\!\!\!\!\leq\!\!\!\!& e^{-\frac{Mt}{2a_{\max}}\min\bigg\{1, \frac{a_{\max}t }{2\sum_{i=1}^Q\sum_{k=1}^K a_{i,k}^2p_{i,k}} \bigg\}},
\end{eqnarray}
where the second inequality uses Markov's inequality, the fourth line follows from the independence of multivariate random variables $(f_{i,k},\dots, f_{i,K}),i=1,\dots,Q$, the third inequality utilizes \cite[Lemma 2]{kawaguchi2022robustness} for $\E{e^{v\sum_{k=1}^Ka_{i,k}\frac{f_{i,k}}{M}}}$, and the last line follows by setting  $v = \frac{Mt}{2\sum_{i=1}^Q\sum_{k=1}^Ka_{i,k}^2 p_{i,k} }\leq \frac{M}{a_{\max}}$ when $t\leq \frac{2\sum_{i=1}^Q\sum_{k=1}^Ka_{i,k}^2 p_{i,k} }{a_{\max}}$ and $v = \frac{M}{a_{\max}}$ when $t\geq \frac{2\sum_{i=1}^Q\sum_{k=1}^Ka_{i,k}^2 p_{i,k} }{a_{\max}}$.
\end{proof}

\begin{lemma}
\label{lower bound of multinomial distribution with negative constant conclusion Q cases}
Suppose that the $Q$ multivariate random variables $(f_{i,k},\dots, f_{i,K}),i=1,\dots,Q$ are mutually independent and follow the multinomial distribution $\operatorname{Multinomial}(M,\vp_i)$  with $\sum_{k=1}^{K}f_{i,k} =M $ and $\vp_i = [p_{i,1}.\dots, p_{i,K}]$, respectively.
Let $a_{i,1},\dots, a_{i,K}\geq 0$ be fixed such that $\sum_{k=1}^Ka_{i,k} p_{i,k}\neq 0, i=1,\dots,Q$. Then, for any $t>0$,
\begin{align}
    \label{lower bound of multinomial distribution with negative constant Q cases}
    \P{\!-\sum_{i=1}^Q\sum_{k=1}^Ka_{i,k}(\frac{f_{i,k}}{M} - p_{i,k}) > t   \!} \! \leq \! e^{-\frac{Mt^2 }{2\sum_{i=1}^Q\sum_{k=1}^K a_{i,k}^2p_{i,k}}}.
\end{align}

\end{lemma}

\begin{proof}
Following the same approach for proving \Cref{lower bound of multinomial distribution with positive constant Q cases}, for any $v<0$, we have
\begin{eqnarray}
    \label{lower bound of multinomial distribution with negative constant Q cases derivation}
    &\!\!\!\!\!\!\!\!&\P{-\sum_{i=1}^Q\sum_{k=1}^Ka_{i,k}(\frac{f_{i,k}}{M} - p_{i,k}) > t   }\nonumber\\
    &\!\!\!\!=\!\!\!\!& \P{v\sum_{i=1}^Q\sum_{k=1}^Ka_{i,k}\frac{f_{i,k}}{M} >   v(\sum_{i=1}^Q\sum_{k=1}^Ka_{i,k}p_{i,k} -t)}\nonumber\\
    &\!\!\!\!\leq\!\!\!\!& \P{e^{v\sum_{i=1}^Q\sum_{k=1}^Ka_{i,k}\frac{f_{i,k}}{M}} \geq  e^{ v(\sum_{i=1}^Q\sum_{k=1}^Ka_{i,k}p_{i,k} -t)}}\nonumber\\
    &\!\!\!\!\leq\!\!\!\!&e^{ -v(\sum_{i=1}^Q\sum_{k=1}^Ka_{i,k}p_{i,k} -t)}\E{e^{v\sum_{i=1}^Q\sum_{k=1}^Ka_{i,k}\frac{f_{i,k}}{M}}}\nonumber\\
    &\!\!\!\!=\!\!\!\!& e^{ -v(\sum_{i=1}^Q\sum_{k=1}^Ka_{i,k}p_{i,k} -t)} \Pi_{i=1}^Q \E{e^{v\sum_{k=1}^Ka_{i,k}\frac{f_{i,k}}{M}}}\nonumber\\
    &\!\!\!\!\leq\!\!\!\!& e^{vt+\sum_{i=1}^Q\sum_{k=1}^K p_{i,k} \frac{a_{i,k}^2 v^2}{2M}}\nonumber\\
    &\!\!\!\!=\!\!\!\!&e^{-\frac{Mt^2 }{2\sum_{i=1}^Q\sum_{k=1}^K a_{i,k}^2p_{i,k}}},
\end{eqnarray}
where the derivations before the last line are the same as those for proving \Cref{lower bound of multinomial distribution with positive constant Q cases} and the last line follows by setting $v = -\frac{tM}{\sum_{i=1}^Q\sum_{k=1}^K a_{i,k}^2p_{i,k}}$.
\end{proof}

\Cref{lower bound of multinomial distribution with positive constant conclusion Q cases} and \Cref{lower bound of multinomial distribution with negative constant conclusion Q cases}, together leads to the following multinomial concentration bounds.
\begin{lemma}
\label{General bound of multinomial distribution Q cases}
Suppose that the $Q$ multivariate random variables $(f_{i,k},\dots, f_{i,K}),i=1,\dots,Q$ are mutually independent and follow the multinomial distribution $\operatorname{Multinomial}(M,\vp_i)$  with $\sum_{k=1}^{K}f_{i,k} =M $ and $\vp_i = [p_{i,1}.\dots, p_{i,K}]$, respectively.
Let $a_{i,1},\dots, a_{i,K}$ be fixed. Then, for any $t>0$,
\begin{align}
    \label{General bound of multinomial distribution for all constant Q cases}
    &\P{\sum_{i=1}^Q\sum_{k=1}^Ka_{i,k}(\frac{f_{i,k}}{M} - p_{i,k}) > t   }\nonumber\\
    &\leq  e^{-\frac{Mt}{4a_{\max}}\! \min\bigg\{\! 1, \frac{a_{\max}t }{4\sum_{i=1}^Q\sum_{k=1}^K a_{i,k}^2p_{i,k}} \! \bigg\}}\! + \! e^{-\frac{Mt^2 }{8\sum_{i=1}^Q\sum_{k=1}^K a_{i,k}^2p_{i,k}}},
\end{align}
where $a_{\max} = \max_{i,k}|a_{i,k}|$.
\end{lemma}
\begin{proof}

Since $\{a_{i,k}\}, i=1,\dots,Q, k=1,\dots,K$ could be positive or negative, we separate the set into three sets $P$, $N$ and $Z$ such that $a_{i,k}> 0$ for $\{i,k\}\in P$, $a_{i,k} < 0$ for $\{i,k\}\in N$, and $a_{i,k} = 0$ for $\{i,k\}\in Z$. In addition, when $p_{i,k} = 0$,  we have $f_{i,k} = 0$ and further obtain $a_{i,k}(\frac{f_{i,k}}{M} - p_{i,k}) = 0$. Thus, without loss of generality, we assume that $p_{i,k} > 0$ for all $i,k$. Now we have
    \begin{align}
    \label{lower bound of multinomial distribution with all constant1}
    &\P{\sum_{i=1}^Q\sum_{k=1}^Ka_{i,k}(\frac{f_{i,k}}{M} - p_{i,k}) > t   } \nonumber\\
    &\leq\P{\hspace{-0.1cm}\sum_{\{i,k\}\in P}\hspace{-0.2cm}\! a_{i,k}(\frac{f_{i,k}}{M}\! -\! p_{i,k}) >\! \frac{t}{2} \cup \hspace{-0.3cm} \sum_{\{i,k\}\in N}\hspace{-0.2cm}\! a_{i,k}(\frac{f_{i,k}}{M}\! -\! p_{i,k}) >\! \frac{t}{2}  \!}\nonumber\\
    &\leq \P{\sum_{\{i,k\}\in P}a_{i,k}(\frac{f_{i,k}}{M} - p_{i,k}) > \frac{t}{2}}\nonumber\\
    & + \P{\sum_{\{i,k\}\in N}a_{i,k}(\frac{f_{i,k}}{M} - p_{i,k}) > \frac{t}{2}  }\nonumber\\
    &=\P{\sum_{\{i,k\}\in P}\hspace{-0.2cm} a_{i,k}(\frac{f_{i,k}}{M} - p_{i,k}) +\hspace{-0.4cm} \sum_{\{i,k\}\in N\cup Z}\hspace{-0.4cm} 0\cdot (\frac{f_{i,k}}{M} - p_{i,k}) > \frac{t}{2}} \nonumber\\
    &\hspace{0.1cm}+ \P{\sum_{\{i,k\}\in N}\hspace{-0.2cm} a_{i,k}(\frac{f_{i,k}}{M} - p_{i,k}) + \hspace{-0.4cm}\sum_{\{i,k\}\in P\cup Z}\hspace{-0.4cm} 0\cdot (\frac{f_{i,k}}{M} - p_{i,k}) > \frac{t}{2}  }\nonumber\\
    &\leq  e^{-\frac{Mt}{4\tilde{a}_{\max}}\min\bigg\{1, \frac{\tilde{a}_{\max}t }{4\sum_{i=1}^Q\sum_{k = 1}^{K} \tilde{a}_{i,k}^2p_{i,k}} \bigg\}}\!\! + e^{-\frac{Mt^2 }{8\sum_{i=1}^Q\sum_{k = 1}^{K} \hat{a}_{i,k}^2p_{i,k}}}\nonumber\\
    &\leq  e^{-\frac{Mt}{4a_{\max}}\min\bigg\{1, \frac{a_{\max}t }{4\sum_{i=1}^Q\sum_{k=1}^K a_{i,k}^2p_{i,k}} \bigg\}} \!\!+ e^{-\frac{Mt^2 }{8\sum_{i=1}^Q\sum_{k=1}^K a_{i,k}^2p_{i,k}}},
\end{align}
where the first inequality follows from the fact that $ \sum_{i=1}^Q\sum_{k=1}^Ka_{i,k}( \frac{f_{i,k}}{M} - p_{i,k}) > t $ implies that either $ \sum_{\{i,k\}\in P}a_{i,k}(\frac{f_{i,k}}{M} - p_{i,k}) > \frac{t}{2}$ or $\sum_{\{i,k\}\in N}a_{i,k}(\frac{f_{i,k}}{M} - p_{i,k}) > \frac{t}{2}$ must hold, in the second inequality we define two sets with elements $\tilde{a}_{i,k} = \begin{cases} a_{i,k}, & \{i,k\}\in P \\ 0, & \{i,k\}\in N\cup Z \end{cases}$ and $\hat{a}_{i,k} = \begin{cases} a_{i,k}, & \{i,k\}\in N \\ 0, & \{i,k\}\in P\cup Z \end{cases}$, respectively, and  the last line uses $\tilde{a}_{\max} = \max_{i,k}|\tilde{a}_{i,k}| \leq a_{\max}$ and
\[\max\bigg\{ \sum_{i=1}^Q\sum_{k = 1}^{K} \tilde{a}_{i,k}^2p_{i,k},  \sum_{i=1}^Q\sum_{k = 1}^{K} \hat{a}_{i,k}^2p_{i,k} \bigg\}\leq \sum_{i=1}^Q\sum_{k=1}^K a_{i,k}^2p_{i,k}.\]
\end{proof}







%
%

%
%
%
%

\end{document}